\documentclass[10pt,journal,compsoc]{IEEEtran}
\ifCLASSOPTIONcompsoc
\usepackage[nocompress]{cite}
\else
\usepackage{cite}
\fi

\ifCLASSINFOpdf
\else
\fi

\usepackage{amsmath,amssymb,amsfonts}
\usepackage{amsthm}
\usepackage{algorithmic}
\usepackage{graphicx}
\usepackage{textcomp}
\usepackage{bm}
\usepackage{enumitem}
\usepackage{subfigure}
\usepackage{fix-cm}
\usepackage{color}
\usepackage[table,xcdraw]{xcolor}

\usepackage[ruled,linesnumbered]{algorithm2e}

\newtheorem{assumption}{Assumption}

\newtheorem{lemma}{Lemma}
\newtheorem{theorem}{Theorem}

\theoremstyle{definition}

\theoremstyle{remark}
\newtheorem*{remark}{Remark}

\allowdisplaybreaks[4]

\newcommand{\etal}{\textit{et al}.}
\newcommand{\ie}{\textit{i}.\textit{e}.}
\newcommand{\eg}{\textit{e}.\textit{g}.}
\newcommand{\mli}[1]{\textit{#1}}

\begin{document}

\title{TODG: Distributed Task Offloading with Delay Guarantees for Edge Computing}

\author{Sheng~Yue,~Ju~Ren, Member, IEEE,~Nan~Qiao,~Yongmin~Zhang, Member, IEEE,~Hongbo~Jiang, Senior Member, IEEE,~Yaoxue~Zhang, Senior Member, IEEE,~and~Yuanyuan Yang, Fellow, IEEE
	\IEEEcompsocitemizethanks{\IEEEcompsocthanksitem Sheng Yue, Nan Qiao and Yongmin Zhang are with the School of Computer and Engineering, Central South University, Changsha,
		Hunan, 410083 China. E-mails: \{sheng.yue, csu\_qiaonan, zhangyongmin\}@csu.edu.cn.
		\IEEEcompsocthanksitem Ju Ren and Yaoxue Zhang are  with the Department of Computer Science and Technology, BNRist, Tsinghua University, Beijing, 100084 China. E-mails: \{renju, zhangyx\}@tsinghua.edu.cn.
		\IEEEcompsocthanksitem Hongbo Jiang is with the School of Information Science and Engineering, Hunan University, Changsha, Hunan, 410006, China. E-mail: hongbojiang@hnu.edu.cn.
		\IEEEcompsocthanksitem Yuanyuan Yang is with the Department of Electrical and Computer Engineering, Stony Brook University, Stony Brook, NY 11794-2350, USA. E-mail: yuanyuan.yang@stonybrook.edu.
		\IEEEcompsocthanksitem Corresponding Author: Ju Ren.
	}
}

\IEEEtitleabstractindextext{%
	\begin{abstract}
		Edge computing has been an efficient way to provide prompt and near-data computing services for resource-and-delay sensitive IoT applications via computation offloading. Effective computation offloading strategies need to comprehensively cope with several major issues, including 1) the allocation of dynamic communication and computational resources, 2) delay constraints of heterogeneous tasks, and 3) requirements for computationally inexpensive and distributed algorithms. However, most of the existing works mainly focus on part of these issues, which would not suffice to achieve expected performance in complex and practical scenarios. To tackle this challenge, in this paper, we systematically study a distributed computation offloading problem with delay constraints, where heterogeneous computational tasks require continually offloading to a set of edge servers via a limiting number of stochastic communication channels. The task offloading problem is formulated as a delay-constrained long-term stochastic optimization problem under unknown prior statistical knowledge. To solve this problem, we first provide a technical path to transform and decompose it into several slot-level sub-problems. Then, we devise a distributed online algorithm, namely TODG, to efficiently allocate resources and schedule offloading tasks. Further, we present a comprehensive analysis for TODG in terms of the optimality gap, the worst-case delay, and the impact of system parameters. Extensive simulation results demonstrate the effectiveness and efficiency of TODG.
	\end{abstract}
	
	}

\maketitle

\IEEEdisplaynontitleabstractindextext

\IEEEpeerreviewmaketitle

\section{Introduction}

\IEEEPARstart{D}{ue} to the rapid development of wireless communications, mobile devices have become the information hub and accessing point to connect physical and cyber worlds. A large number of modern applications, such as activity recognition, interactive gaming, natural language processing, have been developed for mobile devices to provide intelligent and convenient services. However, these applications are usually computation-and-energy intensive and delay-sensitive. They are hardly executed on resource-constrained mobile devices and pose significant challenges in offloading them to the cloud with delay guarantees. To tackle these challenges, edge computing has been proposed as a promising solution to alleviate the computing burdens of mobile devices and reduce service delay~\cite{ren2019survey}. It can leverage the computing capabilities of devices/infrastructures in the proximity of data sources to provide pervasive, prompt, and agile services via computation offloading at anytime and anywhere \cite{abbas2017mobile}. 

Nevertheless, the design of efficient computation offloading strategies in edge computing is a non-trivial task. Different from conventional cloud computing, where a device only needs to decide whether to offload its tasks to a cloud center, in the case of edge computing, the increase of user devices may complicate the offloading decisions, caused by the contention for the insufficient computational resources of each edge server \cite{jovsilo2018decentralized}. If the computational resources of edge servers are not well coordinated for user devices, the performance would seriously degrade due to the overwhelmed offloading tasks. On the other hand, in edge computing, computation offloading must involve wireless communications between user devices and edge servers. The inherently limited and stochastic radio resources call for an effective radio resource allocation strategy; otherwise, the wireless network capacity may be quickly strained, causing low transmission efficiency and dissatisfaction with edge computing services \cite{yi2019multi}. Moreover, the delay guarantees of offloading tasks are essential for many applications, such as interactive gaming, object recognition, and rendering in smart driving \cite{hekmati2019optimal}. However, the stochasticity of communication channels and computing power of edge servers make the delay control extremely difficult, especially for heterogeneous tasks with different delay requirements. In addition, in contrast to the centralized and powerful cloud server, edge servers are deployed in a distributed manner, each of which is often resource-limited and heterogeneous. Thus, it is of critical importance to develop distributed and computationally efficient algorithms for task offloading in the context of edge computing.

Recently we have witnessed significant progress in developing novel approaches to address the challenges in task offloading. In particular, there have been several works on various aspects, including designing energy-efficient offloading strategies \cite{chen2016joint,kamoun2015joint,labidi2015joint,eshraghi2019joint}, jointly allocating communication and computation resources for performance improvement \cite{zhang2019near,xu2018offloading,zhu2019blot,wang2020multi,liu2020resource}, lowering the response latency \cite{mao2016dynamic,hekmati2019optimal,wang2020online}, and developing decentralized offloading methods \cite{chen2018computation, jovsilo2020computation, tang2020deep,chen2019collaborative,jovsilo2018decentralized,liu2020latency,you2016energy}. However, most of these existing works mainly aim at tackling part of the aforementioned issues by weakening other restrictions. Therefore, we argue that the strategy, comprehensively taking the above issues into account, is a requisite for achieving effectual computation offloading in edge computing.

To bridge the gap, this paper systematically study a distributed task offloading problem with delay constraints in edge computing, where heterogeneous computational tasks (with different sizes, required resources, and response times) require continually offloading to a set of edge servers with different computing capabilities via a limiting number of random channels. Accordingly, we formulate the offloading problem as a delay-constrained long-term stochastic optimization problem under unknown prior statistical distributions. Clearly, it is quite tough to solve this stochastic optimization problem because of the inherent complexity of continually scheduling a large number of heterogeneous tasks and jointly allocating the communication and computational resources. To address this challenge, we first provide an approach to transform and decompose the original problem into three sub-problems. Then we develop an online algorithm, called TODG, solving these sub-problems in a distributed manner. In particular, by the ``$\delta$-periodic strategy'', TODG only needs to allocate channels every $\delta$ time slots, which can alleviate the computational cost during the system operation. We also provide a comprehensive performance analysis of TODG. It is demonstrated that TODG can achieve a trade-off between the near-optimal solutions and the computational cost. Besides, we rigorously show that TODG can well satisfy the delay constraints and quantify the impact of the delay requirements and the task buffer sizes on the system utility.

Our main contributions can be summarized as follows.
\begin{itemize}
	\item To the best of our knowledge, we are the first to systematically consider a distributed task offloading and resource allocation strategy for heterogeneous computational tasks with delay guarantees. We formulate the offloading problem as a delay-constrained long-term stochastic optimization problem under unknown prior statistical knowledge about the random task arrivals and the channel states as well as the computing power on edge servers.
	\item We devise an online algorithm to solve the long-term stochastic optimization problem, namely TODG, which can be implemented in parallel among user devices and edge servers, and provide worst-case delay guarantees for all offloading tasks. In particular, we develop a \emph{$\delta$-periodic strategy}, enabling to carry out channel assignment every $\delta$ slots, which largely mitigates the computational cost and communication overhead induced by the complex computation in resource allocation.
	\item We present a comprehensive analysis of the proposed algorithm. We characterize the optimality gap and the response latency, and quantify the impact of system parameters on the performance in terms of the buffer sizes, delay requirements, and the period of the $\delta$-periodic strategy. 
	Further, we provide extensive simulation results to showcase the efficacy of TODG.
\end{itemize}

The remainder of this paper is organized as follows. Section~\ref{sec:related_work} briefly reviews the related work, and Section~\ref{sec:related_work}
introduces the system model and formulates the distributed task offload problem. We present the details of the proposed TODG algorithm in Section~\ref{sec:todg} and analyze the theoretical performance of TODG in Section~\ref{sec:analysis}. Finally, Section~\ref{sec:simulations} shows the performance evaluation results, followed by a conclusion drawn in  Section~\ref{sec:conclusion}.

\section{Related Work}
\label{sec:related_work}

\begin{table*}[ht]
	\caption{Comparison with related works.}
	\label{table:comparison}
	\vspace{-0.5em}
	\centering
	\renewcommand\arraystretch{1.2}
	\resizebox{\textwidth}{!}{
		\begin{tabular}{cccccccc}
			\hline
			\textbf{Paper} & \multicolumn{1}{c}{\textbf{\begin{tabular}[c]{@{}c@{}}Radio\\ management\end{tabular}}} & \multicolumn{1}{c}{\textbf{\begin{tabular}[c]{@{}c@{}}Load\\ balancing\end{tabular}}} & \multicolumn{1}{c}{\textbf{\begin{tabular}[c]{@{}c@{}}Delay\\ constraints\end{tabular}}} & \multicolumn{1}{c}{\textbf{\begin{tabular}[c]{@{}c@{}}Multiple\\ users\end{tabular}}} & \multicolumn{1}{c}{\textbf{\begin{tabular}[c]{@{}c@{}}Multiple\\ servers\end{tabular}}} & \multicolumn{1}{c}{\textbf{\begin{tabular}[c]{@{}c@{}}Distributed\\ algorithm\end{tabular}}} & \multicolumn{1}{c}{\textbf{\begin{tabular}[c]{@{}c@{}}Online\\ offloading process\end{tabular}}} \\ \hline
			Kao and Krishnamachar \cite{kao2014optimizing}          & No                                            & No                                          & Yes                                            & No                                          & No                                            & No                                                 & No                                                     \\
			Mao \etal~\cite{mao2016dynamic}          & No                                            & No                                          & Yes                                            & No                                          & No                                            & No                                                 & Yes                                                    \\
			Liu \etal~\cite{liu2016delay}          & No                                            & No                                          & Yes                                            & No                                          & No                                            & No                                                 & Yes                                                    \\
			Lyu \etal~\cite{lyu2017optimal}          & Yes                                           & No                                          & No                                             & Yes                                         & No                                            & No                                                 & Yes                                                    \\
			Mao \etal~\cite{mao2017stochastic}          & No                                            & No                                          & No                                             & Yes                                         & No                                            & No                                                 & Yes                                                    \\
			Mao \etal~\cite{mao2017joint}          & No                                            & No                                          & Yes                                            & No                                          & No                                            & No                                                 & No                                                     \\
			Zhang \etal~\cite{zhang2017optimal}          & No                                            & Yes                                         & Yes                                            & Yes                                         & Yes                                           & Yes                                                & No                                                     \\
			Chen \etal~\cite{chen2017joint}              & No                                                                  & Yes                                                               & Yes                                                                  & Yes                                                               & Yes                                                                 & No                                                                       & No                                                                           \\
			You \etal~\cite{you2016energy}             & Yes                                                                 & No                                                                & Yes                                                                  & Yes                                                               & No                                                                  & No                                                                       & No                                                                           \\
			Ren \etal~\cite{ren2018latency}             & Yes                                                                 & No                                                                & Yes                                                                  & Yes                                                               & No                                                                  & No                                                                       & No                                                                           \\
			Zhang \etal~\cite{zhang2017energy}          & Yes                                           & No                                          & Yes                                            & Yes                                         & No                                            & No                                                 & No                                                     \\
			Zhou \etal~\cite{zhou2018computation}          & No                                            & No                                          & Yes                                            & Yes                                         & No                                            & No                                                 & No                                                     \\
			Alameddine \etal~\cite{alameddine2019dynamic}          & No                                            & Yes                                         & Yes                                            & Yes                                         & Yes                                           & No                                                 & No                                                     \\
			Chen \etal~\cite{chen2019delay}          & No                                            & No                                          & Yes                                            & Yes                                         & No                                            & No                                                 & No                                                     \\
			Maswood \etal~\cite{maswood2020novel}          & No                                            & Yes                                         & No                                             & Yes                                         & Yes                                           & No                                                 & No                                                     \\
			Jo{\v{s}}ilo and D{\'a}n \etal~\cite{jovsilo2018decentralized}          & No                                            & Yes                                         & Yes                                            & Yes                                         & No                                            & Yes                                                & Yes                                                    \\
			Liang \etal~\cite{liang2019multiuser}          & No                                            & No                                          & Yes                                            & Yes                                         & No                                            & No                                                 & No                                                     \\
			Liu \etal~\cite{liu2020latency}             & No                                                                  & No                                                                & Yes                                                                  & Yes                                                               & No                                                                  & Yes                                                                      & Yes                                                                          \\
			Li \etal~\cite{li2020qos}             & Yes                                                                 & Yes                                                               & Yes                                                                  & Yes                                                               & Yes                                                                 & Yes                                                                      & No                                                                           \\
			Nath \etal~\cite{nath2020multi}             & Yes                                                                 & No                                                                & Yes                                                                  & Yes                                                               & No                                                                  & No                                                                       & No                                                                           \\
			Hekmati \etal~\cite{hekmati2019optimal}             & No                                                                  & No                                                                & Yes                                                                  & No                                                                & No                                                                  & No                                                                       & No                                                                           \\
			Li \etal~\cite{li2021task}          & No                                            & No                                          & Yes                                            & Yes                                         & No                                            & No                                                 & No                                                     \\
			\rowcolor[HTML]{96FFFB} 
			This paper     & Yes                                           & Yes                                         & Yes                                            & Yes                                         & Yes                                           & Yes                                                & Yes                                                    \\ \hline
		\end{tabular}
	}\vspace{-1.0em}
\end{table*}

As a key enabling technology, task offloading has attracted increasing research attention in edge computing recently~\cite{yi2019multi}. Some early studies focus on making offloading decisions whether a mobile device should offload the task to an edge server or not~\cite{chen2014decentralized,dinh2017offloading}. For example, Chen \etal~\cite{chen2014decentralized} design a decentralized offloading game to make the offloading decision for minimizing the energy overhead. Dinh \etal~\cite{dinh2017offloading} propose a computation offloading approach to determine the offloaded tasks and CPU frequency of a mobile device to minimize task execution time and energy consumption.

Some recent works study joint communication and computation resource allocation to improve the performance of task offloading from a system perspective~\cite{zhang2019near,xu2018offloading,zhu2019blot,wang2020multi,liu2020resource,ren2020joint,wang2017computation,fang2016stochastic,hekmati2019optimal,wang2020online,li2020qos}. Specifically, Ren \etal~\cite{ren2020joint} propose a channel allocation and resource management approach to optimize offloading decisions and maximize the long-term network utility. Wang \etal~\cite{wang2017computation} present a system-level resource management approach, including offloading decisions, channel allocation, and caching strategy, to maximize the network utility. However, all these works ignore the latency constraint in task offloading, which is significantly important for delay-sensitive applications and attracts increasing research attention \cite{fang2016stochastic}. To control the offloading latency, Kao \etal~\cite{kao2014optimizing} propose a task partitioning method for one-task offloading, giving a near-optimal solution. Mao \etal~\cite{mao2017joint} and Chen \etal~\cite{chen2019delay} study the multitask offloading problems to optimize the execution delay and energy overhead. You \etal~\cite{you2016energy} design a threshold-based policy to manage offloading data volumes and channel access opportunities in a TDMA-based edge computing system. Alameddine \etal~\cite{alameddine2019dynamic} design a joint computing resource allocation and task offloading approach, considering the heterogeneity in the requirements of the offloaded tasks. Leveraging SDN, Maswood \etal~\cite{maswood2020novel} present a cooperative three-layer fog-cloud computing framework to optimize the bandwidth cost and load balancing. Hekmati \etal~\cite{hekmati2019optimal} develops an energy-optimal task offloading algorithm, called OnOpt, which considers the stochastic wireless channels and exploits the Markov chain to obtain the optimal offloading decisions with hard deadline constraints. Wang \etal~\cite{wang2020online} focus on the task offloading problem in non-orthogonal multiple access (NOMA) based edge computing systems and propose an online-learning algorithm to determine the optimal task and subcarrier allocation decisions for minimizing the task execution delay. 

Decentralized task offloading is another research trend in the study field of edge computing~\cite{chen2018computation, jovsilo2020computation, tang2020deep,chen2019collaborative,jovsilo2018decentralized,liu2020latency,you2016energy}. Jošilo \etal~\cite{jovsilo2018decentralized} propose an efficient decentralized algorithm for computing an equilibrium of task offloading game based on the variational inequality theory. Liu \etal~\cite{liu2020latency} design a decentralized offloading algorithm based on Lyapunov optimization and primal-dual theory, which decomposes the original complex problem into a set of sub-problems that can be solved on a mobile device or an edge server separately. However, they only investigate the one-server system and aim at controlling the average delay instead of strict delay constraints. Based on the Stackelberg game theoretic approach, Zhang \etal~\cite{zhang2017optimal} develop an iterative distributed offloading strategy for the hierarchical Vehicular Edge Computing system.  Nevertheless, most of the existing distributed solutions are based on a strong assumption of fixed task volume or sufficient communication resources. In addition, by introducing a statistical computation and transmission model, Li \etal~\cite{li2020qos} propose a distributed task offloading algorithm to provide statistical delay guarantees with full consideration of stochastic communication resources. Different from one-round scheduling and statistical delay guarantees in \cite{li2020qos}, this paper focus on online offloading process with hard delay constraints.

We summarize the difference between this paper and the existing works in Table \ref{table:comparison}. Notably, this paper focuses on a more complex and practical scenario, where heterogeneous computational tasks with random arrivals require scheduling to different edge servers with varying delay constraints, and designs a distributed algorithm to optimize the long-term utility of the whole system.

\section{System model}
\label{sec:system_model}

\begin{figure}[ht]
	\centering
	\includegraphics[width=0.95\columnwidth]{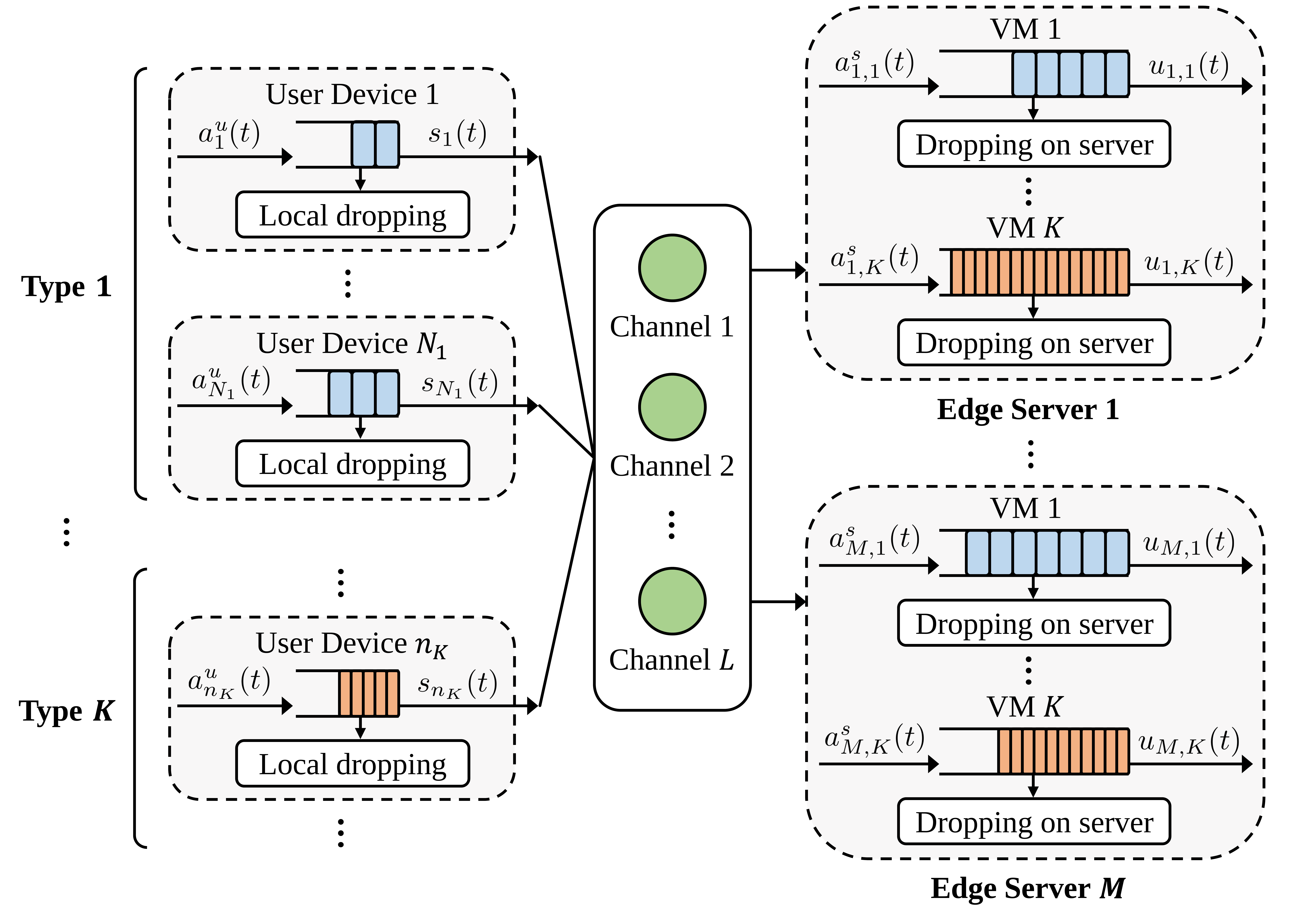} 
	\caption{Illustration of delay-constrained computation offloading for heterogeneous tasks with multiple channels and servers in edge computing.}
	\label{figure:system_model}
\end{figure} 	

As illustrated in Fig. \ref{figure:system_model}, we consider an edge computing system that operates in slotted time indexed by $t\in\{0,1,\dots\}$, where a set of heterogeneous user devices (denoted by $\mathcal{N}$) can offload their computational tasks to the edge servers (denoted by $\mathcal{M}$) via a limited number of communication channels (denoted by $\mathcal{L}$) for faster response or lower energy consumption. In each time slot, the system allocates a certain number of channels and edge servers to a part of user devices, enabling them to offload the computational tasks from their local task buffers. Generally, the user devices can be classified into different types (denoted by $\mathcal{K}$). The devices belonging to the same type will generate the same kind of tasks, such as image processing, video processing, etc. Accordingly, each edge server $m\in\mathcal{M}$ creates $K$ virtual machines (VMs) corresponding to different types of tasks. Such that, the $k$-type tasks would be processed by the corresponding VM in the server. It is worth noting that in this paper, we do not assume any prior knowledge on the statistical distributions of all the stochastic variables (\ie, the task arrivals, the channel states, and the service rates on servers). In the following, we use superscript `u' and `s' to specify the variables on the user and server sides, respectively, and denote $N=\vert\mathcal{N}\vert$, $K=\vert \mathcal{K}\vert$, and $L=\vert \mathcal{L}\vert$ for simplicity. The key notations are summarized in Table \ref{table:notations}.

\begin{table}[t]
	\caption{Key Notations}
	\label{table:notations}
	\vspace{-0.5em}
	\centering
	\renewcommand\arraystretch{1.25}
	\resizebox{\columnwidth}{!}{
		\begin{tabular}{l|l}
			\hline
			\textbf{Notation}                          & \textbf{Definition}                                                                                                     \\ \hline
			$k$, $l$, $m$                              & indexes of task types, channels, and edge servers                                                                       \\
			$t$, $n_k$                                 & indexes of time slots and user devices of $k$-th type                                                                   \\
			$v_{m,k}$                                  & index of VMs in server $m$ for executing $k$-th type tasks                                                               \\
			$z_{n_k,l,m}(t)$                           & \begin{tabular}[c]{@{}l@{}}binary variable indicating if device $n_k$ offloads tasks to \\ server $m$ via channel $l$ in $t$ or not (1 or 0)\end{tabular} \\
			$s_{n_k}(t)$                               & amount of tasks offloaded by device $n_k$ in $t$                                                                                  \\
			$\xi_{n_k}(t)$                             & upper bound of $s_{n_k}(t)$ in $t$                                                                                      \\
			$c_{n_k,l,m}(t)$                           & channel capacity of $l$ between device $n_k$ and server $m$ in $t$          \\
			$\tau^\mli{u}_{n_k,t}$, $\tau^\mli{s}_{n_k,t}$ & queuing delays of task $a^\mli{u}_{n_k}(t)$ in local device and server                                                           \\
			$a^\mli{u}_{n_k}(t)$, $d^\mli{u}_{n_k}(t)$ & amount of tasks arriving and dropped in device $n_k$ in $t$                                                                       \\
			$a^\mli{s}_{m,k}(t)$, $d^\mli{s}_{m,k}(t)$ & amount of tasks arriving and dropped in VM $v_{m,k}$ in $t$                                                                       \\
			$Q^\mli{u}_{n_k}(t)$, $Q^\mli{s}_{m,k}(t)$    & task queue backlogs in device $n_k$ and VM $v_{m,k}$                                                                            \\
			$Z^\mli{u}_{n_k}(t)$, $Z^\mli{s}_{m,k}(t)$    & virtual queue backlogs in device $n_k$ and VM $v_{m,k}$                                                                         \\
			$\zeta^\mli{u}_{n_k}(t)$, $\zeta^\mli{s}_{m,k}(t)$    & parameters for controlling delays in $Q^\mli{u}_{n_k}(t)$ and $Q^\mli{s}_{m,k}(t)$ \\
			$g_{n_k}(\cdot)$                           & utility function                                                                                                        \\ \hline
		\end{tabular}
	}	
\vspace{-1.0em}
\end{table}

\subsection{Transmission Model}
\label{subsec:transmission}

We let $n_k\in\mathcal{N}_k$ denote the $n$-th user device of the $k$-th type, and $\mathcal{L}$ be a set of orthogonal channels that the user devices can dynamically access for task offloading in each time slot. Let a stochastic variable $c_{n_k,l,m}(t)$ be the channel capacity of $l$ between device $n_k$ and server $m$ in $t$ (\ie, the maximum amount of tasks that can be transmitted from $n_k$ to $m$ by $l$), satisfying $0\le c_{n_k,l,m}(t)\le c^\mli{max}$ for a positive constant $c^\mli{max}$. 
A user device can access a channel at each slot to offload its tasks to a server or do nothing. Let a binary variable $z_{n_k,l,m}(t)$ denote $n_k$'s offloading strategy in $t$, \ie,
\begin{align}
\label{eq:binary_variable}
z_{n_k,l,m}(t)=0~\text{or}~1,
\end{align}
for each $k \in \mathcal{K}$, $n_k \in \mathcal{N}_k$, $l \in \mathcal{L}$ and $m\in\mathcal{M}$. $z_{n_k,l,m}(t)=1$ indicates that device $n_k$ can offload its tasks to server $m$ via channel $l$; otherwise, $z_{n_k,l,m}(t)=0$. Note that, a user device can only access one channel in a slot, so it follows that
\begin{align}
\label{eq:allocation_1}
\sum_{l \in \mathcal{L}}\sum_{m \in \mathcal{M}}z_{n_k,l,m}(t) \le 1,
\end{align}
for each $k \in \mathcal{K}$ and $n_k \in \mathcal{N}_k$. Further, to circumvent the interference, each channel can only be accessed by one device in the meanwhile, which implies that
\begin{align}
\label{eq:allocation_2}
\sum_{k \in \mathcal{K}}\sum_{n_k \in \mathcal{N}_k}\sum_{m \in \mathcal{M}}z_{n_k,l,m}(t) \le 1.
\end{align}


Define $s_{n_k}(t)$ as the transmission rate for offloading tasks by device $n_k$ in slot $t$. Due to the power limitation, it is reasonable to assume that $s_{n_k}(t)$ is bounded, \ie,
\begin{align}
\label{eq:offloading_1}
0\le s_{n_k}(t)\le \xi_{n_k}(t),
\end{align}
where $0\le\xi_{n_k}(t)\le\xi^\mli{max}_{n_k}$ is a bounded random variable over time slots representing the maximal amount of tasks that can be offloaded to the edge servers. 
Besides, $s_{n_k}(t)$ is also impacted by the channel state, captured by 
\begin{align}
\label{eq:offloading_2}
s_{n_k}(t)\le\sum_{m\in\mathcal{M}}\sum_{l\in\mathcal{L}}z_{n_k,l,m}(t)c_{n_k,l,m}(t),
\end{align}
which implies that $s_{n_k}(t)$ cannot exceed the channel capacity accessed by device $n_k$ in the current slot.

\subsection{Task Buffer Model}
Define stochastic variable $a^\mli{u}_{n_k}(t)$ as the amount of task arrivals on device $n_k$ in $t$, satisfying $0\le a^\mli{u}_{n_k}(t)\le a^\mli{u,max}_{n_k}$ with $a^\mli{u,max}_{n_k}>0$. Due to the delay sensitivity and system stability, user devices can drop some out-of-date tasks, denoted by $d^\mli{u}_{n_k}(t)$, from its local task buffer, and we have
\begin{align}
\label{eq:drop_local}
0\le d^\mli{u}_{n_k}(t)\le d^\mli{u,max}_{n_k},
\end{align}
with a positive constant $d^\mli{u,max}_{n_k}$. Let $Q^\mli{u}_{n_k}(t)$ be the local task buffer of device $n_k$ with the following dynamics
\begin{align}
\label{eq:task_queue_local}
Q^\mli{u}_{n_k}(t+1)=\max\left\{Q^\mli{u}_{n_k}(t)-s_{n_k}(t)-d^\mli{u}_{n_k}(t),0\right\}+a^\mli{u}_{n_k}(t),
\end{align}
where $Q^\mli{u}_{n_k}(0)=0$.

We denote $v_{m,k}$ as the $k$-th VM of server $m$, and $a^\mli{s}_{m,k}(t)$ and $u_{m,k}(t)$ as the amount of arriving and processed tasks on $v_{m,k}$ in slot $t$ respectively. In particular, 
$a^\mli{s}_{m,k}(t)$ can be expressed by
\begin{align}
a^\mli{s}_{m,k}(t)\triangleq \sum_{n_k\in\mathcal{N}_k}\sum_{l \in \mathcal{L}}z_{n_k,l,m}(t)\hat{s}_{n_k}(t),
\end{align}
where $\hat{s}_{n_k}(t)\triangleq \min\{s_{n_k}(t),Q^\mli{u}_{n_k}(t\}$ represents the actual departures of device $n_k$ in $t$. Due to the unpredictable states of edge servers, we assume that the processing rates of VMs are random. Let the bounded stochastic variable $u_{m,k}(t)$ denote the amount of processed tasks on $v_{m,k}$ in $t$, and $0\le u_{m,k}(t)\le u^\mli{max}_{m,k}$ holds. Similarly, let $d^\mli{s}_{m,k}(t)$ denote the amount of dropped tasks on $v_{m,k}$ in slot $t$, satisfying 	
\begin{align}
\label{eq:drop_server}
0\le d^\mli{s}_{m,k}(t)\le d^\mli{s,max}_{m,k},
\end{align}
where $d^\mli{s,max}_{m,k}>0$. Each $v_{m,k}$ also maintains a task buffer, denoted by $Q^\mli{s}_{m,k}(t)$, with the following queuing dynamics\footnote{All the queues use \emph{first-in-first-out} (FIFO) service in this paper.}
\begin{align}
\label{eq:task_queue_server}
Q^\mli{s}_{m,k}(t+1)=&\max\big\{Q^\mli{s}_{m,k}(t)-u_{m,k}(t)-d^\mli{s}_{m,k}(t),0\big\}\nonumber\\
&+a^\mli{s}_{m,k}(t).
\end{align}
Further, since the task buffers in both devices and servers are finite, we impose the following constraints on $Q^\mli{u}_{n_k}(t)$ and $Q^\mli{s}_{m,k}(t)$ for each $m\in\mathcal{M}$, $k\in\mathcal{K}$, and $n_k\in\mathcal{N}_k$
\begin{align}
\label{eq:buffer_size}
Q^\mli{u}_{n_k}(t)\le Q^\mli{u,max}_{n_k},\quad Q^\mli{s}_{m,k}(t)\le Q^\mli{s,max}_{m,k},
\end{align}
where $Q^\mli{u,max}_{n_k}\ge a^\mli{u,max}_{n_k}$ and $Q^\mli{s,max}_{m,k}\ge c^\mli{max}L$ are the buffer sizes for device $n_k$ and VM $v_{m,k}$, respectively.

\subsection{Delay-Constrained Model}

\begin{figure*}[t]
	\begin{align}
		\label{eq:h_u}
		h^\mli{u}_{n_k,t}&\triangleq \min_{T'} \Bigg\{T'~\Bigg |~\sum^{T'}_{t'=t}  \underbrace{\min\big\{s_{n_k}(t') + d^\mli{u}_{n_k}(t'), Q^\mli{u}_{n_k}(t')\big\}}_\text{output of $Q^\mli{u}_{n_k}$ in $t'$}\ge Q^\mli{u}_{n_k}(t) + a^\mli{u}_{n_k}(t),~T'\ge t \Bigg\}\\
		\label{eq:h_s}
		h^\mli{s}_{n_k,t}&\triangleq \min_{T'} \Bigg\{T'~\Bigg |~\sum^{T'}_{t'=h^\mli{u}_{n_k,t}}  \underbrace{\min\big\{u_{m',k}(t') + d^\mli{s}_{m',k}(t'), Q^\mli{s}_{m',k}(t')\big\}}_\text{output of $Q^\mli{s}_{m',k}$ in $t'$}  \ge Q^\mli{s}_{m',k}(t) + a^\mli{s}_{m',k}(t),~T'\ge h^\mli{u}_{n_k,t} \Bigg\}
	\end{align}
\noindent\makebox[\linewidth]{\rule{\textwidth}{0.5pt}}
\end{figure*}

For device $n_k$, define $h^\mli{u}_{n_k,t}$ in \eqref{eq:h_u} as the time slot when task $a^\mli{u}_{n_k}(t)$ leaves local task buffer $Q^\mli{u}_{n_k}$. It indicates that in $h^\mli{u}_{n_k,t}$, all the tasks before $a^\mli{u}_{n_k}(t)$ along with $a^\mli{u}_{n_k}(t)$ have just been offloaded or dropped. Thus, the queuing delay of $a^\mli{u}_{n_k}(t)$ in local device can be denoted as
\begin{align}
	\label{eq:tau_u}
	\tau^\mli{u}_{n_k,t} \triangleq h^\mli{u}_{n_k,t} - t + 1.
\end{align}
When $a^\mli{u}_{n_k}(t)$ is not dropped, we denote $m'$ as the server\footnote{For simplicity, we omit the subscripts, $n_k$ and $t$, for $m'$.}, which $a^\mli{u}_{n_k}(t)$ is offloaded to, such that $z_{n_k,l,m'}(h^\mli{u}_{n_k,t})=1$\footnote{For convenience, we assume task $a^\mli{u}_{n_k}(t)$ is not be split here. It is easy to extend the expression to the splitting case.}. Similar to $h^\mli{u}_{n_k,t}$, define $h^\mli{s}_{n_k,t}$ in \eqref{eq:h_s} as the time slot in which task $a^\mli{u}_{n_k}(t)$ leaves the VM's buffer. Accordingly, the processing delay of $a^\mli{u}_{n_k}(t)$ on the server-side can be expressed as
\begin{align*}
	\tau^\mli{s}_{n_k,t}\triangleq
	\begin{cases}
		0&,~\text{if}~a^\mli{u}_{n_k}(t)~\text{is dropped locally},\\
		h^\mli{s}_{n_k,t}-h^\mli{u}_{n_k,t}&,~\text{otherwise}.
	\end{cases}
\end{align*}
There exists an acceptable worst-case delay for each device's tasks, denoted by $\tau^\mli{max}_{n_k}$, such that 	
\begin{align}
\label{eq:delay_constr}
\tau^\mli{u}_{n_k,t}+\tau^\mli{s}_{n_k,t}\le\tau^\mli{max}_{n_k},~\forall t\in\{0,1,\dots\}.
\end{align}
It implies that each task should be executed before its deadline. 	Since the transmission delay can be seen as a constant, we neglect it for brevity.

\subsection{Task Offloading Problem}
The system's goal is to maximize the long-term utility among user devices while satisfying the worst-case delay constraints. To do so, we first define the average amount of arriving and dropped tasks on each user device below
\begin{align}
\label{eq:average_arrival}
\bar{a}^\mli{u}_{n_k}&\triangleq \lim_{T\rightarrow\infty}\frac{1}{T}\sum^{T-1}_{t=0}\mathbb{E}\left\{a^\mli{u}_{n_k}(t) \right\},\\
\label{eq:average_drop_local}
\bar{d}^\mli{u}_{n_k}&\triangleq \lim_{T\rightarrow\infty}\frac{1}{T}\sum^{T-1}_{t=0}\mathbb{E}\left\{d^\mli{u}_{n_k}(t) \right\}.
\end{align}
In particular, we slightly abuse the notation and define $d^\mli{s}_{n_k}(t)$ (distinct from $d^\mli{s}_{m,k}(t)$) as the total amount of tasks of $n_k$ that have been offloaded to servers but are dropped in $t$. Similarly, we define the expected average of $d^\mli{s}_{n_k}(t)$ below
\begin{align}
\label{eq:average_drop_server}
\bar{d}^\mli{s}_{n_k}&\triangleq \lim_{T\rightarrow\infty}\frac{1}{T}\sum^{T-1}_{t=0}\mathbb{E}\left\{d^\mli{s}_{n_k}(t)\right\}.
\end{align}
Based on that, we cast the task offloading problem as a long-term stochastic optimization problem, \ie, 
\begin{align}
\label{problem_primal}
\begin{split}
\textbf{(P)}~\mathop{\max}_{\left\{\bm{d}(t),\bm{z}(t),\bm{s}(t)\right\}}&~\sum_{k\in\mathcal{K}}\sum_{n_k\in\mathcal{N}_k}g_{n_k}\left(\bar{a}^\mli{u}_{n_k}-\bar{d}^\mli{u}_{n_k}-\bar{d}^\mli{s}_{n_k}\right),\\
\text{s.t.}&~~\eqref{eq:binary_variable}-\eqref{eq:drop_local},\eqref{eq:drop_server},\eqref{eq:buffer_size}-\eqref{eq:delay_constr},
\end{split}
\end{align}
where $g_{n_k}(\cdot)$ is a differentiable, concave, and non-decreasing utility function with a finite maximum first derivative denoted by $\beta_{n_k}$. We denote $\bm{d}(t)\triangleq\{d^\mli{u}_{n_k}(t),d^\mli{s}_{m,k}(t)\mid m\in\mathcal{M},k\in\mathcal{K},n_k\in\mathcal{N}_k\}$, and the same applies to $\bm{z}(t)$ and $\bm{s}(t)$. This formulation aims to find good offloading and resource allocation policies that enable executing as many tasks as possible while satisfying the delay and stability constraints.

\section{TODG: Distributed Task Offloading Scheme with Delay Guarantees}
\label{sec:todg}
In this section, we develop an online algorithm for addressing the problem mentioned above in a distributed manner. We first transform original problem \eqref{problem_primal} into an easy-to-handle form, then decompose it into three slot-level sub-problems. After that, we provide approaches to solve the sub-problems in each time slot.

\subsection{Problem Transformation and Decomposition}

It is not easy to control $d^\mli{s}_{n_k}(t)$ directly since tasks of different devices are mixed in the task buffers on servers. Thus, we derive the next lemma to decouple $\bar{d}^\mli{s}_{n_k}$ from $\bar{d}^\mli{u}_{n_k}$ and show \eqref{problem_primal} can be transformed equivalently to the following handy form with $\beta=\max_{n_k}\beta_{n_k}$.
\begin{align}
\label{problem_handy}
\begin{split}
\mathop{\max}_{\left\{\bm{d}(t),\bm{z}(t),\bm{s}(t)\right\}}&~\sum_{k\in\mathcal{K}}\sum_{n_k\in\mathcal{N}_k}g_{n_k}\left(\bar{a}^\mli{u}_{n_k}-\bar{d}^\mli{u}_{n_k}\right)\\
&~-\beta\sum_{m\in\mathcal{M}}\sum_{k\in\mathcal{K}}\bar{d}^\mli{s}_{m,k},\\
\text{s.t.}&~~\eqref{eq:binary_variable}-\eqref{eq:drop_local},\eqref{eq:drop_server},\eqref{eq:buffer_size}-\eqref{eq:delay_constr},
\end{split}
\end{align}

\begin{lemma}
	\label{lemma:transform}
	If there exists an optimal solution $\{\bm{d}^*(t),\bm{z}^*(t),$ $\bm{s}^*(t)\}$ of problem \eqref{problem_handy}, it is also an optimal solution to problem \eqref{problem_primal}.
\end{lemma}
\begin{proof}
	\label{proof:transform}
	The proof is similar to \cite{neely2011opportunistic}. We first show that by decoupling $\bar{d}^\mli{s}_{n_k}$ from $\bar{d}^\mli{u}_{n_k}$ and rewriting the objective function, the following problem is equivalent to original problem \eqref{problem_primal}
	\begin{align}
	\label{eq:objective}
	\begin{split}
	\mathop{\max}_{\left\{\bm{d}(t),\bm{z}(t),\bm{s}(t)\right\}}&~\sum_{k\in\mathcal{K}}\sum_{n_k\in\mathcal{N}_k}\left(g_{n_k}(\bar{a}^\mli{u}_{n_k}-\bar{d}^\mli{u}_{n_k})-\beta\bar{d}^\mli{s}_{n_k}\right),\\
	\text{s.t.}&~~\eqref{eq:binary_variable}-\eqref{eq:drop_local},\eqref{eq:drop_server},\eqref{eq:buffer_size}-\eqref{eq:delay_constr},
	\end{split}
	\end{align}
	where $\beta\triangleq\max_{n_k}\{\beta_{n_k}\}$. It is easy to see that the optimal solution of problem \eqref{problem_primal} can be obtained at $\bar{d}^\mli{s}_{n_k}=0$ since tasks dropped from the edge server could just as easily have been dropped at the user side. On the other hand, since $\beta$ is greater than or equal to the maximum derivative of $g_{n_k}(\cdot)$. Thus, for the objective \eqref{eq:objective}, transmitting an extra unit of task (improving the utility $g_{n_k}(\bar{a}^\mli{u}_{n_k}-\bar{d}^\mli{u}_{n_k})$) then dropping it on the server (leading to the $\beta\bar{d}^\mli{s}_{n_k}$ penalties) is no better than dropping it directly on the user device. Thus, the above equivalent transformation holds. 
	
	Besides, from the definitions of $d^\mli{s}_{n_k}(t)$ and $d^\mli{s}_{m,k}(t)$, we have $\sum_{k\in\mathcal{K}}\sum_{n_k\in\mathcal{N}_k}d^\mli{s}_{n_k}(t)=\sum_{m\in\mathcal{M}}\sum_{k\in\mathcal{K}}d^\mli{s}_{m,k}(t)$, which implies $\sum_{k\in\mathcal{K}}\sum_{n_k\in\mathcal{N}_k}\bar{d}^\mli{s}_{n_k}=\sum_{m\in\mathcal{M}}\sum_{k\in\mathcal{K}}\bar{d}^\mli{s}_{m,k}$. Combining these two conclusions, we complete the proof.
\end{proof}

Moreover, the coupling of $\tau^\mli{u}_{n_k,t}$ and $\tau^\mli{s}_{n_k,t}$ in \eqref{eq:delay_constr} causes troubles in parallel computing between devices and servers. To handle the challenge, we introduce two types of virtual queues to decouple the control of $\tau^\mli{u}_{n_k,t}$ and $\tau^\mli{s}_{n_k,t}$ and ``relax'' the delay constraints. More specifically, for each user $n_k$ and VM $v_{m,k}$, we define \emph{delay state queues} $Z^\mli{u}_{n_k}(t)$ and $Z^\mli{s}_{m,k}(t)$ to measure the delay in current task buffers as follows
\begin{align}
\label{eq:zqueue_local}
Z^\mli{u}_{n_k}(t+1) =& \max \left\{Z^\mli{u}_{n_k}(t)-s_{n_k}(t)
-d^\mli{u}_{n_k}(t)+\zeta^\mli{u}_{n_k},0\right\},\\
\label{eq:zqueue_server}
Z^\mli{s}_{m,k}(t+1) =& \max \big\{Z^\mli{s}_{m,k}(t)-u_{m,k}(t)-d^\mli{s}_{m,k}(t)\nonumber\\
&+\zeta^\mli{s}_{m,k},0\big\},
\end{align}
with parameters $0\le\zeta^\mli{u}_{n_k}\le d^\mli{u,max}_{n_k}$ and $0\le\zeta^\mli{s}_{m,k}\le d^\mli{s,max}_{m,k}$. Intuitively, the arrivals $\zeta^\mli{u}_{n_k}$ and $\zeta^\mli{s}_{m,k}$ of $Z^\mli{u}_{n_k}(t)$ and $Z^\mli{s}_{m,k}(t)$ can be seen as penalties upon the tasks stuck in the task buffers at each slot. For example, if there is no departure of $Q^\mli{u}_{n_k}(t)$ in the current slot, \ie, $s_{n_k}(t)+d^\mli{u}_{n_k}(t)=0$, the delay of all tasks stored in $Q^\mli{u}_{n_k}(t)$ would increase by one slot. Accordingly, $Z^\mli{u}_{n_k}(t)$ will also increase by $\zeta^\mli{u}_{n_k}$ in slot $t$. Thus, a large backlog of $Z^\mli{u}_{n_k}(t)$ indicates that high latency happens in the current task queue $Q^\mli{u}_{n_k}(t)$. We demonstrate that the delay constraints can be well satisfied if appropriately selecting the parameters $\zeta^\mli{u}_{n_k}$ and $\zeta^\mli{s}_{m,k}$ in Lemma \ref{lemma:delay} after imposing the following assumption. 
\begin{assumption}
	\label{assu:dmax}
	For any $k\in\mathcal{K}$, $n_k\in\mathcal{N}_k$, and $m\in\mathcal{M}$, dropping rates $d^\mli{u,max}_{n_k}$ and $d^\mli{s,max}_{m,k}$ are large enough such that
	\begin{align}
		d^\mli{s,max}_{m,k}&\ge\max\left\{c^\mli{max}L,\mathop{\max}_{n_k}\left\{2Q^\mli{s,max}_{m,k}/\tau^\mli{max}_{n_k}\right\}\right\},\\
		d^\mli{u,max}_{n_k}&\ge\max\left\{a^\mli{u,max}_{n_k},\frac{2Q^\mli{u,max}_{n_k}}{\tau^\mli{max}_{n_k}-\mathop{\max}_{m}\left\{2Q^\mli{s,max}_{m,k}/d^\mli{s,max}_{m,k}\right\}}\right\}.
	\end{align}
\end{assumption}
Intuitively, Assumption \ref{assu:dmax} implies that even though the delay-constrained task arrivals outweighs the system's processing capacity, it is capable of dropping some tasks for system stability. It enables the system to stay in the solution space.
\begin{lemma}
	\label{lemma:delay}
	Suppose for each $m\in\mathcal{M}$, $k\in\mathcal{K}$, and $n_k\in\mathcal{N}_k$, $Z^\mli{u}_{n_k}(t)$ and $Z^\mli{s}_{m,k}(t)$ are bounded by
	\begin{align}
	\label{eq:lem_d_1}
	Z^\mli{u}_{n_k}(t)\le Q^\mli{u,max}_{n_k},\quad Z^\mli{s}_{m,k}(t)\le Q^\mli{s,max}_{m,k}.
	\end{align}
	Then, given Assumption \ref{assu:dmax}, the delay constraint \eqref{eq:delay_constr} can be satisfied if the following holds true for each $m\in\mathcal{M}$, $k\in\mathcal{K}$, and $n_k\in\mathcal{N}_k$
	\begin{align}
	\label{eq:lem_d_3}
	&\mathop{\max}_{n_k}\left\{\frac{2Q^\mli{s,max}_{m,k}}{\tau^\mli{max}_{n_k}}\right\}<\zeta^\mli{s}_{m,k},\\
	\label{eq:lem_d_4}
	&\frac{2Q^\mli{u,max}_{n_k}}{\tau^\mli{max}_{n_k}-\mathop{\max}_{m}\left\{\frac{2Q^\mli{s,max}_{m,k}}{\zeta^\mli{s}_{m,k}}\right\}}\le\zeta^\mli{u}_{n_k}.
	\end{align}
\end{lemma}
\begin{proof}
	For fixed $\zeta^\mli{u}_{n_k}$ and $\zeta^\mli{s}_{m,k}$, we define the worst-case queuing delays on $Q^\mli{u}_{n_k}(t)$ and $Q^\mli{s}_{m,k}(t)$ are $w^\mli{u}_{n_k}$ and $w^\mli{s}_{m,k}$ respectively. Based on \cite[Lemma 5.5]{neely2010stochastic}, for each $m\in\mathcal{M}$, $k\in\mathcal{K}$, and $n_k\in\mathcal{N}_k$, it can be shown that
	\begin{align}
	\label{eq:lem_d_6}
	w^\mli{u}_{n_k}=\frac{2Q^\mli{u,max}_{n_k}}{\zeta^\mli{u}_{n_k}},\quad
	w^\mli{s}_{m,k}=\frac{2Q^\mli{s,max}_{m,k}}{\zeta^\mli{s}_{m,k}}.
	\end{align}
	To satisfy \eqref{eq:delay_constr}, it requires that
	\begin{align}
	\label{eq:lem_d_5}
	&w^\mli{u}_{n_k}+\mathop{\max}_m \{w^\mli{s}_{m,k}\}\le\tau^\mli{max}_{n_k},\nonumber\\
	\Leftrightarrow\quad & \frac{2Q^\mli{u,max}_{n_k}}{\zeta^\mli{u}_{n_k}}+\mathop{\max}_m \left\{\frac{2Q^\mli{s,max}_{m,k}}{\zeta^\mli{s}_{m,k}}\right\}\le\tau^\mli{max}_{n_k}.
	\end{align}
	Due to $w^\mli{u}_{n_k},w^\mli{s}_{m,k}>0$, from \eqref{eq:lem_d_5}, we have
	\begin{align}
	\frac{2Q^\mli{s,max}_{m,k}}{\zeta^\mli{s}_{m,k}}<\tau^\mli{max}_{n_k},~\forall m\in\mathcal{M},n_k\in\mathcal{N}_k,
	\end{align}
	by which we obtain \eqref{eq:lem_d_3}. Plugging \eqref{eq:lem_d_3} into \eqref{eq:lem_d_5} and rearranging the terms, we complete the proof.
\end{proof}

\begin{remark}
	\label{remark:zeta}
	Lemma \ref{lemma:delay} implies that if we can bound $Z^\mli{u}_{n_k}(t)\le Q^\mli{u,max}_{n_k}$ and $Z^\mli{s}_{m,k}(t)\le Q^\mli{s,max}_{m,k}$, constraint \eqref{eq:delay_constr} can be satisfied by setting sufficiently large $\zeta^\mli{s}_{m,k}$ and $\zeta^\mli{u}_{n_k}$ corresponding to \eqref{eq:lem_d_3}-\eqref{eq:lem_d_4}. Besides, since $Z^\mli{u}_{n_k}(t)$ and $Z^\mli{s}_{m,k}(t)$ are bounded, the following holds
	\begin{align}
	\label{eq:zeta_imp_1}
	&\bar{d}^\mli{u}_{n_k}\ge \zeta^\mli{u}_{n_k}-\bar{s}^\mli{u}_{n_k}\ge\zeta^\mli{u}_{n_k}-\min\{\xi_{n_k}(t),c^\mli{max}\},\\
	\label{eq:zeta_imp_2}
	&\sum_{k,n_k}\bar{d}^\mli{s}_{n_k}=\sum_{m,k}\bar{d}^\mli{s}_{m,k}\ge \sum_{m,k}\left(\zeta^\mli{s}_{m,k}-\bar{u}_{m,k}\right).
	\end{align}
	Therefore, although larger $\zeta^\mli{s}_{m,k}$ and $\zeta^\mli{u}_{n_k}$ result in lower delay for offloaded tasks (see \eqref{eq:lem_d_6}), inequalities \eqref{eq:zeta_imp_1}-\eqref{eq:zeta_imp_2} indicate that a further increase in $\zeta^\mli{s}_{m,k}$ and $\zeta^\mli{u}_{n_k}$ may cause an increase of the dropped tasks and lead to additional degradation of the system utility.
\end{remark}

Based on Lemma \ref{lemma:transform} and \ref{lemma:delay}, original problem \eqref{problem_primal} can be transformed into the following one
\begin{align}
\label{problem_tran}
\begin{split}
\textbf{(TP)}~\mathop{\max}_{\left\{\bm{d}(t),\bm{z}(t),\bm{s}(t)\right\}}&~\sum_{k\in\mathcal{K}}\sum_{n_k\in\mathcal{N}_k}g_{n_k}\left(\bar{a}^\mli{u}_{n_k}-\bar{d}^\mli{u}_{n_k}\right)\\
&~-\beta\sum_{m\in\mathcal{M}}\sum_{k\in\mathcal{K}}\bar{d}^\mli{s}_{m,k},\\
\text{s.t.}&~~\eqref{eq:binary_variable}-\eqref{eq:drop_local},\eqref{eq:drop_server},\eqref{eq:buffer_size},\eqref{eq:lem_d_1}.
\end{split}
\end{align}

Next, we employ the dual-based \emph{drift-plus-penalty} technique \cite{neely2006energy} to decompose problem \eqref{problem_tran} into slot-level sub-problems. Rather than optimizing the problem directly, the idea of this technique is to minimize the slot-level drift-plus-penalty function, composed by the one-slot utility function and the successive difference of the queue state measures (namely ``Lyapunov drift''). 
By doing so, it enables maximizing the performance while implicitly controlling the system stability. More specifically, in each slot $t$, we define the drift-plus-penalty function $D_{\epsilon}(t)$ as follows
\begin{align}
\label{eq:drift_penality}
D_{\epsilon}(t)\triangleq & -\epsilon\cdot\Big(\sum_{k\in\mathcal{K}}\sum_{n_k\in\mathcal{N}_k}g_{n_k}\left(a^\mli{u}_{n_k}(t)-d^\mli{u}_{n_k}(t)\right)\nonumber\\
&-\sum_{m\in\mathcal{M}}\sum_{k\in\mathcal{K}}\beta d^\mli{s}_{m,k}(t)\Big)+L(t+1)-L(t),
\end{align}
where $\epsilon\ge0$ is a weight parameter to balance the utility and latency, and $L(t)$ is the \emph{Lyapunov function} defined by
\begin{align}
\label{eq:lyapunov_func}
L(t)\triangleq&
\frac{1}{2}\sum_{k\in\mathcal{K}}\sum_{n_k\in\mathcal{N}_k}\left(Q^\mli{u}_{n_k}(t)^2+Z^\mli{u}_{n_k}(t)^2\right)\nonumber\\
&\frac{1}{2}\sum_{m\in\mathcal{M}}\sum_{k\in\mathcal{K}}\left(Q^\mli{s}_{m,k}(t)^2+Z^\mli{s}_{m,k}(t)^2\right).
\end{align}
However, minimizing $D_{\epsilon}(t)$ directly is often computationally costly. A common alternative method is to minimize its upper bound, which is given by the following lemma.
\begin{lemma}
	\label{lemma:drift_bound}
	Under any $\epsilon\ge0$, the upper bound of the drift-plus-penalty function $D_{\epsilon}(t)$ can be expressed by
	\begin{align}
	\label{eq:drift_bound}
	\hat{D}_{\epsilon}(t)&\triangleq C-\epsilon\cdot\Big( \sum_{k\in\mathcal{K}}\sum_{n_k\in\mathcal{N}_k}g_{n_k}\left(a^\mli{u}_{n_k}(t)-d^\mli{u}_{n_k}(t)\right)\nonumber\\
	&-\sum_{m\in\mathcal{M}}\sum_{k\in\mathcal{K}}\beta d^\mli{s}_{m,k}(t)\Big) \nonumber\\
	\nonumber
	&+\sum_{k\in\mathcal{K}}\sum_{n_k\in\mathcal{N}_k}\Big(Q^\mli{u}_{n_k}(t)\left(a^\mli{u}_{n_k}(t)-s_{n_k}(t)-d^\mli{u}_{n_k}(t)\right)\\
	\nonumber
	&+Z^\mli{u}_{n_k}(t)\left(\zeta^\mli{u}_{n_k}-s_{n_k}(t)-d^\mli{u}_{n_k}(t)\right)\Big)\\
	\nonumber
	&+\sum_{m\in\mathcal{M}}\sum_{k\in\mathcal{K}}\Big(Q^\mli{s}_{m,k}(t)\left(a^\mli{s}_{m,k}(t)-u_{m,k}(t)-d^\mli{s}_{m,k}(t)\right)\\
	&+Z^\mli{s}_{m,k}(t)\left(\zeta^\mli{s}_{m,k}-u_{m,k}(t)-d^\mli{s}_{m,k}(t)\right)\Big),
	\end{align}
	where $C$ is denoted by
	\begin{align}
	\label{eq:C}
	C \triangleq &\sum_{k\in\mathcal{K}}\sum_{i\in\mathcal{N}_k}\left((a^\mli{u,max}_{n_k})^2+(\xi^\mli{max}_{n_k}+d^\mli{u,max}_{n_k})^2\right)\nonumber\\
	&+\sum_{m\in\mathcal{M}}\sum_{k\in\mathcal{K}}\big((c^\mli{max}L)^2+(u^\mli{max}_{m,k}+d^\mli{s,max}_{m,k})^2\big).
	\end{align}
\end{lemma}
\begin{proof}
	\label{proof:drift_bound}
	Note that for any $Q\ge0$, $b\ge0$, $a\ge0$, we have
	\begin{align}
	\label{eq:db_1}
	\left(\max\{Q-b,0\}+a\right)^2\le Q^2+a^2+b^2+2Q(a-b).
	\end{align}
	Based on \eqref{eq:db_1}, squaring the dynamics of $Q^\mli{u}_{n_k}(t)$ in \eqref{eq:task_queue_local} yields
	\begin{align}
	&\left(Q^\mli{u}_{n_k}(t+1)\right)^2-\left(Q^\mli{u}_{n_k}(t)\right)^2\nonumber\\
	\le& \left(a^\mli{u,max}_{n_k}\right)^2 + \left(\xi_{n_k}^\mli{max}+d^\mli{u,max}_{n_k}\right)^2+2Q^\mli{u}_{n_k}(t)\big(a^\mli{u}_{n_k}(t)\nonumber\\
	&-s^\mli{u}_{n_k}(t)-d^\mli{u}_{n_k}(t)\big).
	\end{align}
	Similarly, we have
	\begin{align}
	&\left(Q^\mli{s}_{m,k}(t+1)\right)^2-\left(Q^\mli{s}_{m,k}(t)\right)^2\nonumber\\
	\le& \left(c^\mli{max}L\right)^2 + \big(u^\mli{max}_{m,k}+d^{s,max}_{m,k}\big)^2+2Q^\mli{s}_{m,k}(t)\big(a^\mli{s}_{m,k}(t)\nonumber\\
	&-u_{m,k}(t)-d^\mli{e}_{m,k}(t)\big).
	\end{align}
	Squaring the dynamics of $Z^\mli{u}_{n_k}(t)$ in \eqref{eq:zqueue_local} and using the fact that $\max\{a,0\}^2\le a^2$ and $(\zeta^\mli{u}_{n_k}(t)-s_{n_k}(t)-d^\mli{u}_{n_k}(t))^2\le (\xi_{n_k}^\mli{max}+d^\mli{u,max}_{n_k})^2$, we have
	\begin{align}
	\label{eq:db_2}
	&\left(Z^\mli{u}_{n_k}(t+1)\right)^2-\left(Z^\mli{u}_{n_k}(t)\right)^2\nonumber\\
	\le& \left(\xi_{n_k}^\mli{max}+d^\mli{u,max}_{n_k}\right)^2+2Z^\mli{u}_{n_k}(t)\big(\zeta^\mli{u}_{n_k}(t)-s^\mli{u}_{n_k}(t)-d^\mli{u}_{n_k}(t)\big).
	\end{align}
	Similar to \eqref{eq:db_2}, the following holds
	\begin{align}
	&\big(Z^\mli{s}_{m,k}(t+1)\big)^2-\big(Z^\mli{s}_{m,k}(t)\big)^2\nonumber\\
	\le& \big(u^\mli{max}_{m,k}+d^\mli{s,max}_{m,k}\big)^2+ 2Z^\mli{s}_{m,k}(t)\big(\zeta^\mli{s}_{m,k}-u_{m,k}(t)-d^\mli{s}_{m,k}(t)\big).
	\end{align}
	Summing the squared differences in the queues yields the result.
\end{proof}
Therefore, instead of optimizing long-term problem \eqref{problem_tran}, we attempt to minimize the following dual problem in each slot
\begin{align}
\label{problem_dual}
\begin{split}
&\textbf{(DP)}~\mathop{\min}_{\bm{d}(t),\bm{z}(t),\bm{s}(t)}~\hat{D}_{\epsilon}(t),\\
&\quad\quad\quad\quad~~\text{s.t.}~~\eqref{eq:binary_variable}-\eqref{eq:drop_local},\eqref{eq:drop_server}.
\end{split}
\end{align}
It is worth noting that we do not need to explicitly handle Constraints \eqref{eq:buffer_size} and \eqref{eq:lem_d_1} in solving dual problem \eqref{problem_dual}, simplifying the control process. Later, we will rigorously show that \eqref{eq:buffer_size} and \eqref{eq:lem_d_1} can be satisfied by the proposed algorithm via properly setting weight parameter $\epsilon$ in Theorem \ref{thm:stability}.

Problem \eqref{problem_dual} can be decomposed into the following three sub-problems.
\begin{itemize}
	\item \textbf{Dropping on user device.} In each slot $t$, observing the current task arrivals, the buffer state, and the virtual queue backlog, user device $n_k$ decides the amount of dropped tasks (\ie, $d^\mli{u}_{n_k}(t)$) via solving the following sub-problem
	\begin{align}
	\label{subpro_drop_local}
	\begin{split}
	\max_{d^\mli{u}_{n_k}}&~\epsilon \cdot g_{n_k}\left(a^\mli{u}_{n_k}(t)-d^\mli{u}_{n_k}\right)+\left(Q^\mli{u}_{n_k}(t)+Z^\mli{u}_{n_k}(t)\right)d^\mli{u}_{n_k}, \\
	\text{s.t.}&~0\le d^\mli{u}_{n_k}\le d^\mli{u,max}_{n_k}.
	\end{split}
	\end{align}
	\item \textbf{Dropping on edge server.} During slot $t$, based on the current buffer state and the virtual queue backlog, the amount of dropped tasks on VM $v_{m,k}$ (\ie, $d^\mli{s}_{m,k}(t)$) depends on the solution of the sub-problem below
	\begin{align}
	\label{subpro_drop_server}
	\begin{split}
	\max_{d^\mli{s}_{m,k}}&~\left(Q^\mli{s}_{m,k}(t)+Z^\mli{s}_{m,k}(t)-\beta\epsilon\right)d^\mli{s}_{m,k}, \\
	\text{s.t.}&~~0\le d^\mli{s}_{m,k}\le d^\mli{s,max}_{m,k}.
	\end{split}
	\end{align}
	\item \textbf{Offloading decision.} We solve the following sub-problem to obtain the offloading decision and the amount of transmitted tasks (\ie, $s_{n_k}$ and $z_{n_k,l,m}(t)$) per slot, \ie,
	\begin{align}
	\label{subpro_offloading}
	\begin{split}
	\max_{\{s_{n_k},z_{n_k,l,m}\}} &~\sum_{k\in\mathcal{K}}\sum_{n_k\in\mathcal{N}_k}\left(Q^\mli{u}_{n_k}(t)+Z^\mli{u}_{n_k}(t)\right)s_{n_k}\\
	&~-\sum_{m\in\mathcal{M}}\sum_{k\in\mathcal{K}}\sum_{n_k\in\mathcal{N}_k}\sum_{l\in\mathcal{L}}Q^\mli{s}_{m,k}(t)s_{n_k}z_{n_k,l,m},\\
	\text{s.t.}&~~\eqref{eq:binary_variable}-\eqref{eq:offloading_2}.
	\end{split}
	\end{align}
\end{itemize}

\subsection{Distributed Task Offloading with Delay Guarantees}
This subsection provides a distributed \emph{Task Offloading with Delay Guarantees} algorithm (TODG) to solve problem \eqref{problem_tran} by deriving the solutions to sub-problems \eqref{subpro_drop_local}-\eqref{subpro_offloading}. Note that \eqref{subpro_drop_local} and \eqref{subpro_drop_server} are both convex optimization problems, of which the closed-form expression can be easily found. However, solving problem \eqref{subpro_offloading} is a non-trivial task since \eqref{subpro_offloading} resembles a 3-dimensional matching problem, \ie, matching among user devices, channels and edge servers (see Fig. \ref{figure:system_model}), which has been proven an NP-hard problem in the literature \cite{kann1991maximum,kushagra2020three}. Besides, the offloading decision variable $z_{n_k,l,m}(t)$ is also coupled with the amount of transmitted tasks $s_{n_k}(t)$, which exacerbates the complexity of the problem. Next, we provide the solutions to each sub-problem separately.

\subsubsection{Dropping on user device}
Let $d^{\text{u}*}_{n_k}(t)$ denote the optimal solution of sub-problem \eqref{subpro_drop_local} for each $k\in\mathcal{K}$ and $n_k\in\mathcal{N}_k$, then the expression of $d^{\text{u}*}_{n_k}(t)$ can be derived from the theorem below.
\begin{lemma}
	\label{lemma:solution_drop_local}
	Suppose that $g_{n_k}(0)=0$. If there exists $d_0\in\mathbb{R}$ such that $-\epsilon\cdot g'\big(a^\mli{u}_{n_k}(t)-d_0\big)+Q^\mli{u}_{n_k}(t)+Z^\mli{u}_{n_k}(t)=0$ holds, $d^{\text{u}*}_{n_k}(t)$ can be expressed by
	\begin{align}
	\label{eq:d_opt1_local}
	d^{\text{u}*}_{n_k}(t)=
	\begin{cases}
	\left[a^\mli{u}_{n_k}(t)-{g'}^{-1}_{n_k}\left(\frac{Q^\mli{u}_{n_k}(t)+Z^\mli{u}_{n_k}(t)}{\epsilon}\right)\right]^{d^\mli{u,max}_{n_k}}_{0}&,~\epsilon>0,\\
	d^\mli{u,max}_{n_k}&,~\epsilon=0,
	\end{cases}
	\end{align}
	where $g'_{n_k}(\cdot)$ denotes the first-order derivative of $g_{n_k}(\cdot)$ and $[x]^a_b\triangleq\min\{\max\{x,b\},a\}$. Otherwise, we have
	\begin{align}
	\label{eq:d_opt2_local}
	d^{\text{u}*}_{n_k}(t)=&\mathop{\arg\max}_{d\in\left\{0,d^\mli{u,max}_{n_k}\right\}}\Big\{\epsilon \cdot g_{n_k}\left(a^\mli{u}_{n_k}(t)-d\right)+\nonumber\\
	&\left(Q^\mli{u}_{n_k}(t)+Z^\mli{u}_{n_k}(t)\right)\cdot d\Big\}.
	\end{align}
\end{lemma}
\begin{proof}
	\label{proof:solution_drop_local}
	For convenience, in slot $t$, we denote 
	\begin{align}
	f_{n_k}(d)\triangleq \epsilon \cdot g_{n_k}\left(a^\mli{u}_{n_k}(t)-d\right)+\left(Q^\mli{u}_{n_k}(t)+Z^\mli{u}_{n_k}(t)\right)d.
	\end{align}
	When ${g'}^{-1}_{n_k}(\cdot)$ is defined at $(Q_{n_k}(t)+Z_{n_k}(t))/V$, we discuss the optimal solutions under two cases. If $\epsilon>0$, the first-order stationary point $\tilde{d}_{n_k}(t)$ of $f_{n_k}(\cdot)$ is 
	\begin{align}
	\tilde{d}_{n_k}(t)\triangleq a^\mli{u}_{n_k}(t)-{g'}^{-1}_{n_k}\left(\frac{Q^\mli{u}_{n_k}(t)+Z^\mli{u}_{n_k}(t)}{\epsilon}\right).
	\end{align}
	Let $d^{\text{u}*}_{n_k}(t)$ denote the optimal solution of sub-problem \eqref{subpro_drop_local} for each $k\in\mathcal{K}$ and $n_k\in\mathcal{N}_k$, then the expression of $d^{\text{u}*}_{n_k}(t)$ can be derived from the theorem below.
\end{proof}
Lemma \ref{lemma:solution_drop_local} indicates that the optimal $d^\mli{u}_{n_k}(t)$ depends on the task queue and delay state queue backlogs. A small value of $Q^\mli{u}_{n_k}(t)+Z^\mli{u}_{n_k}(t)$ implies that the available task buffer is sufficient, and the current latency is relatively low on the local device, so it is unnecessary to drop tasks. On the contrary, when $Q^\mli{u}_{n_k}(t)+Z^\mli{u}_{n_k}(t)$ is large, it will be better to drop some outdated tasks for system stability and rapid response.

\subsubsection{Dropping on edge server}
Clearly, sub-problem \eqref{subpro_drop_server} is a linear programming problem. For each $m\in\mathcal{M}$ and $k\in\mathcal{K}$, let $d^{\text{s}*}_{m,k}(t)$ denote the optimal solution of \eqref{subpro_drop_server}, then we have
\begin{align}
\label{eq:d_opt_server}
d^{\text{s}*}_{m,k}(t)=
\begin{cases}
d^\mli{s,max}_{m,k}&,~\text{if}~Q^\mli{s}_{m,k}(t)+Z^\mli{s}_{m,k}(t)>\beta\epsilon,\\
0&,~\text{else}.
\end{cases}
\end{align}
Similar to Lemma \ref{lemma:solution_drop_local}, edge servers will dynamically adapt the dropping rates based on the current queue states. 

\subsubsection{Offloading decision}
Aiming at optimizing sub-problem \eqref{subpro_offloading}, we first attempt to transform it into a bipartite matching problem. Then, we provide a decentralized method to resolve it efficiently. To begin with, we derive the following theorem to find the optimal amount of s from the local buffer.
\begin{lemma}
	\label{lemma:optimal_s}
	For each $k\in\mathcal{K}$, $n_k\in\mathcal{N}_k$, $l\in\mathcal{L}$, and $m\in\mathcal{M}$, let $s^*_{n_k}(t)$ and $z^*_{n_k,l,m}(t)$ denote the optimal solution of problem \eqref{subpro_offloading} in slot $t$. Then, the following holds
	\begin{align}
	\label{eq:s_opt}
	s^*_{n_k}(t)=\sum_{m\in\mathcal{M}}\sum_{l\in\mathcal{L}}z^*_{n_k,l,m}(t)\cdot\min\left\{\xi_{n_k}(t),c_{n_k,l,m}(t)\right\}.
	\end{align}
\end{lemma}
\begin{proof}
	\label{proof:optimal_s}
	Note that if user device $n_k$ is not selected to offload tasks in slot $t$, \ie, $z^*_{n_k,l,m}(t)=0$ for all $l\in\mathcal{L}$, then $s^*_{n_k}(t)=0$ must hold. Besides, if there exists $m'\in\mathcal{M}$ and $l'\in\mathcal{L}$ such that $z^*_{n_k,l',m'}(t)=1$, it is easy to see that $Q^\mli{u}_{n_k}(t)+Z^\mli{u}_{n_k}(t)-Q^\mli{s}_{m',k}(t)>0$; otherwise, $z^*_{n_k,l',m'}(t)=0$ is a  better or equal solution. Denote $h(\bm{s},\bm{z})$ in slot $t$ as
	\begin{align}
	\label{eq:h}
	h(\bm{s},\bm{z})\triangleq &\sum_{k\in\mathcal{K}}\sum_{n_k\in\mathcal{N}_k}\left(Q^\mli{u}_{n_k}(t)+Z^\mli{u}_{n_k}(t)\right)s_{n_k}\nonumber\\
	&-\sum_{m\in\mathcal{M}}\sum_{k\in\mathcal{K}}\sum_{n_k\in\mathcal{N}_k}\sum_{l\in\mathcal{L}}Q^\mli{s}_{m,k}(t)s_{n_k}z_{n_k,l,m}.
	\end{align}
	Then, we have $\partial h(\bm{z}^*(t),\bm{s}(t))/\partial s^*_{n_k}(t)>0$. Combined with \eqref{eq:offloading_1} and \eqref{eq:offloading_2}, the proof is completed.
\end{proof}
Lemma \ref{lemma:optimal_s} shows that if a user device is selected to offload tasks during a slot, it will transmit as many tasks as possible to the corresponding edge server from its local task buffer. Based on that, we can transform sub-problem \eqref{subpro_offloading} into the following one to obtain $z^*_{n_k,l,m}(t)$, \ie,
\begin{align}
\label{subpro_z}
\begin{split}
\max_{\{z_{n_k,l,m}\}} &~\sum_{m\in\mathcal{M}}\sum_{k\in\mathcal{K}}\sum_{n_k\in\mathcal{N}_k}\sum_{l\in\mathcal{L}}\Big(Q^\mli{u}_{n_k}(t)+Z^\mli{u}_{n_k}(t)\\
&~-Q^\mli{s}_{m,k}(t)\Big)\min\left\{\xi_{n_k}(t),c_{n_k,l,m}(t)\right\}\cdot z_{n_k,l,m},\\
\text{s.t.}&~~
\begin{cases}
\sum_{l \in \mathcal{L}}\sum_{m \in \mathcal{M}}z_{n_k,l,m} \le 1,\\
\sum_{k \in \mathcal{K}}\sum_{n_k \in \mathcal{N}_k}\sum_{m \in \mathcal{M}}z_{n_k,l,m} \le 1,\\
z_{n_k,l,m}=0~\text{or}~1.
\end{cases}
\end{split}
\end{align}
It is easy to see that the optimal solution to the above problem can be roughly seen as a ``maximum weight matching'' over $\mathcal{N}\times\mathcal{L}\times\mathcal{M}$, with $(Q^\mli{u}_{n_k}(t)+Z^\mli{u}_{n_k}(t)-Q^\mli{s}_{m,k}(t))\cdot\min\{\xi_{n_k}(t),c_{n_k,l,m}(t)\}$ being the weight of tuple $(n_k,l,m)$. Despite \eqref{subpro_z} resembling a 3-dimensional matching problem, we argue that it is equivalent to a bipartite matching in the next lemma. 
\begin{lemma}
	\label{lemma:bipartite}
	In slot $t$, for each $k\in\mathcal{K}$, $n_k\in\mathcal{N}_k$, and $l\in\mathcal{L}$, if there exists $m_{n_k,l}\in\mathcal{M}$ such that $z^*_{n_k,l,m_{n_k,l}}(t)=1$, we have
	\begin{align}
	\label{eq:m_opt}
	m_{n_k,l}=&\mathop{\arg\max}_{m\in\mathcal{M}}\Big\{\left(Q^\mli{u}_{n_k}(t)+Z^\mli{u}_{n_k}(t)-Q^\mli{s}_{m,k}(t)\right)\nonumber\\
	&\cdot\min\left\{\xi_{n_k}(t),c_{n_k,l,m}(t)\right\}\Big\}.
	\end{align}
\end{lemma}
\begin{proof}
	\label{proof:bipartite}
	Based on Lemma \ref{lemma:optimal_s}, the result can be easily obtained by contradiction.
\end{proof}
Lemma \ref{lemma:bipartite} implies that the matching between user devices and  channels suffices to determine the optimal solution of \eqref{subpro_z}. That is, problem \eqref{subpro_z} can be rewritten as a \emph{maximum weight bipartite matching} problem over the source set $\mathcal{N}$ and destination set $\mathcal{L}$ as follows
\begin{align}
\label{subpro_mwbm}
\begin{split}
\max_{\{z_{n_k,l}\}} &~\sum_{k\in\mathcal{K}}\sum_{n_k\in\mathcal{N}_k}\sum_{l\in\mathcal{L}} w_{n_k,l}(t) z_{n_k,l},\\
\text{s.t.}&~~
\begin{cases}
\sum_{l \in \mathcal{L}}z_{n_k,l} \le 1,\\
\sum_{k \in \mathcal{K}}\sum_{n_k \in \mathcal{N}_k}z_{n_k,l} \le 1,\\
z_{n_k,l}=0~\text{or}~1,\\
k\in\mathcal{K},n_k\in\mathcal{N}_k,l\in\mathcal{L}.
\end{cases}
\end{split}
\end{align}
The weight $w_{n_k,l}(t)$ between node $n_k$ and $l$ is denoted by
\begin{align}
\label{eq:weight}
w_{n_k,l}(t)\triangleq
\begin{cases}
\phi_{n_k,l}(t)&,~\text{if}~\phi_{n_k,l}(t)>0,\\
-\infty&,~\text{else},\\
\end{cases}
\end{align}
where
\begin{align}
	\phi_{n_k,l}(t)\triangleq\;
	&\Big(Q^\mli{u}_{n_k}(t)+Z^\mli{u}_{n_k}(t)-Q^\mli{s}_{m_{n_k,l},k}(t)\Big)\nonumber\\
	& \cdot \min\big\{\xi_{n_k}(t),c_{n_k,l,m_{n_k,l}}(t)\big\}.
\end{align}
Based on that, $z^*_{n_k,l,m}(t)$ can be derived via the optimal solution $z^*_{n_k,l}(t)$ of \eqref{subpro_mwbm}, \ie,
\begin{align}
\label{eq:z_star}
z^*_{n_k,l,m}(t)=
\begin{cases}
1,&~\text{if}~z^*_{n_k,l}(t)=1~\text{and}~m=m_{n_k,l},\\
0,&~\text{else}.
\end{cases}
\end{align}
Next, the key point is how to efficiently solve the maximum bipartite matching problem across time slots. One feasible solution is to employ the well-known \emph{Kuhn-Munkres} method to find the maximum matching within $\mathcal{O}\big(\max\{N,L\}^3\big)$ iterations \cite{bourgeois1971extension}. However, it is a computationally costly and highly centralized way, where all the information (including queue and channel states) requires sending to a central controller. Thus, it may be unworkable in some practical scenarios. Instead, building on the recent advance in multi-robot applications, we argue that the optimal matching can also be achieved in a distributed manner via the \emph{Multi-Robot Assignment} algorithm \cite{chopra2017distributed}. More specifically, each user device can send its corresponding weight (\ie, $w_{n_k.l}(t)$ in \eqref{eq:weight}) to only one of the connectable edge servers. After that, each edge server can carry out the local matching in parallel while exchanging necessary information with the adjacent servers as in \cite{chopra2017distributed}. Nevertheless, although it can alleviate the high computational complexity of the centralized methods via parallel computing, it may lead to relatively large communication cost, \ie, in the worst case requiring $\mathcal{O}(H\cdot\max\{N,L\}^2)$ communication rounds, with $H$ being the maximum hops among servers (see \cite[Corollary 4]{chopra2017distributed} for more details). 



\textbf{\emph{Periodic Strategy.}} To tackle these issues, we provide a \emph{$\delta$-periodic strategy} to mitigate the computational and communication cost incurred by the \emph{Multi-Robot Assignment} approach. That is, we allow only running the matching \eqref{eq:z_star} every $\delta$ slots, \ie, in $t\in\mathcal{T}_\delta\triangleq\{0,\delta,2\delta,\dots\}$; otherwise, we keep the channel allocation the same as that in the last slot. More specifically, in slot $t\notin\mathcal{T}_{\delta}$: 
\begin{itemize}
	\item The devices \emph{selected} in the last slot occupy the same channels and only need to decide the target servers;
	\item The devices \emph{not selected} in the last slot remain idle for the current slot.
\end{itemize}
In a nutshell, the $\delta$-periodic strategy computes the offloading decisions by 
\begin{align}
\label{eq:delta_periodic}
z_{n_k,l,m}(t)=
\begin{cases}
z^*_{n_k,l,m}(t),~&\text{if}~t\in\mathcal{T}_\delta,\\
z^*_{n_k,l}(t-1),~&\text{else if}~m=\hat{m}_{n_k,l}(t)\\
~&\text{and}~\hat{m}_{n_k,l}(t)\neq\mli{idle},\\
0,~&\text{else},
\end{cases}
\end{align}
where $z^*_{n_k,l,m}(t)$ is derived from \eqref{eq:z_star}, and $m_{n_k,l}(t)$ is expressed by
\begin{align}
	\hat{m}_{n_l,l}(t)=
	\begin{cases}
		m_{n_k,l} ,~&\text{if}~\phi_{n_k,l}(t)>0,\\
		\mli{idle},~&\text{else}.
	\end{cases}
\end{align}
Although the devices still need to decide the target servers $m_{n_k,l}(t)$ in each slot, the corresponding computational and communication cost is negligible, because $Q^\mli{s}_{m,k}(t)$ is a scalar, and the complexity of computing $m_{n_k,l}(t)$ is only $O(M)$. The idea behind the $\delta$-periodic strategy is based upon the observation that the weights $\{w_{n_k,l}(t)\}$ in \eqref{eq:weight}, determined by the channel capacities and queue backlogs, generally would not change too sharply in adjacent slots. Thus, intuitively the error posed by the periodic strategy between the optimal solution of \eqref{subpro_offloading} is acceptable. In Section 5, we rigorously quantify the induced error and show the $\delta$-periodic strategy can achieve a trade-off between the near-optimal system utility and the computational cost. It is worth noting that the periodic strategy also leads to a variant of the standard analytical techniques due to violating the optimality of the sub-problems' solutions.
\vspace{-0.5em}

\begin{figure}[ht]
\centering
\includegraphics[width=1.0\columnwidth]{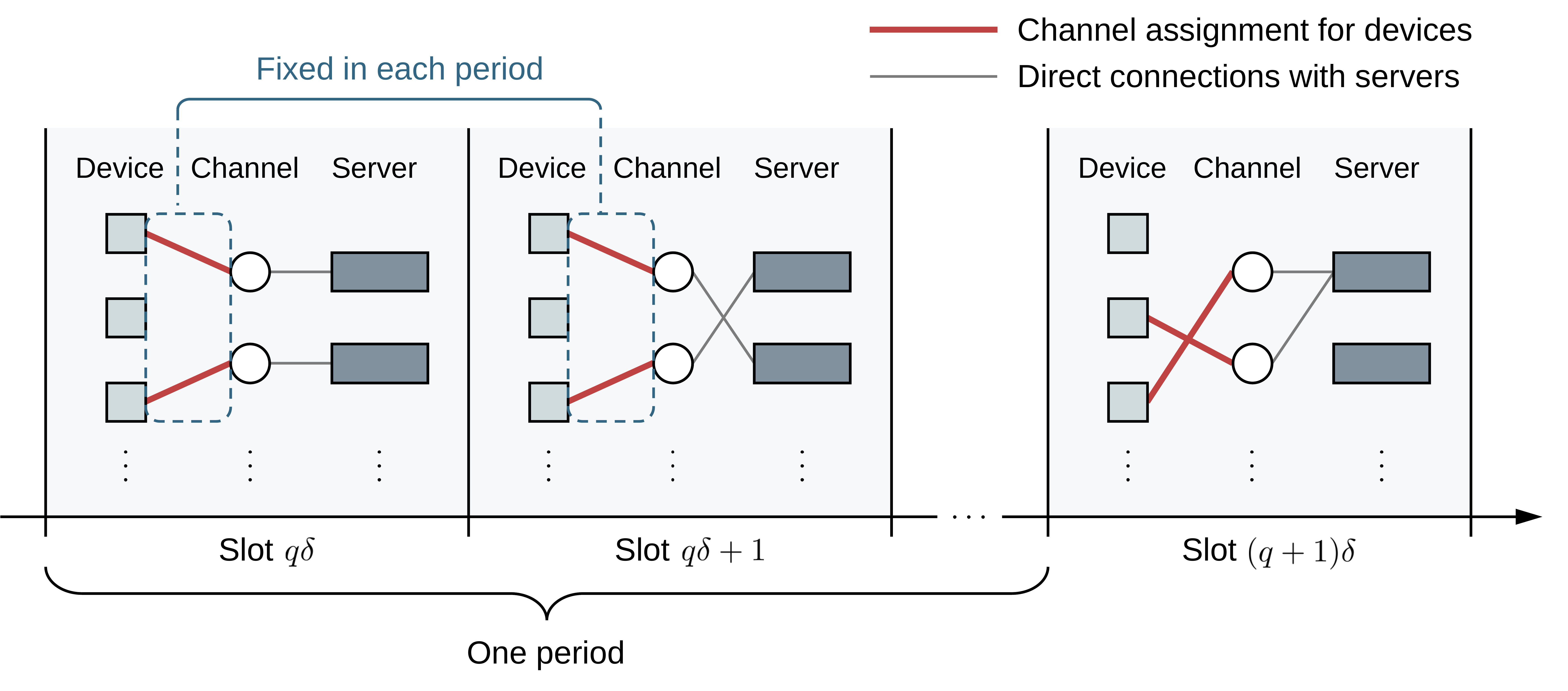} 
\vspace{-2.0em}
\caption{Periodic strategy. Here, $q$ is a non-negative integer. The $\delta$-periodic strategy enables the system to carry out channel allocation every $\delta$ slots. Within a period, each device occupying a channel only needs to decide the transmission rate and target server.}
\label{figure:periodic_strategy}
\end{figure} 

\begin{algorithm}[ht]
	\label{alg}
	\caption{Online Task Offloading with Delay Guarantees Algorithm (TODG)}
	\label{TODG}
	\LinesNumbered
	\KwIn{$\xi^\mli{u}_{n_k}(0)$, $Q^\mli{u}_{n_k}(0)$, $Q^\mli{u}_{m,k}(0)$, $Z^\mli{u}_{n_k}(0)$, $Z^\mli{u}_{m,k}(0)$, $c_{n_k,l,m}(0)$ for all $k\in\mathcal{K}$, $n_k\in\mathcal{N}_k$, $l\in\mathcal{L}$, $m\in\mathcal{M}$}
	\KwOut{$\bm{s}(t)$, $\bm{z}(t)$, $\bm{d}(t)$ for all slot $t$}
	{Select $\{\zeta^\mli{u}_{n_k}\}$, $\{\zeta^\mli{s}_{m,k}\}$ and $\epsilon$ according to \eqref{eq:lem_d_4}, \eqref{eq:lem_d_3}, and \eqref{eq:epsilon_selection} respectively;}\\         
	\For{$t = 0$ \KwTo $T$}{
		\tcp{User side}
		\For{user $n_k\in\mathcal{N}_k$ for all $k\in\mathcal{K}$}{
			\If{$t\in\mathcal{T}_\delta$}
			{
				{Compute $w_{n_k,l}(t)$ between user $n_k$ and channel $l$ for all $l\in\mathcal{L}$ according to \eqref{eq:weight};}\\
				{Send $\{w_{n_k,l}(t)\}$ to one of the connectable servers and receive the channel allocation decision $z^*_{n_k,l}(t)$;}\\
			}
			{Compute the offloading decision $\{z_{n_k,l,m}(t)\}$ according to \eqref{eq:delta_periodic};}\\
			\If{ existing $z_{n_k,l',m'}(t)=1$}
			{Send $s_{n_k}(t)$ tasks to server $m'$ via channel $l'$ according to \eqref{eq:s_opt};}
			{Compute the amount of dropped tasks $d^\mli{u}_{n_k}(t)$ according to \eqref{eq:d_opt1_local}-\eqref{eq:d_opt2_local};}\\
			{Update the task queue $Q^\mli{u}_{n_k}(t+1)$ and  delay state queue $Z^\mli{u}_{n_k}(t+1)$ according to \eqref{eq:task_queue_local} and \eqref{eq:zqueue_local} respectively;}  \\
		}
		\tcp{Server side}
		\If{$t\in\mathcal{T}_\delta$}
		{
			Edge servers compute the optimal channel allocation decisions $\{z^*_{n_k,l}(t)\}$ in a distributed manner with the \emph{Multi-Robot Assignment} algorithm; 
		}
		\For{server $m\in\mathcal{M}$}{
			\For{VM $k\in\mathcal{K}$}{
				{Compute the amount of dropped tasks $d^\mli{s}_{m,k}(t)$ according to \eqref{eq:d_opt_server};}\\
				{Update the task queue $Q^\mli{s}_{m,k}(t+1)$ and delay state queue $Z^\mli{s}_{m,k}(t+1)$ according to \eqref{eq:task_queue_server} and \eqref{eq:zqueue_server} respectively;}\\
			}
		} 
	}	
\end{algorithm}

Combining the proposed solutions of the sub-problems, we summarize the details of TODG in Algorithm \ref{alg}.

\section{Performance Analysis}
\label{sec:analysis}
In this section, we analyze the performance of TODG. First, we establish the delay and system stability guarantees. Then, we characterize the optimality gap and study the impact of the system parameters on the performance. 

We provide the following lemma to show how the parameters affect the system queue lengths. 
\begin{lemma}
	\label{lemma:epsilon}
	Suppose Assumption \ref{assu:dmax} holds. Given $\epsilon\ge0$, TODG can achieve $Q^\mli{u}_{n_k}(t)\le \beta\epsilon + a^\mli{u,max}_{n_k}$ and $Q^\mli{s}_{m,k}(t)\le \beta\epsilon + c^\mli{max}L$ for each $k\in\mathcal{K}$, $n_k\in\mathcal{N}_k$, $m\in\mathcal{M}$, and $t\in\mathbb{N}$.
\end{lemma}
\begin{proof}
	\label{proof:epsilon}
	We first prove $Q^\mli{u}_{n_k}(t)\le \beta\epsilon + a^\mli{u,max}_{n_k}$ for all slots. It is easy to see that this holds for $t=0$ since $Q^\mli{u}_{n_k}(0)=0$ for all $ k\in\mathcal{K}$ and $n_k\in\mathcal{N}_k$. Suppose this is true for a particular time slot $t$. We show that it also holds for $t+1$. If $Q^\mli{u}_{n_k}(t)\le \beta\epsilon$, $Q^\mli{u}_{n_k}(t+1)\le \beta\epsilon + a^\mli{u,max}_{n_k}$ holds because it can increase by at most $a^\mli{u,max}_{n_k}$ in any time slot. If $Q^\mli{u}_{n_k}(t)> \beta\epsilon$, based on Lemma \ref{lemma:solution_drop_local} and Assumption \ref{assu:dmax}, $d^\mli{u}_{n_k}(t)\ge a^\mli{u}_{n_k}(t)$. Hence, queue $Q^\mli{u}_{n_k}(t)$ cannot increase in the next time slot, \ie, $Q^\mli{u}_{n_k}(t+1)\le Q^\mli{u}_{n_k}(t)$, thereby yielding the result.
	
	Based on \eqref{eq:d_opt_server}, similarly we have $Q^\mli{s}_{m,k}(t)\le \beta\epsilon + c^\mli{max}L$ for all slots. Thus, the proof is completed.
\end{proof}
Lemma \ref{lemma:epsilon} indicates that the task queue sizes are bounded by $\epsilon$, so the weight parameter can implicitly control task queues. Combined with Lemma \ref{lemma:delay}, we derive the following theorem.

\begin{theorem}
	\label{thm:stability}
	Suppose that $\zeta^\mli{u}_{n_k}$ and $\zeta^\mli{s}_{m,k}$ satisfy \eqref{eq:lem_d_3}-\eqref{eq:lem_d_4}, and the following holds
	\begin{align}
	\label{eq:epsilon_selection}
	\epsilon\le\min\{\epsilon^\mli{u},\epsilon^\mli{s}\}/\beta,
	\end{align}
	where $\epsilon^\mli{u}$ and $\epsilon^\mli{s}$ are denoted as
	\begin{align}
	\label{eq:epsilon_u}
	\epsilon^\mli{u}&\triangleq \min_{k,n_k}\left\{Q^\mli{u,max}_{n_k}-\max\{a^\mli{u,max}_{n_k},\zeta^\mli{u}_{n_k}\}\right\},\\
	\label{eq:epsilon_s}
	\epsilon^\mli{s}&\triangleq\min_{m,k}\left\{ Q^\mli{s,max}_{m,k}-\max\{c^\mli{max}L,\zeta^\mli{s}_{m,k}\}\right\}.
	\end{align}
	Given Assumption \ref{assu:dmax}, if for each $m\in\mathcal{M}$, $k\in\mathcal{K}$, and $n_k\in\mathcal{N}_k$, $\zeta^\mli{s}_{m,k}$ and $\zeta^\mli{u}_{n_k}$ satisfy \eqref{eq:lem_d_3} and \eqref{eq:lem_d_4} respectively, \eqref{eq:buffer_size} and \eqref{eq:delay_constr} hold during all slots.
\end{theorem}
\begin{proof}
	\label{proof:stability}
	Similar to the proof in Lemma \ref{lemma:epsilon}, it is easy to show that $Z^\mli{u}_{n_k}(t)\le \beta\epsilon + \zeta^\mli{u}_{n_k}$ and $Z^\mli{s}_{m,k}(t)\le \beta\epsilon + \zeta^\mli{s}_{m,k}$ for each $k\in\mathcal{K}$, $n_k\in\mathcal{N}_k$, $m\in\mathcal{M}$, and $t\in\mathbb{N}$. Based on \eqref{eq:epsilon_u} and \eqref{eq:epsilon_s}, we have
	\begin{align}
	Q^\mli{u}_{n_k}(t)&\le Q^\mli{u,max}_{n_k},\quad Q^\mli{s}_{m,k}(t)\le Q^\mli{s,max}_{m,k},\\
	Z^\mli{u}_{n_k}(t)&\le Q^\mli{u,max}_{n_k},\quad Z^\mli{s}_{m,k}(t)\le Q^\mli{s,max}_{m,k}.
	\end{align}
	Combined with Lemma \ref{lemma:delay}, we finish the proof.
\end{proof}
Theorem \ref{thm:stability} demonstrates that the buffer size constraints of and the task response time (\ie, \eqref{eq:buffer_size} and \eqref{eq:delay_constr}) can be well satisfied by appropriately selecting the values of parameters $\zeta^\mli{s}_{m,k}$, $\zeta^\mli{u}_{n_k}$, and $\epsilon$ for TODG. It is  beneficial to implicitly handle the constraints on the worst-case delay, which is hard to satisfy via only slot-level decisions. 

To analyze the optimality gap for TODG, we provide the following lemma to characterize the error induced by the $\delta$-periodic strategy.


\begin{lemma}
	\label{lemma:periodic_gap}
	For any $p,q\in\mathbb{N}$ and $q<\delta$, under the $\delta$-periodic strategy, the solution found by TODG in slot $t$ satisfies that
	\begin{align}
	\label{eq:lem_pd_1}
	\hat{D}(t)-\hat{D}^*(t)\le qG,~\text{if}~t=p\delta+q,
	\end{align}
	where $\hat{D}^*(t)$ denotes the optimal value of dual problem \eqref{problem_dual}, and $G\triangleq 2c^\mli{max}\cdot\max_{k,n_k}\{Q^\mli{u,max}_{n_k}\}\cdot\min\{N,L\}$.
\end{lemma}
\begin{proof}
	\label{proof:periodic_gap}
	Because the optimal solutions of sub-problems \eqref{subpro_drop_local} and \eqref{subpro_drop_server} can be found, the following holds
	\begin{align}
	\label{eq:pg_1}
	&\hat{D}(t)-\hat{D}^*(t)\nonumber\\
	=&\sum_{m\in\mathcal{M}}\sum_{k\in\mathcal{K}}\sum_{n_k\in\mathcal{N}_k}\sum_{l\in\mathcal{L}}\left(Q^\mli{u}_{n_k}(t)+Z^\mli{u}_{n_k}(t)-Q^\mli{s}_{m,k}(t)\right)\nonumber\\
	&\cdot\min\left\{\xi_{n_k}(t),c_{n_k,l,m}(t)\right\}\cdot \left(z^*_{n_k,l,m}(t)-z_{n_k,l,m}(t)\right).
	\end{align}
	If $q=0$, it is easy to see that $\hat{D}(t+1)-\hat{D}^*(t+1)=0$. Based on \eqref{eq:delta_periodic}, $\vert z^*_{n_k,l,m}(t)-z_{n_k,l,m}(t)\vert=1$ if and only if $Q^\mli{u}_{n_k}(t)+Z^\mli{u}_{n_k}(t)-Q^\mli{s}_{m,k}(t)>0$. Let $\hat{i}_l\in\mathcal{N}$ and $i^*_l\in\mathcal{N}$ denote the device that is assigned with channel $l$ by the $\delta$-periodic and the optimal strategy in slot $t$ respectively, \ie, $z_{\hat{i}_l,l,m_{i,l}}(t)=1$ and $z^*_{i^*_l,l,m_{i,l}}(t)=1$ ($m_{i,l}$ is defined in Lemma \ref{lemma:bipartite}). Then, for $q>0$, \eqref{eq:pg_1} can be rewritten as 
	\begin{align}
	\label{eq:pg_2}
	&\hat{D}(t)-\hat{D}^*(t)\nonumber\\
	\le&\sum_{l\in\mathcal{L}} \phi_{i^*_l,l}(t)-\phi_{\hat{i}_l,l}(t)\nonumber\\
	\le& 2q c^\mli{max}\cdot\max_{k,n_k}\{Q^\mli{u,max}_{n_k}\}\cdot\min\{N,L\},
	\end{align}
	thereby \eqref{eq:lem_pd_1} holds. 
\end{proof}
Lemma \ref{lemma:periodic_gap} shows that the cumulative error in the decomposed sub-problems would linearly increase with the computation period $\delta$. We note that the gap $G$ in Lemma \ref{lemma:periodic_gap} is derived in a rare worst case, \ie, all the transmission rates of selected user devices suddenly become zero while the other devices enjoy the maximum channel capacity in the current slot. However, the channel and task queue states commonly do not fluctuate so sharply between adjacent slots. Thus, it is reasonable to expect that the error incurred by the $\delta$-periodic strategy is often much smaller than the theoretical gap. 

Based on Lemma \ref{lemma:periodic_gap}, we are ready to establish the optimality gap for TODG.

\begin{theorem}
	\label{thm:optimal_gap}
	Let $\bar{U}^\mli{opt}$ and $\bar{U}$ denote the objective values in \eqref{problem_tran} corresponding to the optimal and our solutions respectively. As in Theorem \ref{thm:stability}, suppose Assumption \ref{assu:dmax}, and $\epsilon\le\min\{\epsilon^\mli{u},\epsilon^\mli{s}\}/\beta$ are satisfied and all the stochastic variables are independent and identically distributed (i.i.d) over time slots, then the following holds
	\begin{align}
	\label{eq:utilitybound}
	\bar{U}^\mli{opt}-\bar{U}\le \frac{C+(\delta-1)G/2}{\epsilon},
	\end{align}
	where $C$ is defined in \eqref{eq:C}.
\end{theorem}
\begin{proof}
	\label{proof:optimal_gap}
	Based on \cite[Theorem 5.1]{neely2010stochastic}, it is easy to show that for any fixed $\sigma>0$, there exists a \emph{stationary and randomized policy} that can choose feasible control actions $\bm{\tilde{z}}(t)$, $\bm{\tilde{s}}(t)$, and $\bm{\tilde{d}}(t)$ independent of current queue backlogs in each slot $t$, and satisfy that
	\begin{align}
	\label{eq:utility_1}
	&\sum_{k\in\mathcal{K}}\sum_{n_k\in\mathcal{N}_k}\epsilon\cdot\mathbb{E}\left\{g_{n_k}\left(a^\mli{u}_{n_k}(t)-\tilde{d}^\mli{u}_{n_k}(t)\right)\right\}\nonumber\\
	&- \sum_{m\in\mathcal{M}}\sum_{k\in\mathcal{K}} \beta\epsilon\cdot\mathbb{E}\big\{\tilde{d}^\mli{s}_{m,k}(t)\big\}\ge \bar{U}^\mli{opt}-\sigma,\\
	&\mathbb{E}\big\{\tilde{a}^\mli{u}_{n_k}(t)\big\}\le \mathbb{E}\big\{\tilde{s}_{n_k}(t)\big\}+\mathbb{E}\big\{\tilde{d}^\mli{u}_{n_k}(t)\big\}+\sigma,\\
	&\mathbb{E}\big\{\zeta^\mli{u}_{n_k}(t)\big\}\le \mathbb{E}\big\{\tilde{s}_{n_k}(t)\big\}+\mathbb{E}\big\{\tilde{d}^\mli{u}_{n_k}(t)\big\}+\sigma,\\
	&\mathbb{E}\big\{\tilde{a}^\mli{s}_{m,k}(t)\big\}\le \mathbb{E}\big\{\tilde{u}_{m,k}(t)\big\} + \mathbb{E}\big\{\tilde{d}^\mli{s}_{m,k}(t)\big\}+\sigma,\\
	\label{eq:utility_2}
	&\mathbb{E}\big\{\zeta^\mli{s}_{m,k}(t)\big\}\le \mathbb{E}\big\{\tilde{u}_{m,k}(t)\big\} + \mathbb{E}\big\{\tilde{d}^\mli{s}_{m,k}(t)\big\}+\sigma,
	\end{align}
	for each $k\in\mathcal{K}$, $n_k\in\mathcal{N}_k$, and $m\in\mathcal{M}$. It should be note that in \cite[Theorem 5.1]{neely2010stochastic}, they define a set of \emph{auxiliary queues} to obtain \eqref{eq:utility_1}. However, because the feasible region of $d^\mli{u}_{n_k}(t)$ is deterministic, we can directly derive \eqref{eq:utility_1} via the stationary and randomized policy without additional virtual queues. Then, based on Lemma \ref{lemma:periodic_gap}, combining \eqref{eq:drift_penality}, \eqref{eq:drift_bound} and \eqref{eq:utility_1}-\eqref{eq:utility_2}, and taking $\delta\rightarrow\infty$, we have
	\begin{align}
	\label{eq:utility_3}
	&\mathbb{E}\Big\{L(t+1)-L(t)-\epsilon\cdot\Big(\sum_{k\in\mathcal{K}}\sum_{n_k\in\mathcal{N}_k}g_{n_k}\left(a^\mli{u}_{n_k}(t)-d^\mli{u}_{n_k}(t)\right)\nonumber\\
	&-\sum_{m\in\mathcal{M}}\sum_{k\in\mathcal{K}}\beta d^\mli{s}_{m,k}(t)\Big)\Big\vert\bm{\Theta}(t)\Big\}\le C+q G-\epsilon\cdot \bar{U}^\mli{opt},
	\end{align}
	for $t=p\delta + q$ as in Lemma \ref{lemma:periodic_gap}, where $\bm{\Theta}(t)$ denotes the current queue states. Taking expectations over $\bm{\Theta}(t)$ on both sides of \eqref{eq:utility_3} and summing over $t\in\{0,\dots,p\delta-1\}$ yield 
	\begin{align}
	\label{eq:utility_4}
	&\mathbb{E}\left\{L(p\delta-1)-L(0)\right\}\nonumber\\
	&-\sum_{t=0}^{p\delta-1}\sum_{k\in\mathcal{K}}\sum_{n_k\in\mathcal{N}_k}\epsilon\cdot \mathbb{E}\left\{g_{n_k}\left(a^\mli{u}_{n_k}(t)-d^\mli{u}_{n_k}(t)\right)\right\}\nonumber\\
	&+\sum^{p\delta-1}_{t=0}\sum_{m\in\mathcal{M}}\sum_{k\in\mathcal{K}}\beta\epsilon\cdot\mathbb{E}\left\{d^\mli{s}_{m,k}(t)\right\}\nonumber\\
	&\le (C-\epsilon\cdot \bar{U}^\mli{opt})p\delta+(\delta-1)\delta pG/2.
	\end{align}
	Using the fact $L(0)=0$ and rearranging the terms, \eqref{eq:utility_4} can be written as
	\begin{align}
	\label{eq:utility_5}
	&\frac{1}{p\delta}\sum_{t=0}^{p\delta-1}\epsilon\bigg(\sum_{k\in\mathcal{K}}\sum_{n_k\in\mathcal{N}_k} \mathbb{E}\left\{g_{n_k}\left(a^\mli{u}_{n_k}(t)-d^\mli{u}_{n_k}(t)\right)\right\}\nonumber\\
	&-\sum_{m\in\mathcal{M}}\sum_{k\in\mathcal{K}}\beta\mathbb{E}\left\{d^\mli{s}_{m,k}(t)\right\}\bigg)\ge \epsilon\cdot \bar{U}^\mli{opt}-\Big(C+\frac{(\delta-1)G}{2}\Big).
	\end{align}
	Taking a limit as $q\rightarrow\infty$, and using Jensen's inequality, we have the following
	\begin{align}
	&\sum_{k\in\mathcal{K}}\sum_{n_k\in\mathcal{N}_k}g_{n_k}\left(\bar{a}^\mli{u}_{n_k}-\bar{d}^\mli{u}_{n_k}\right)\nonumber\\
	&-\beta\sum_{m\in\mathcal{M}}\sum_{k\in\mathcal{K}}\bar{d}^\mli{s}_{m,k}\ge \bar{U}^\mli{opt}-\frac{C+(\delta-1)G/2}{\epsilon},
	\end{align}
	thereby completing the proof.
\end{proof}
Theorem \ref{thm:optimal_gap} shows that the achievable system utility of TODG has a controllable gap regarding the period $\delta$ and the weight parameter $\epsilon$ between the optimal network utility. In particular, although a relatively large period $\delta$ can mitigate the high computational cost caused by the task scheduling and channel allocation, it may also result in performance degradation to the system utility. Further, combined with Lemma \ref{lemma:delay} and \ref{lemma:epsilon}, the weight parameter $\epsilon$ achieves a trade-off between latency and utility, \ie, a larger $\epsilon$ provides a lower optimality gap, but may increase the task queue sizes, thereby leading to a poor response time in average. Note that even though we derive the result of Theorem \ref{thm:optimal_gap} in the i.i.d. case, it can be generalized to the non-ergodic case via \cite[Theorem 4.13]{neely2010stochastic}, which is outside the scope of this paper.

\begin{remark}
	Theorem \ref{thm:optimal_gap} also quantifies the impact of the task buffer sizes and the delay requirements on system utility. It is not difficult to see that
	\begin{align}
	\label{eq:epsilon_u_1}
	\epsilon^\mli{u}&>\underbrace{\min_{k,n_k}\left\{\left(1-\frac{2}{\tau^\mli{max}_{n_k}}\right)Q^\mli{u,max}_{n_k}\right\}}_\mathbf{(a)},\\
	\label{eq:epsilon_s_1}
	\epsilon^\mli{s}&>\underbrace{\min_{m,k}\left\{\left(1-\mathop{\max}_{n_k}\left\{\frac{2}{\tau^\mli{max}_{n_k}}\right\}\right)Q^\mli{s,max}_{m,k}\right\}}_\mathbf{(b)}.
	\end{align}
	Due to the fact that $\tau^\mli{max}_{n_k}\ge 2$ for all $k\in\mathcal{K}$ and $n_k\in\mathcal{N}_k$, the following holds
	\begin{align}
	\label{eq:utility_gap}
	U^\mli{opt}-U^*<\frac{C+(\delta-1)G/2}{\min\{\mathbf{(a)},\mathbf{(b)}\}},
	\end{align}
	which implies that large task buffers may improve system utility, while small delay requirements may cause the opposite. 
\end{remark}

\section{Performance Evaluation}
\label{sec:simulations}

In this section, we provide further insights into the algorithm performance via extensive simulation experiments. 

\textbf{Simulation setup.} We consider an edge computing system consisting of $M=3$ edge servers. Each edge server creates $K=3$ virtual machines (VMs), \ie, each server can serve $K=3$ heterogeneous types of computing requests (\eg, image processing, data compression, and mathematical calculation). The number of total user devices $N$ is an integer from $[3,100]$, with that of each type $N_k$ being the same. We set the number of channels $L$ as 20. We set the length of a time slot as 1 second. The transmission rate $c_{n_k,l,m}(t)$ Mbps of channel $l$ between device $n_k$ and edge server $m$ in slot $t$ is uniformly distributed between $[0,1]$. Task arrival rate $a^\mli{u}_{n_k}(t)$ Mbps of $n_k$ across time slots is drawn from continuous uniform distributions. Different types of devices have varying statistic characteristics, \ie, the arrival rates of the three types of tasks follow $U(0.8,1)$, $U(0.4,0.6)$, and $U(0,0.4)$, respectively. Note that the proposed algorithm can offer an efficient solution without any prior knowledge about the stochastic processes. that is, these random variables can have other different stochastic characteristics without affecting the performance of our algorithm. We let the processing rate $u_{m,k}(t)$ of VM $v_{m,k}$ follow $U(0,3)$. The parameters' details are summarized in Table \ref{table:parameters}.

\begin{table}[htpb]
	\caption{Parameters in Simulation}
	\label{table:parameters}
	\centering
	\renewcommand\arraystretch{1.25}
	\resizebox{\columnwidth}{!}{
	\begin{tabular}{l|l}
		\hline
		\textbf{Parameter}                          & \textbf{Value}                                                                                                                        \\ \hline
		\# edge servers ($M$)                       & an integer varying between $[1,6]$                                                                                                    \\
		\# task types ($K$)                         & 3                                                                                                                                     \\
		\# devices ($N$)                            & an integer varying between $[3,100]$                                                                                                  \\
		transmission rates ($c_{n_k,l,m}(t)$)        & real numbers (Mbit/slot) varying between $[0,1]$                                                                                          \\
		task arrival rates ($a^\textit{u}_{n_k}(t)$) & \begin{tabular}[c]{@{}l@{}}real numbers (Mbit/slot) following $U(0.8,1)$, $U(0.4,0.6)$, \\ and $U(0,0.4)$ for types $k=1,2,3$, respectively\end{tabular} \\
		processing rates ($u_{m,k}(t)$)              & 
		\begin{tabular}[c]{@{}l@{}}real numbers (Mbit/slot) following $U(0.5,1.5)$, $U(1,2)$, \\ and $U(1.5,2.5)$ for types $k=1,2,3$, respectively\end{tabular}                                                                                                     \\ \hline
	\end{tabular}
	}
\end{table}

\textbf{Baselines.} We compare our proposed algorithm (TODG) with a centralized stochastic control algorithm (SCA) proposed in \cite{fang2016stochastic} for heterogeneous task offloading. We also consider a greedy control algorithm (GA), which selects $\min\{N,L\}$ user devices in each slot with the shortest task queues. Each user will send as many tasks as possible to the edge server with a minimal task backlog. Due to the lack of channel allocation mechanisms in SCA and GA, we set that they randomly assign channels for the selected user devices.

\textbf{Implementation.} We implement the code in MATLAB R2019b on a server with two Intel$^\circledR$ Xeon$^\circledR$ Golden 5120 CPUs and one Nvidia$^\circledR$ Tesla-V100 32G GPU.

\textbf{System utility and delay under different $\bm{\delta}$ and $\bm{\epsilon}$.} To validate the achievable system utility as demonstrated in \eqref{eq:utility_gap}, we fix $N=45$ and run the experiments under different periods $\delta$ and weight parameters $\epsilon$. We repeat the experiments ten times and plot the results in Fig. \ref{fig:epsilon_delta_utility}. It can be seen from Fig. \ref{subfig:epsilon_utility} that the achieved system utility by TODG becomes larger as $\epsilon$ increases. More specifically, the system utility increases sharply with the increase of $\epsilon$ at the beginning, and then the increasing speed decreases when $\epsilon$ gets large. The underlying rationale is that the lower bound of the system utility is a concave function with respect to $\epsilon$, as shown in \eqref{eq:utilitybound}. However, Fig. \ref{subfig:epsilon_delay} also illustrates that the response delay increases linearly with $\epsilon$, so the weight parameter $\epsilon$ enables a trade-off between the utility and latency. On the other hand, we study the impact of period $\delta$ on performance. Combining Fig. \ref{subfig:epsilon_utility} and Fig. \ref{subfig:delta_utility}, it is easy to see large $\delta$ would lead to a utility degradation, especially for a small $\epsilon$ value, which indicates the correctness of the analytical results in Theorem \ref{thm:optimal_gap}. Besides, Fig. \ref{subfig:delta_utility} and Fig. \ref{subfig:delta_delay} also imply that the decreasing rate of the utility induced by $\delta$-periodic strategy is relatively low to $\delta$, while the delay can remain stable as $\delta$ increases. For example, even though we make the optimal offloading decisions every 15 slots, the utility only decreases to 9.46, and the average maximal delay is 77.2. By comparison, the maximal utility and delay achieved by SCA are 9.28 and 223 in this setting, as shown in Fig. \ref{fig:comparason}. Therefore, the $\delta$-periodic offloading strategy is expected to reach a high system utility with delay guarantees and inexpensive computational cost.

\begin{figure}
	\centering
	\subfigure[Impact of $\epsilon$ on utility. ]{\label{subfig:epsilon_utility}\includegraphics[width=0.225\textwidth]{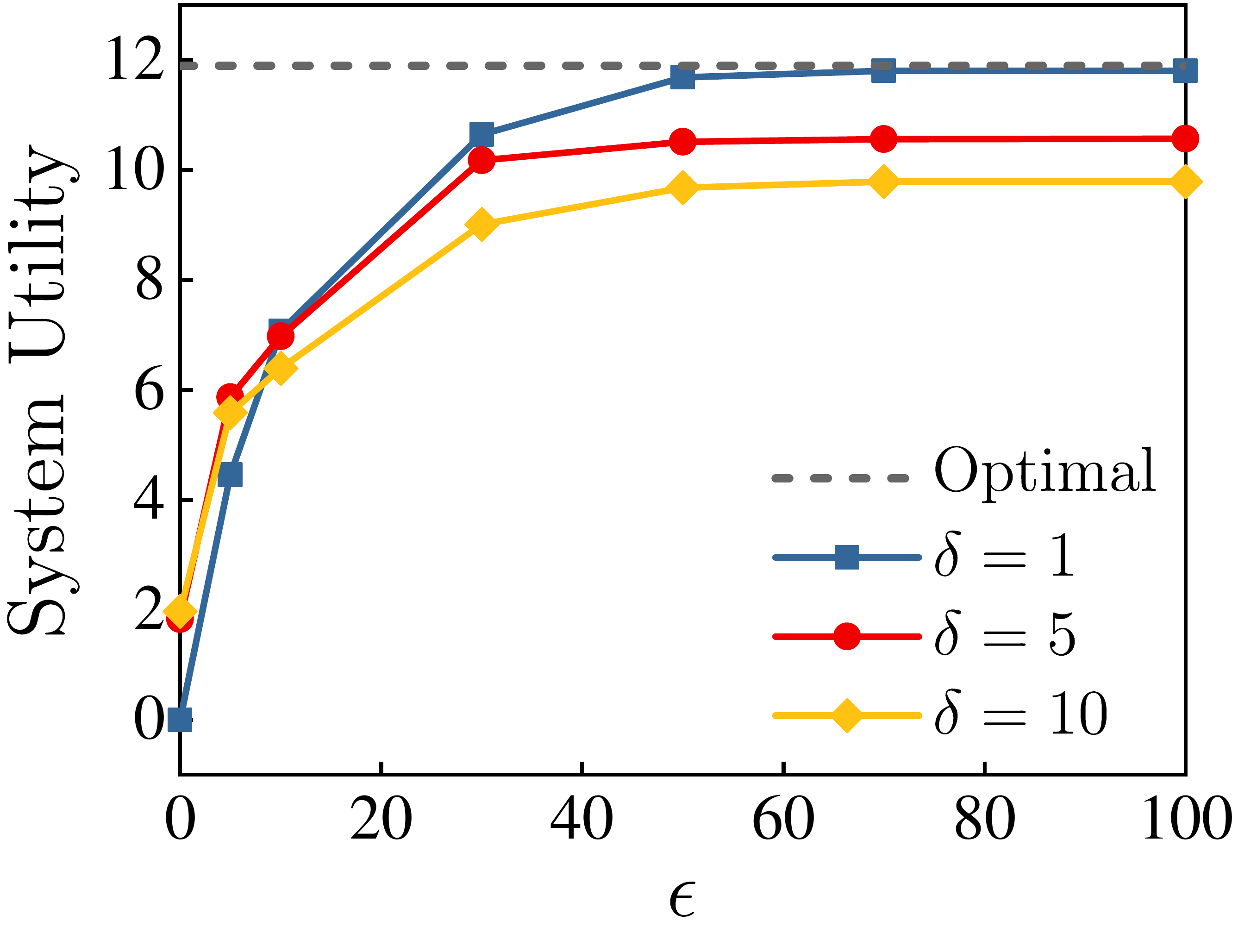}}
	\subfigure[Impact of $\delta$ on utility.]{\label{subfig:delta_utility}\includegraphics[width=0.225\textwidth]{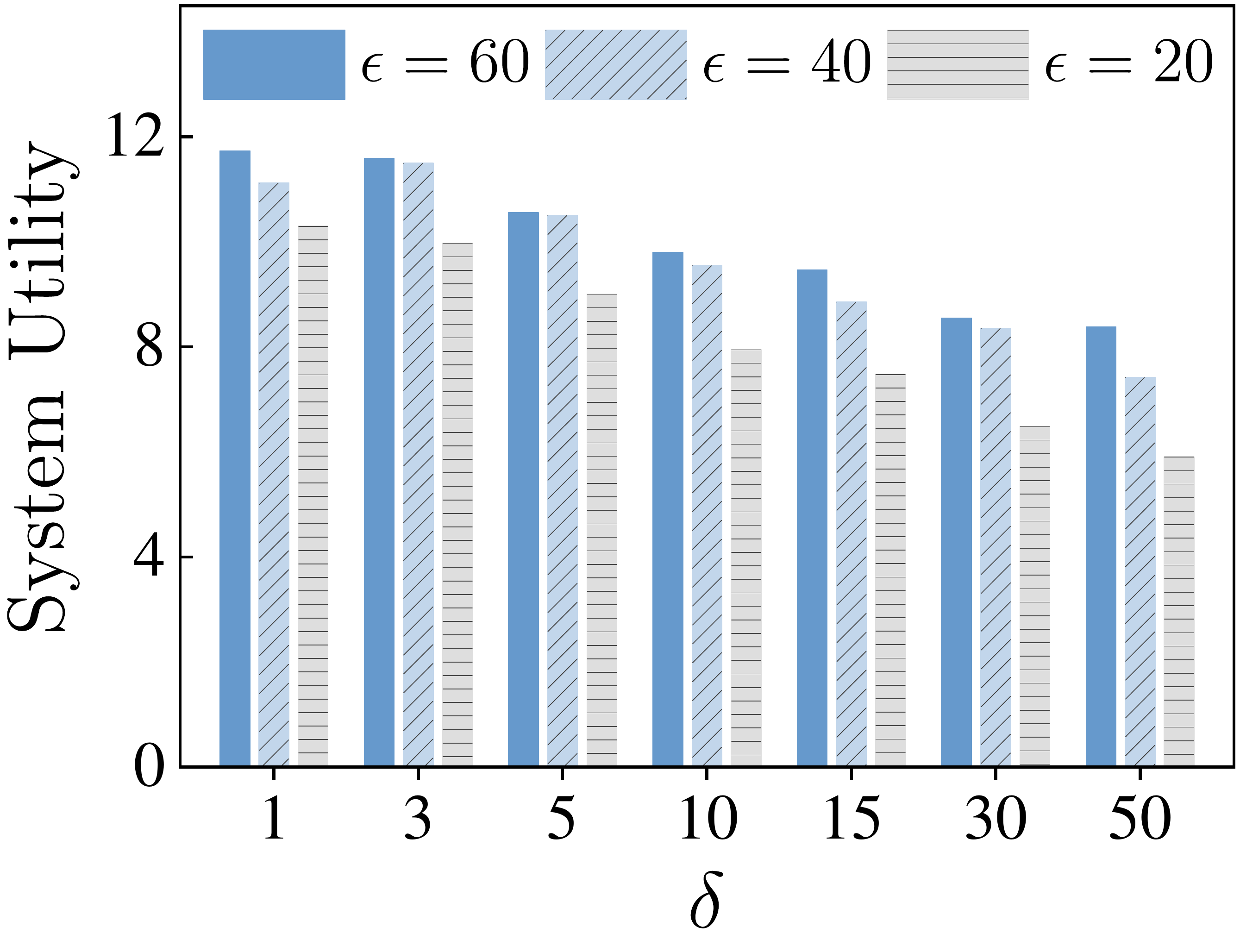}}
	
	\subfigure[Impact of $\epsilon$ on delay. ]{\label{subfig:epsilon_delay}\includegraphics[width=0.225\textwidth]{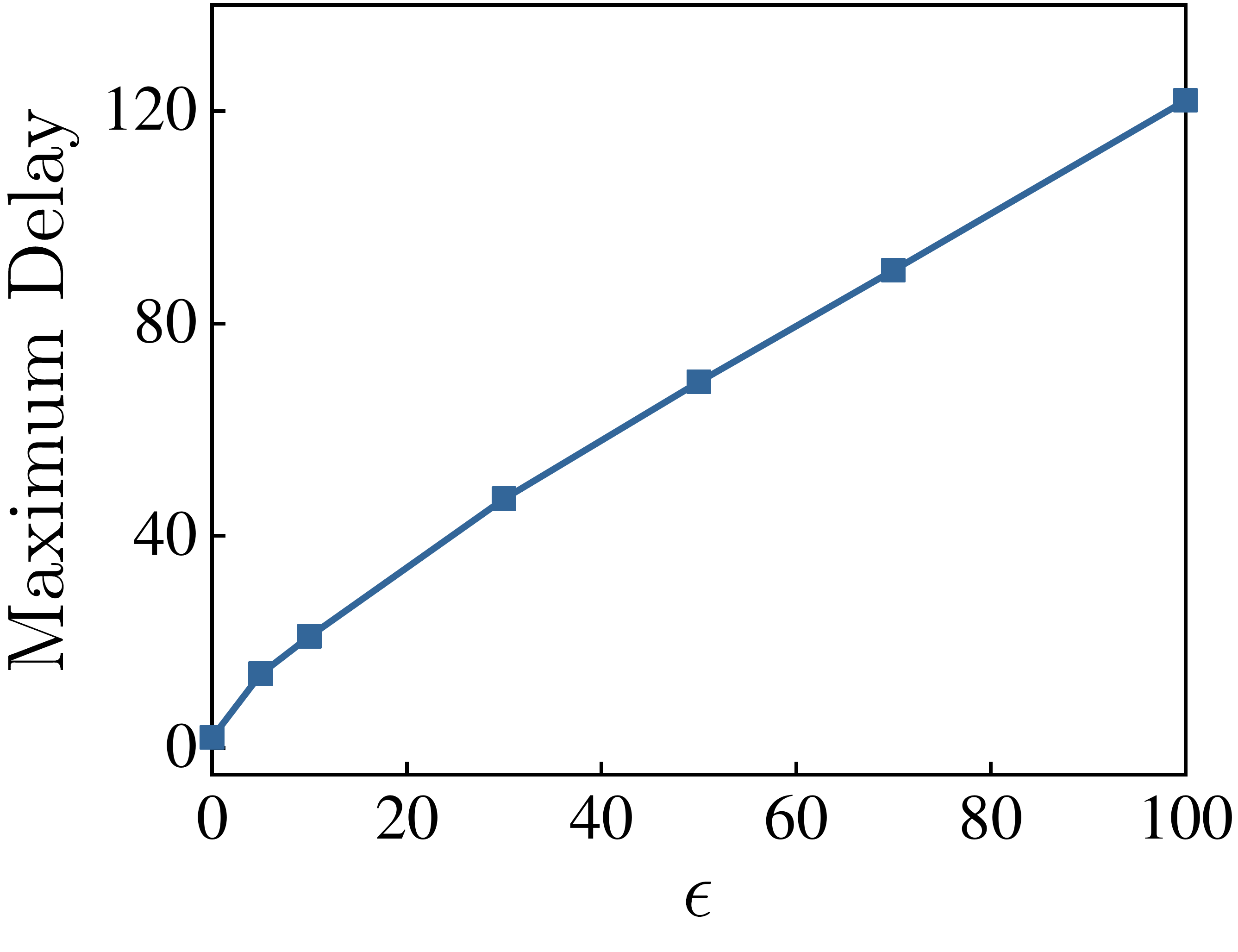}}
	\subfigure[Impact of $\delta$ on delay.]{\label{subfig:delta_delay}\includegraphics[width=0.225\textwidth]{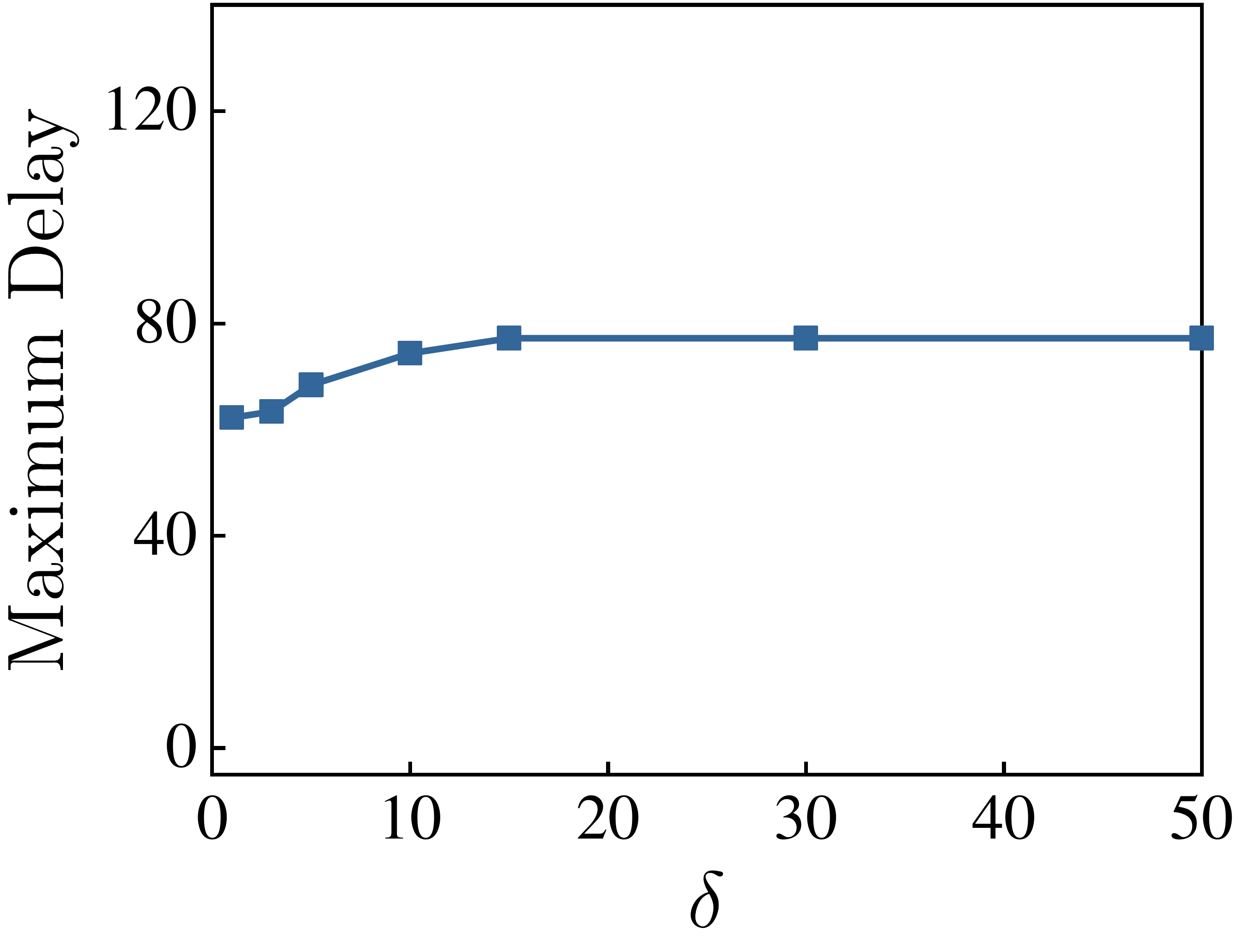}}
	\caption{System utility and delay under different $\delta$ and $\epsilon$.}
	\label{fig:epsilon_delta_utility}
\end{figure}

\begin{figure}
	\centering
	\subfigure[User task queue backlogs. ]{\label{subfig:user_queue}\includegraphics[width=0.225\textwidth]{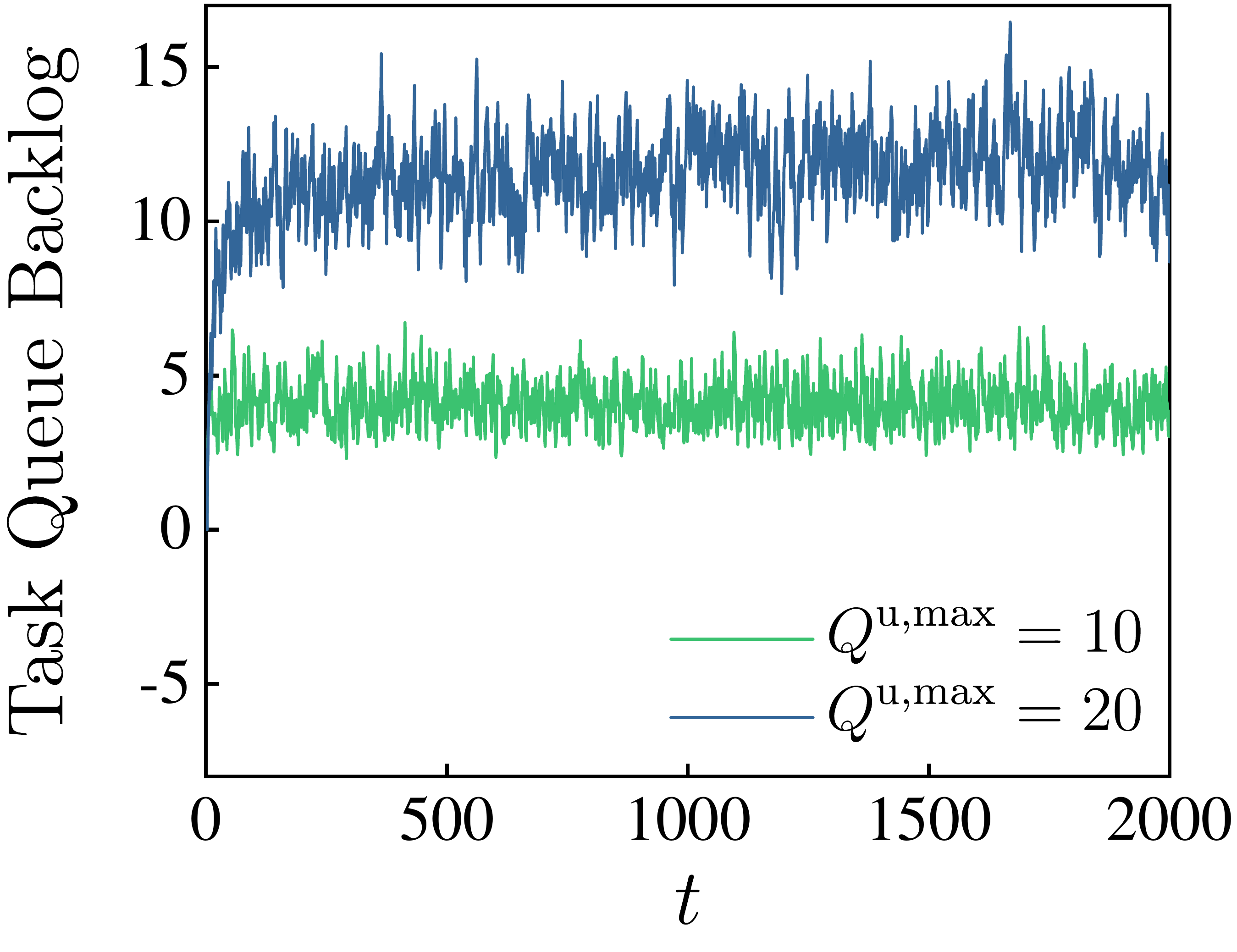}}
	\subfigure[Server task queue backlogs.]{\label{subfig:server_queue}\includegraphics[width=0.225\textwidth]{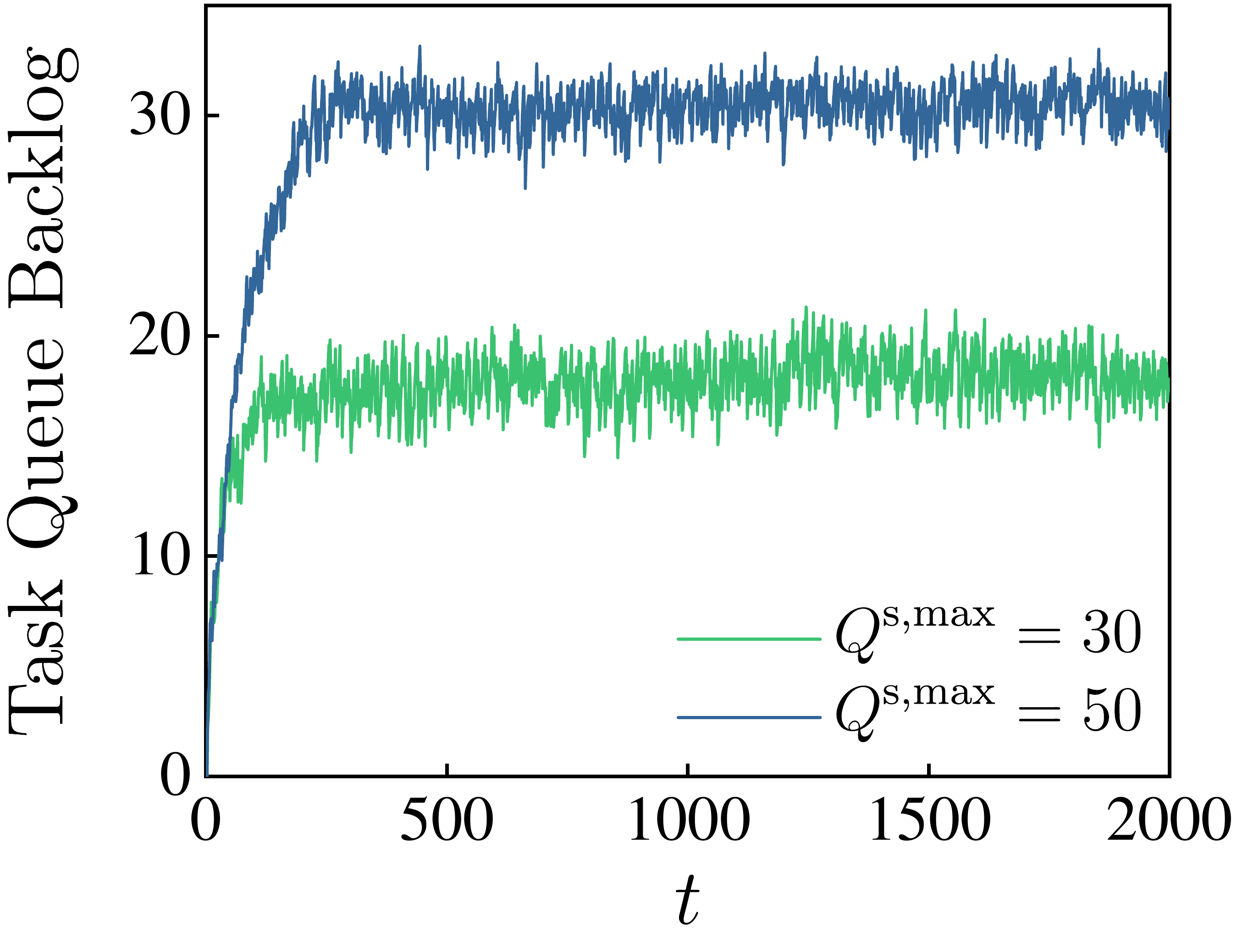}}
	\subfigure[System utility under different delay requirements.]{\label{subfig:tradeoff}\includegraphics[width=0.225\textwidth]{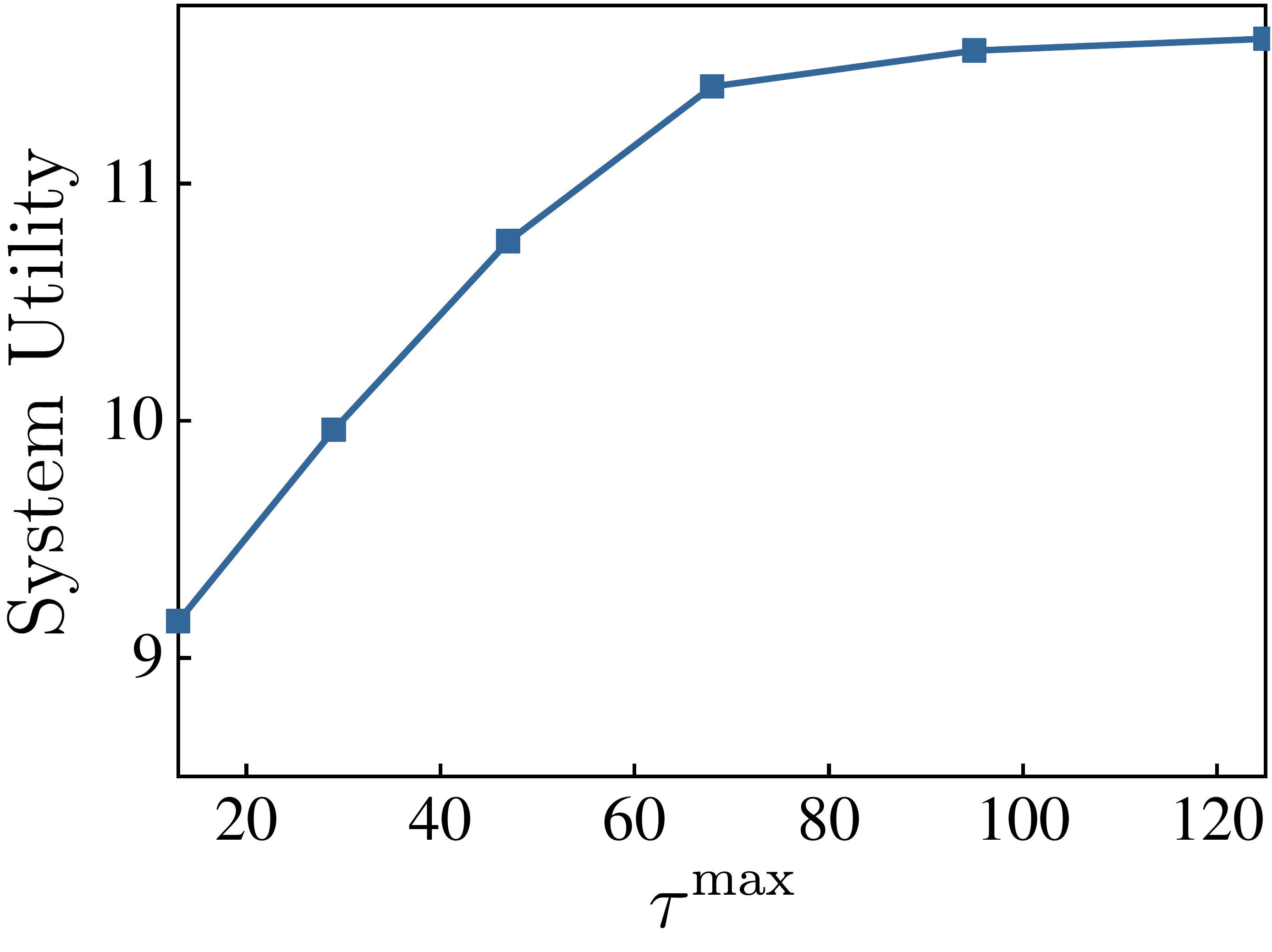}}
	\subfigure[Impact of $\zeta$ on utility and delay.]{\label{subfig:zeta}\includegraphics[width=0.225\textwidth]{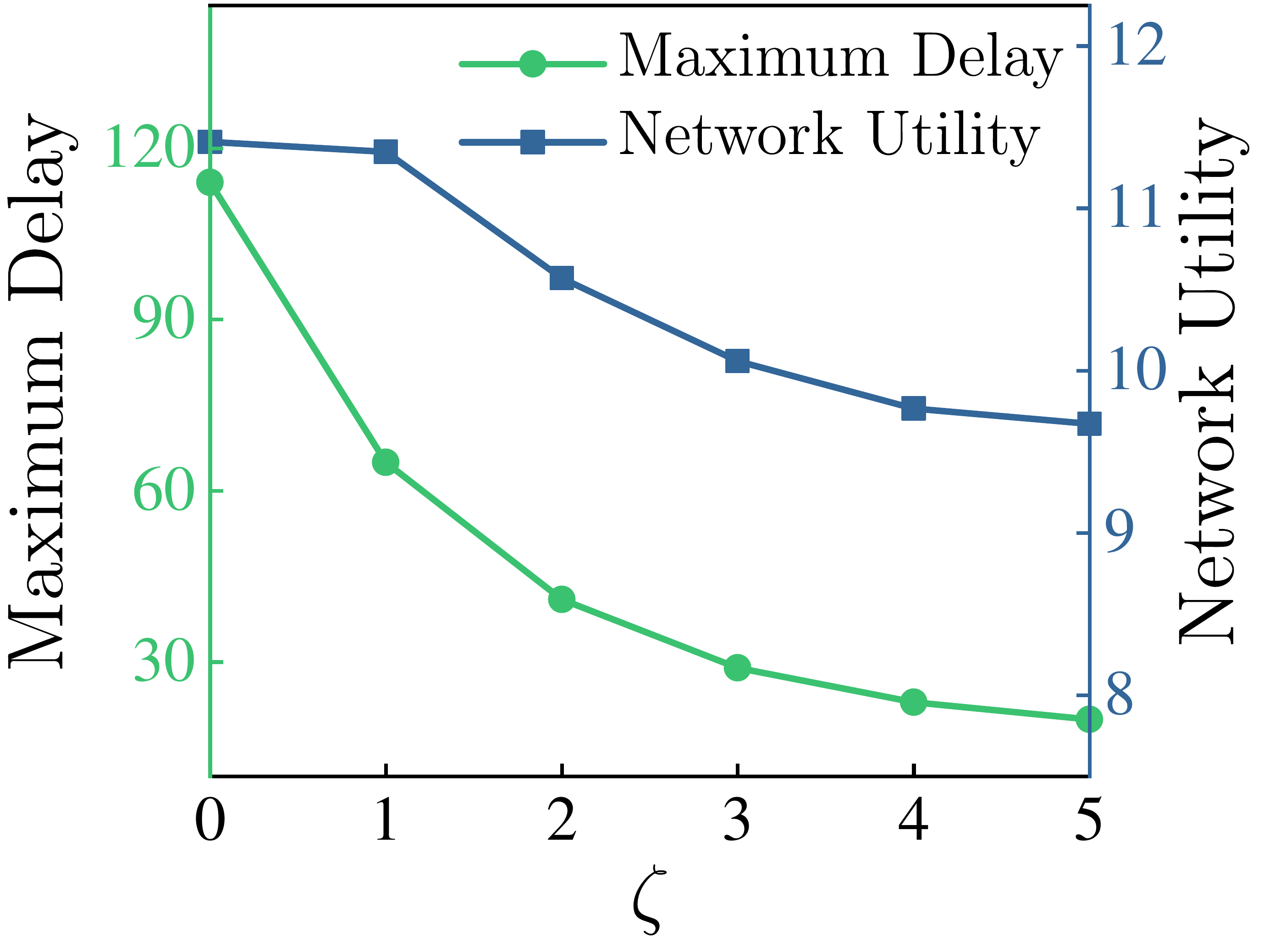}}
	\caption{System stability and delay requirements.}
	\label{fig:stability_delay}
\end{figure}

\textbf{System stability and delay requirements.} Fig. \ref{subfig:user_queue} and Fig. \ref{subfig:server_queue} show the dynamics of the task queues on user devices and edge servers under different task buffer sizes over 2000 slots. As illustrated in \ref{subfig:user_queue} and Fig. \ref{subfig:server_queue}, under fixed buffers, TODG can well guarantee the system stability, and the task backlogs are kept at a low level to provide a short response latency. We also evaluate the impact of the delay constraints on the system utility. For ease of exposition, we set the delay constraints to be the same among user devices, \ie, $\tau^\mli{max}_{n}=\tau^\mli{max}$ for all $n\in\mathcal{N}$. It can be easily seen from Fig. \ref{subfig:tradeoff}, strict delay requirements would lead to a utility degradation. The reason is that if the delay requirements outweigh the scheduling and processing capability of the system, it will drop the outdated tasks, which will exacerbate the degradation when the delay constraints are tight. In addition, we study the impact of the penalty parameter $\zeta$ on performance. Similarly, we set all $\zeta^\mli{u}_n=\zeta^\mli{s}_{m,k}=\zeta$ for convenience. As shown in Fig. \ref{subfig:zeta}, with the increase of $\zeta$, we obtain a lower response delay but also lead to a worse system utility, which validates the results of Lemma \ref{lemma:delay} and Remark \ref{remark:zeta}. Meanwhile, it is not difficult to see that the role of $\zeta$ is opposite to $\epsilon$. Actually, their relationship has been given by Theorem \ref{thm:stability} and \ref{thm:optimal_gap}.

\begin{figure}[ht]
	\centering
	\subfigure[Comparison of system utility.]{\label{subfig:comp_utility}\includegraphics[width=0.225\textwidth]{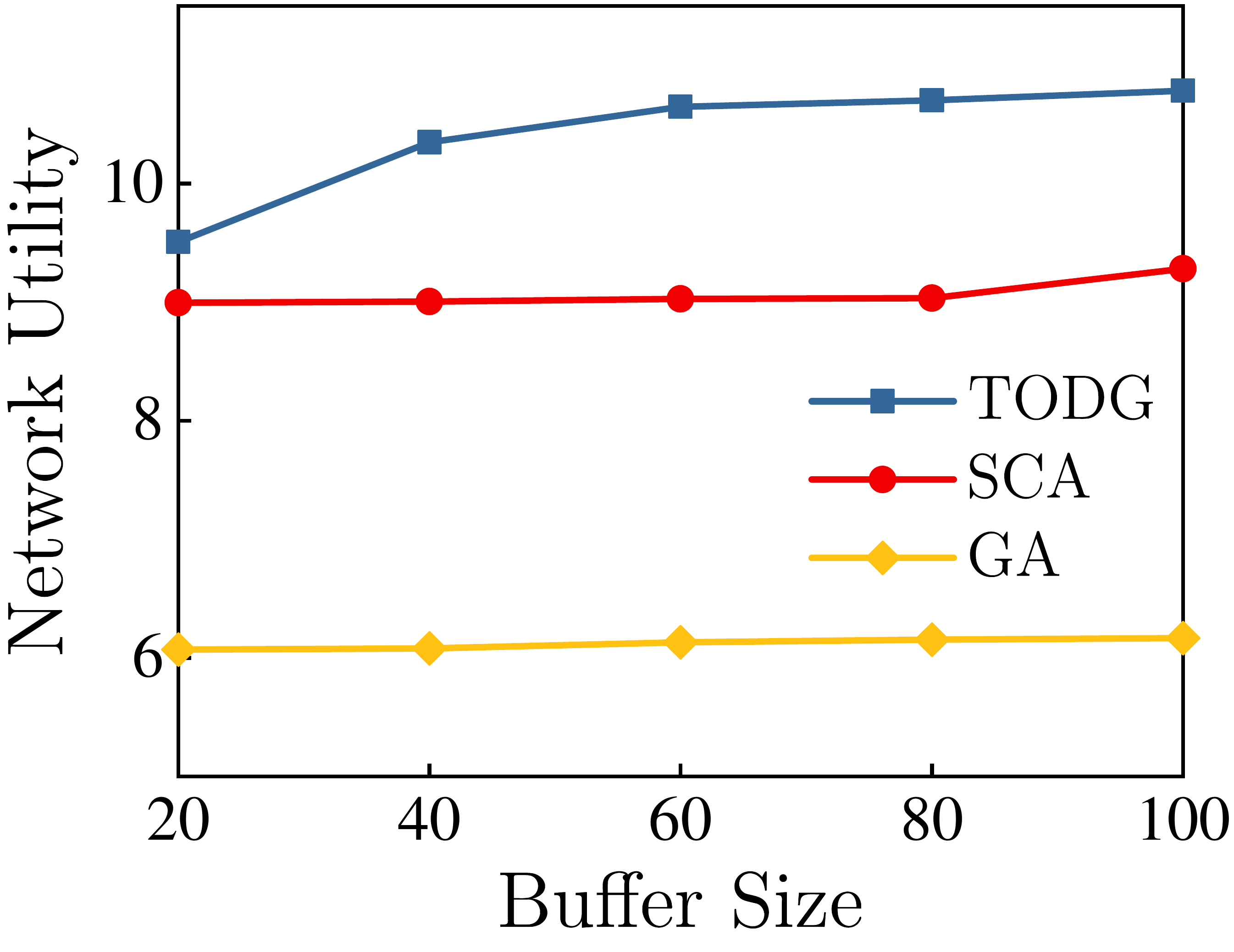}}
	\subfigure[Comparison of delay.]{\label{subfig:comp_delay}\includegraphics[width=0.225\textwidth]{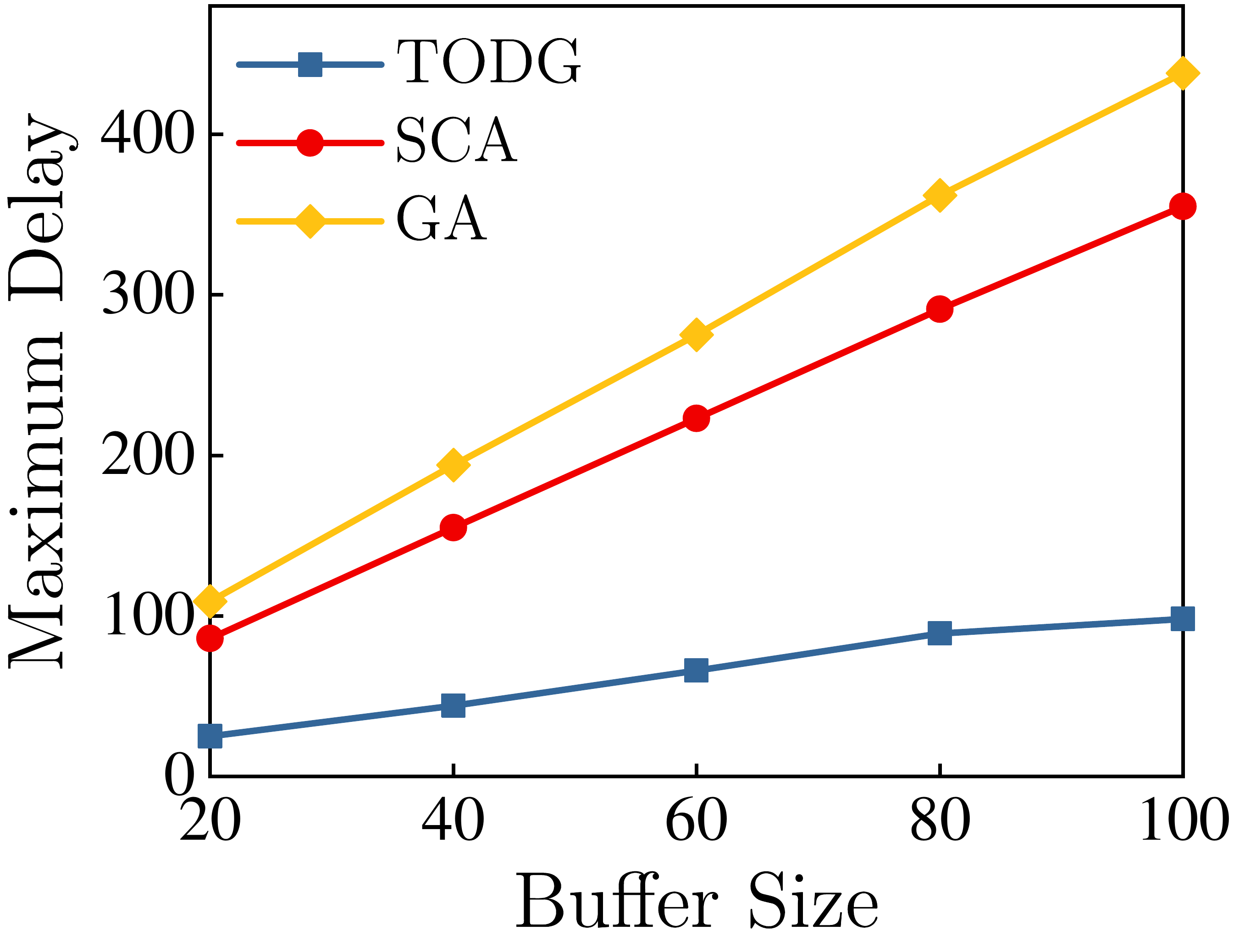}}
	\caption{Comparison of utility and delay among algorithms under different buffer sizes.}
	\label{fig:comparason}
\end{figure}

\textbf{Comparison of performance among different algorithms.} In order to compare the performance among TODG and baseline algorithms, we vary the task buffer sizes on both user devices and edge servers and show the corresponding system utility and task maximal response delay in Fig. \ref{fig:comparason}. As illustrated in Fig. \ref{fig:comparason}, TODG outperforms SCA and GA with fixed buffer sizes while significantly reducing the response delay. The reason is that TODG enables more effectively exploiting the stochastic features of communication resources and computational capabilities on edge servers. Meanwhile, by jointly scheduling different types of tasks, TODG will prioritize the tasks with high delay requirements. Moreover, it can be seen from Fig. \ref{subfig:comp_utility} that large buffer sizes can bring a performance increase of TODG since it offers more flexibility for task scheduling. However, due to the lack of effective scheduling mechanisms to reduce response delay, SCA and GA can only achieve limited system utility with larger task buffers, resulting in worse latency caused by the longer queuing time.

\begin{figure}[ht]
	\centering
	\subfigure[$K=3$.]{\label{subfig:comp_M_K3}\includegraphics[width=0.32\columnwidth]{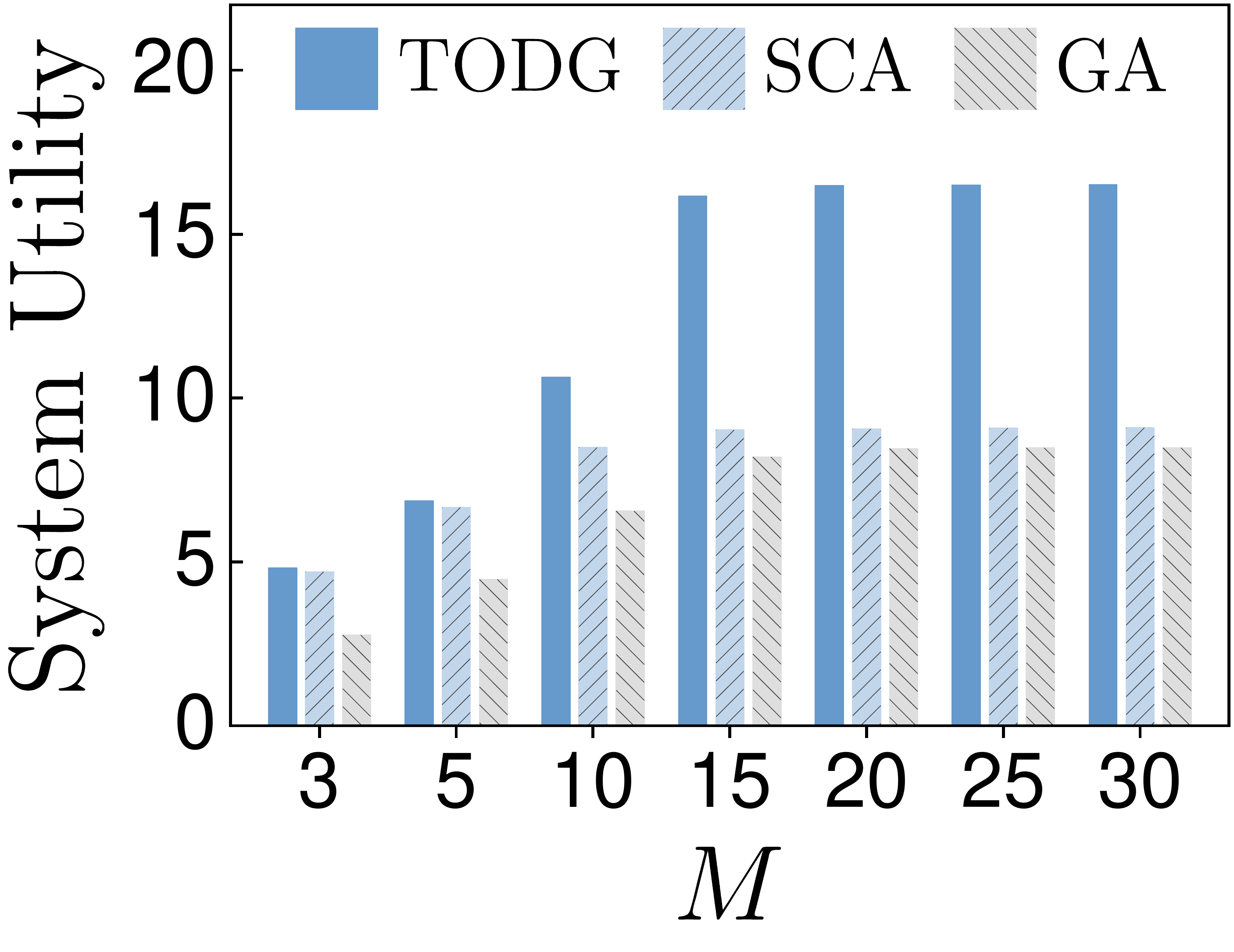}}
	\subfigure[$K=6$.]{\label{subfig:comp_M_K6}\includegraphics[width=0.32\columnwidth]{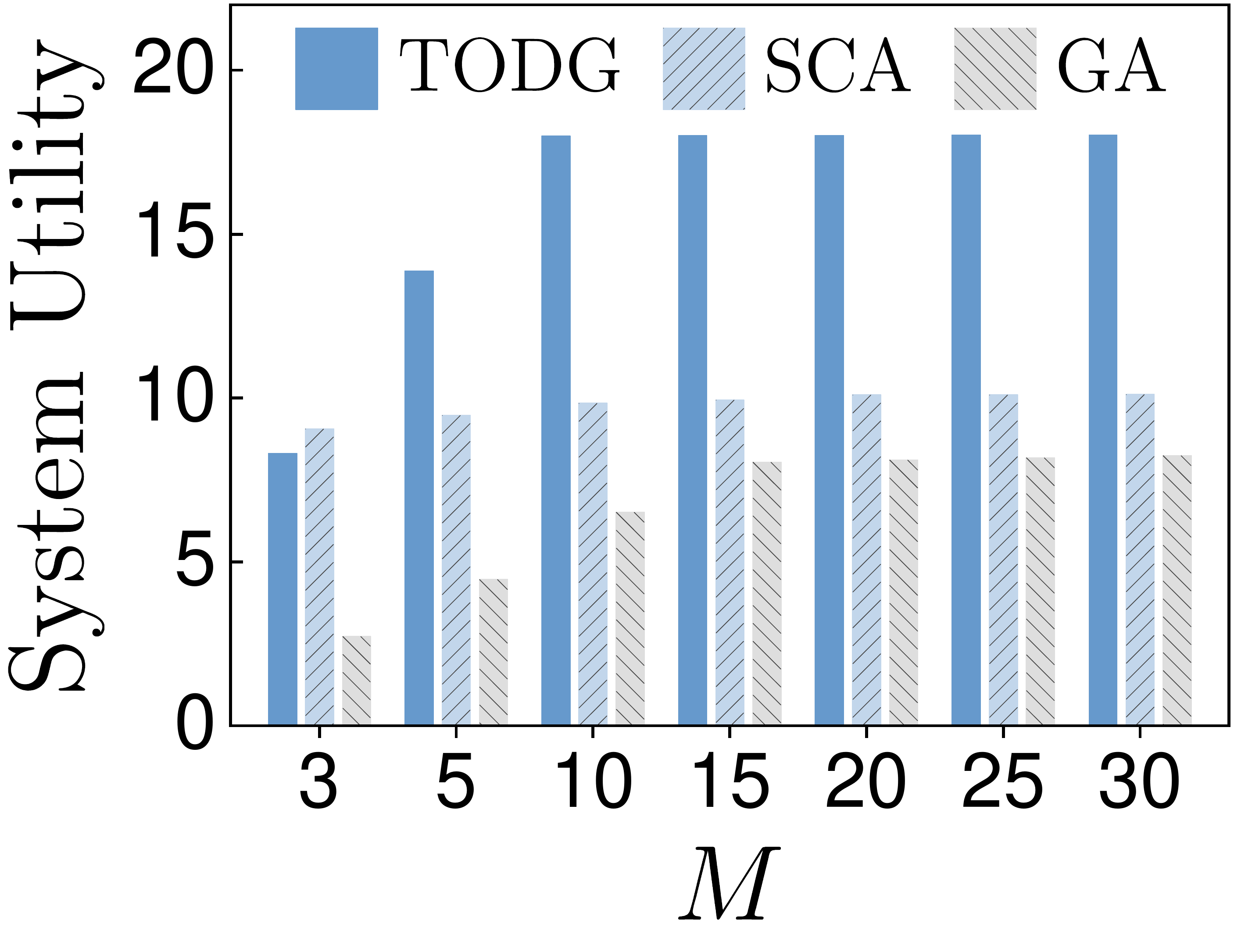}}
	\subfigure[$K=9$.]{\label{subfig:comp_M_K9}\includegraphics[width=0.32\columnwidth]{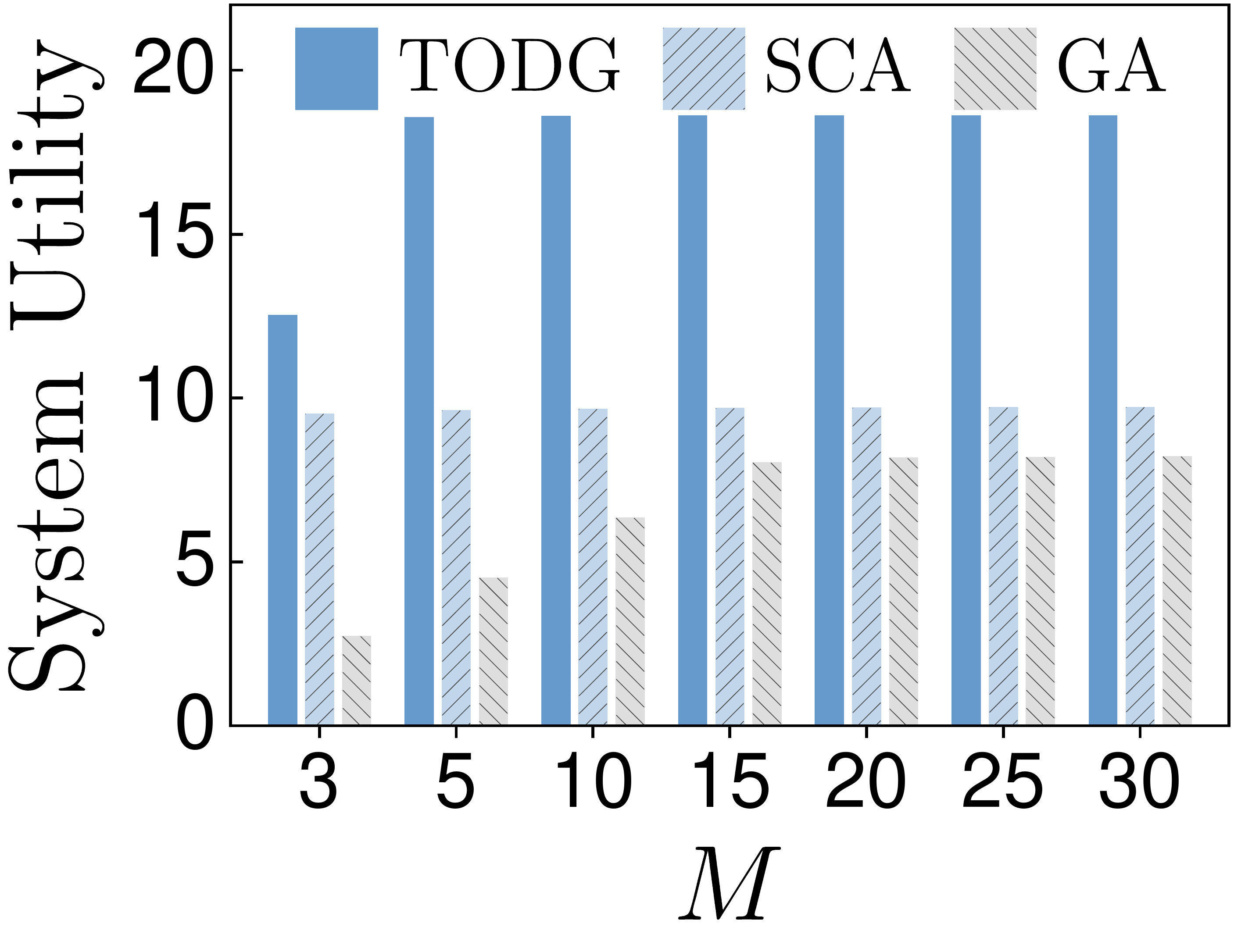}}
	\vspace{-0.5em}
	\caption{Comparison of utility under different network scales.}
	\label{fig:comp_M}
\end{figure}

\textbf{Scalability of TODG.} To evaluate the scalability of TODG, we vary the number of servers $M$ from 3 to 30 and task types $K$ from 3 to 9, while fixing the number of each type's devices as $N_k=10$ and channels as $L=20$. As Fig. \ref{fig:comp_M} shows, with the increase of $M$, TODG outperforms the baseline algorithms, and the gap becomes larger. It indicates that TODG can make full use of the insufficient communication bandwidth and limited computational resources, especially for large-scale networks. Besides, even with sufficient edge servers, GA and SCA hardly improve their performance, because it is bottlenecked by the limited channel capacities (see Fig. \ref{subfig:range_C} for more details). 

\begin{figure}[ht]
	\centering
	\subfigure[Impact on the performance.]{\label{subfig:self_M_K}\includegraphics[width=0.32\columnwidth]{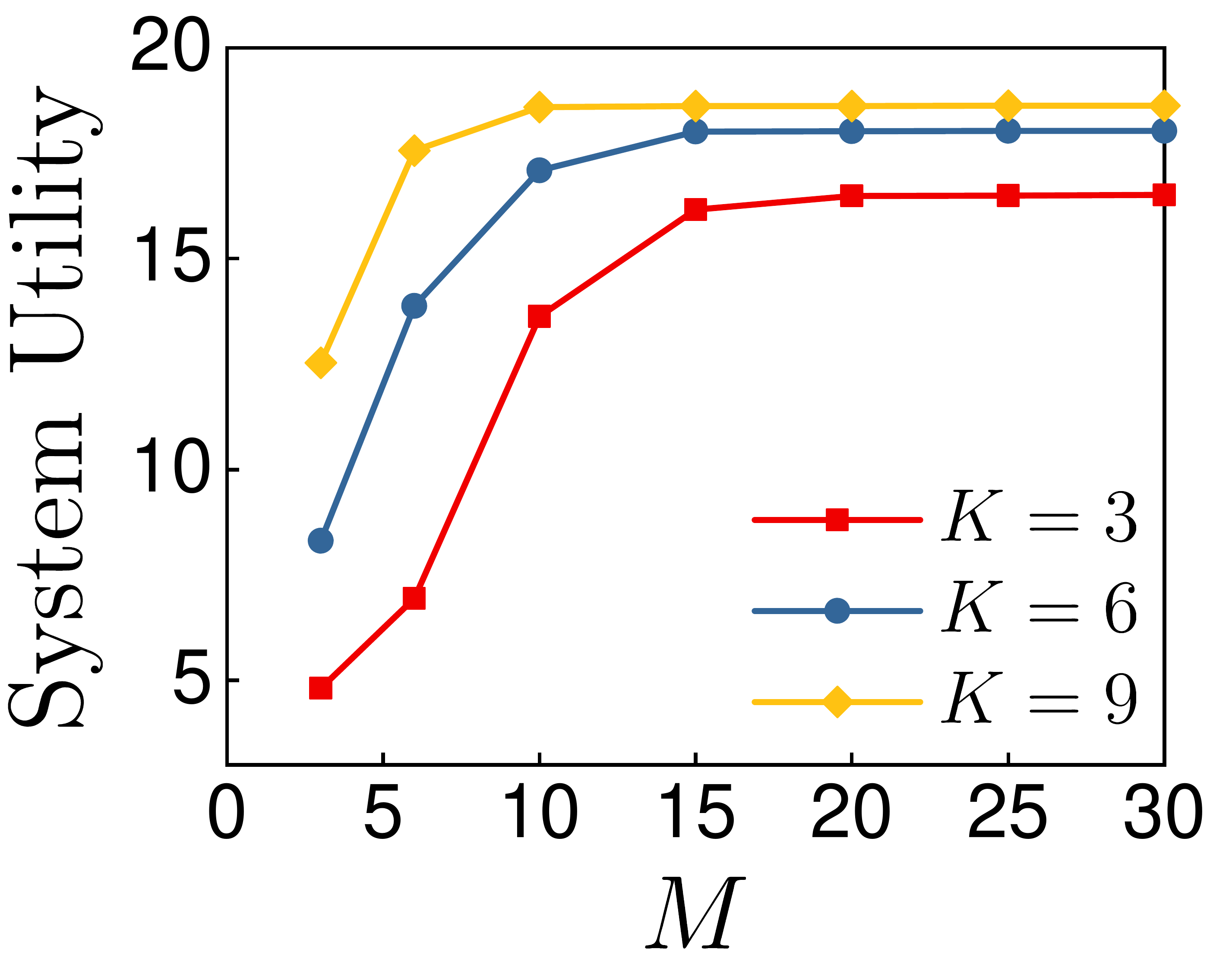}}
	\subfigure[Impact of the number of servers on the running time.]{\label{subfig:run_time_M}\includegraphics[width=0.32\columnwidth]{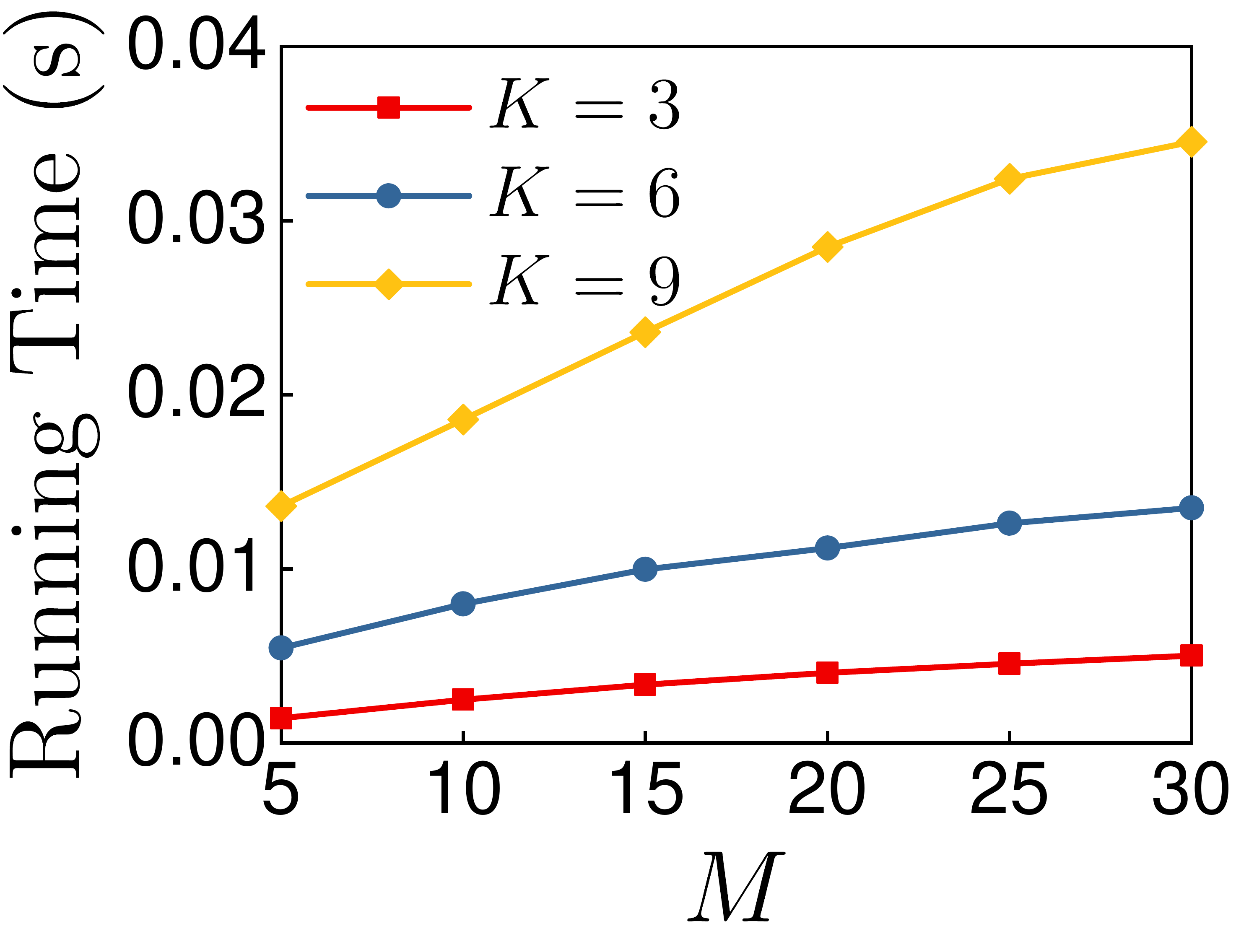}}
	\subfigure[Impact of the number of devices on the running time.]{\label{subfig:run_time_N}\includegraphics[width=0.32\columnwidth]{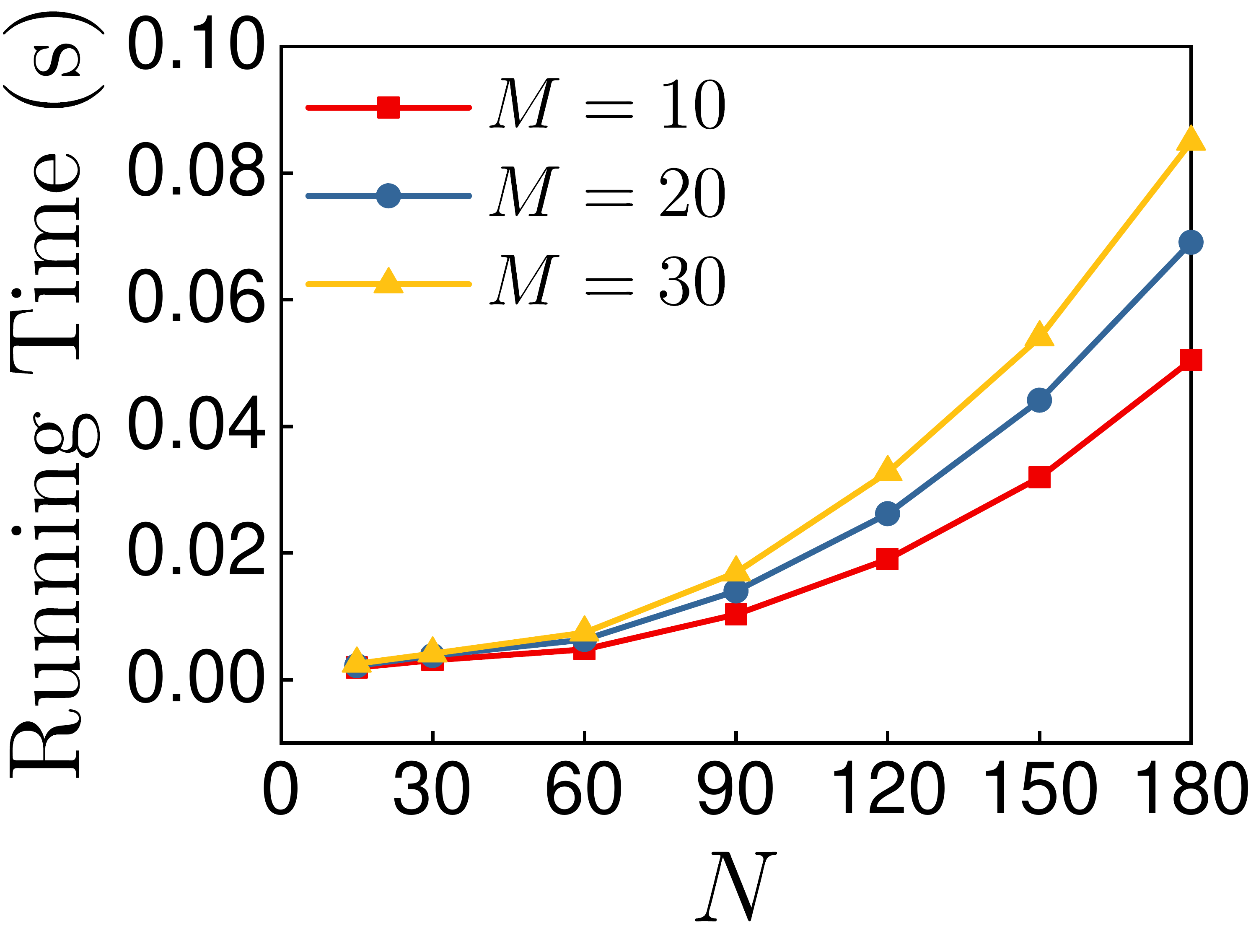}}
	\caption{Scalability of TODG.}
	\label{fig:scala}
\end{figure}

As illustrated in Fig. \ref{fig:scala}, the utility achieved by TODG first increases sharply with larger $M$, then the increasing rate slows down. The reason is that a certain number of edge servers suffice to provide computing resources for the local tasks. In addition, fixing $\delta=1$, we report TODG's average running time in each slot under different $N$, $K$ and $M$ in Fig. \ref{fig:scala}. Recall that $N=10K$ in this series of experiments. The results validate the theoretical computation complexity of TODG, which is provided in Sec. 4.2.3. Notably, due to the periodic strategy, we do not need to carry out channel assignment in each slot, and thus the running time can be significantly reduced.

\begin{figure}[ht]
	\centering
	\subfigure[Impact of the number of user devices.]{\label{subfig:range_N}\includegraphics[width=0.225\textwidth]{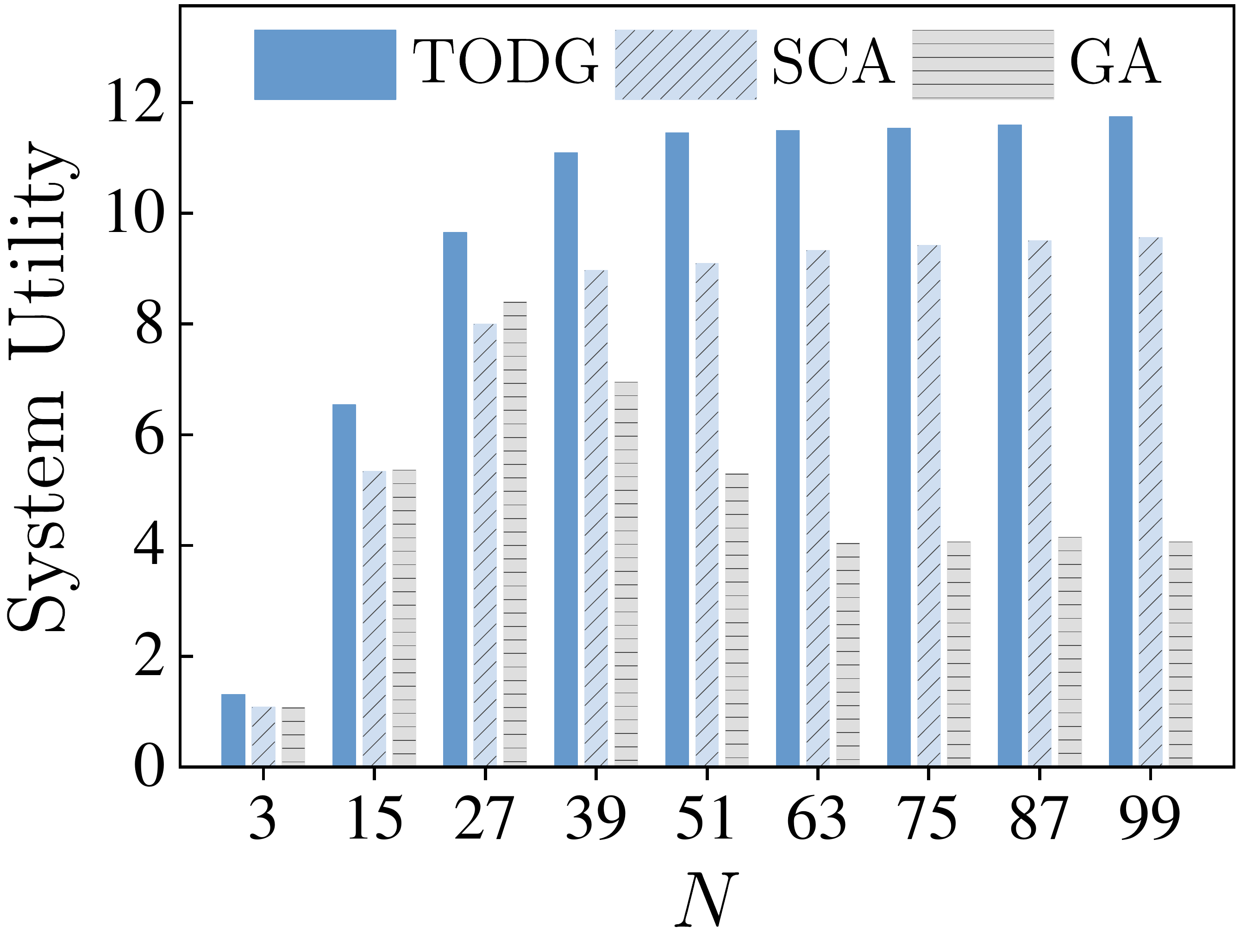}}
	\subfigure[Impact of the number of task types.]{\label{subfig:comp_K_M3}\includegraphics[width=0.225\textwidth]{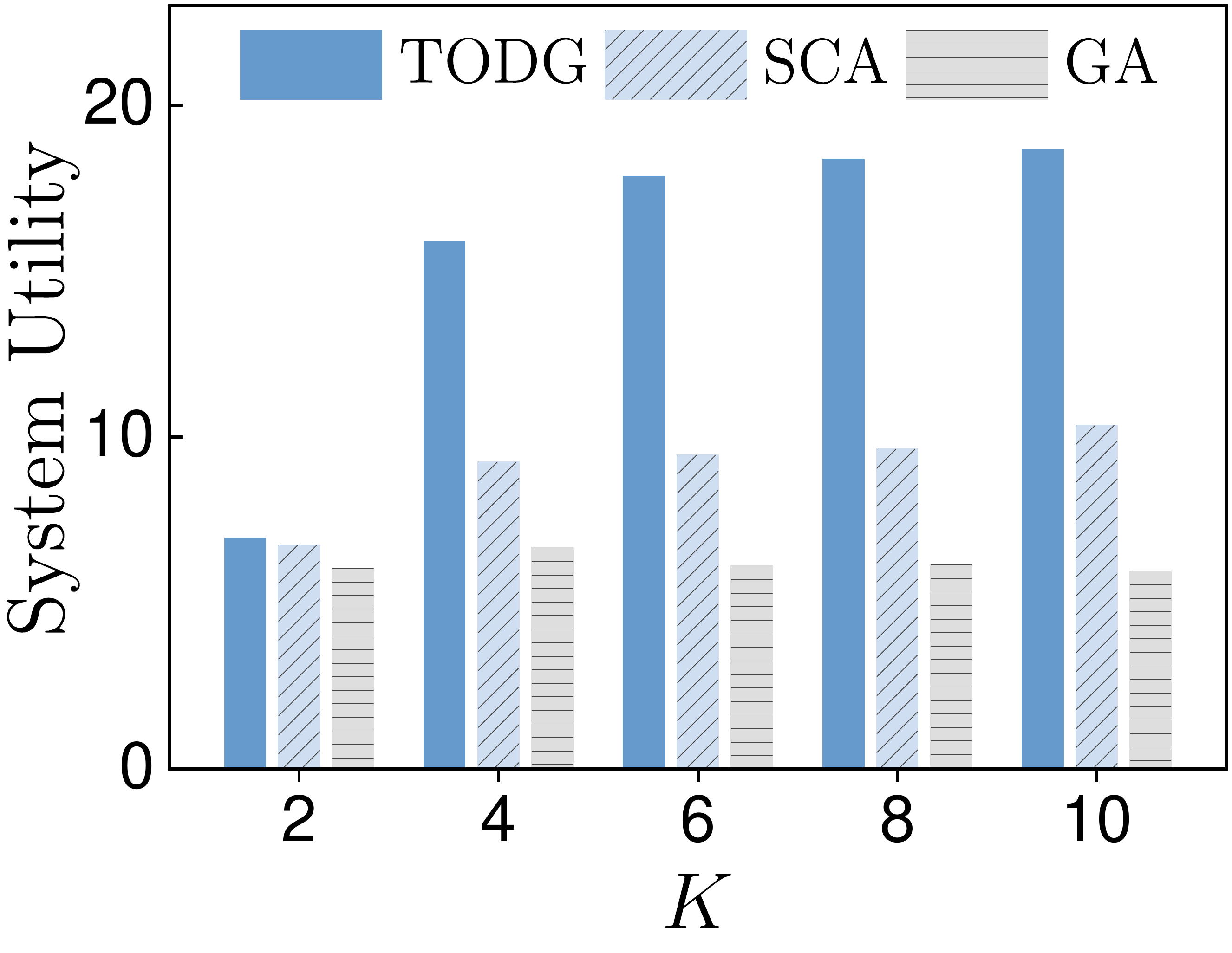}}
	
	\subfigure[Impact of channel capacities.]{\label{subfig:range_C}\includegraphics[width=0.225\textwidth]{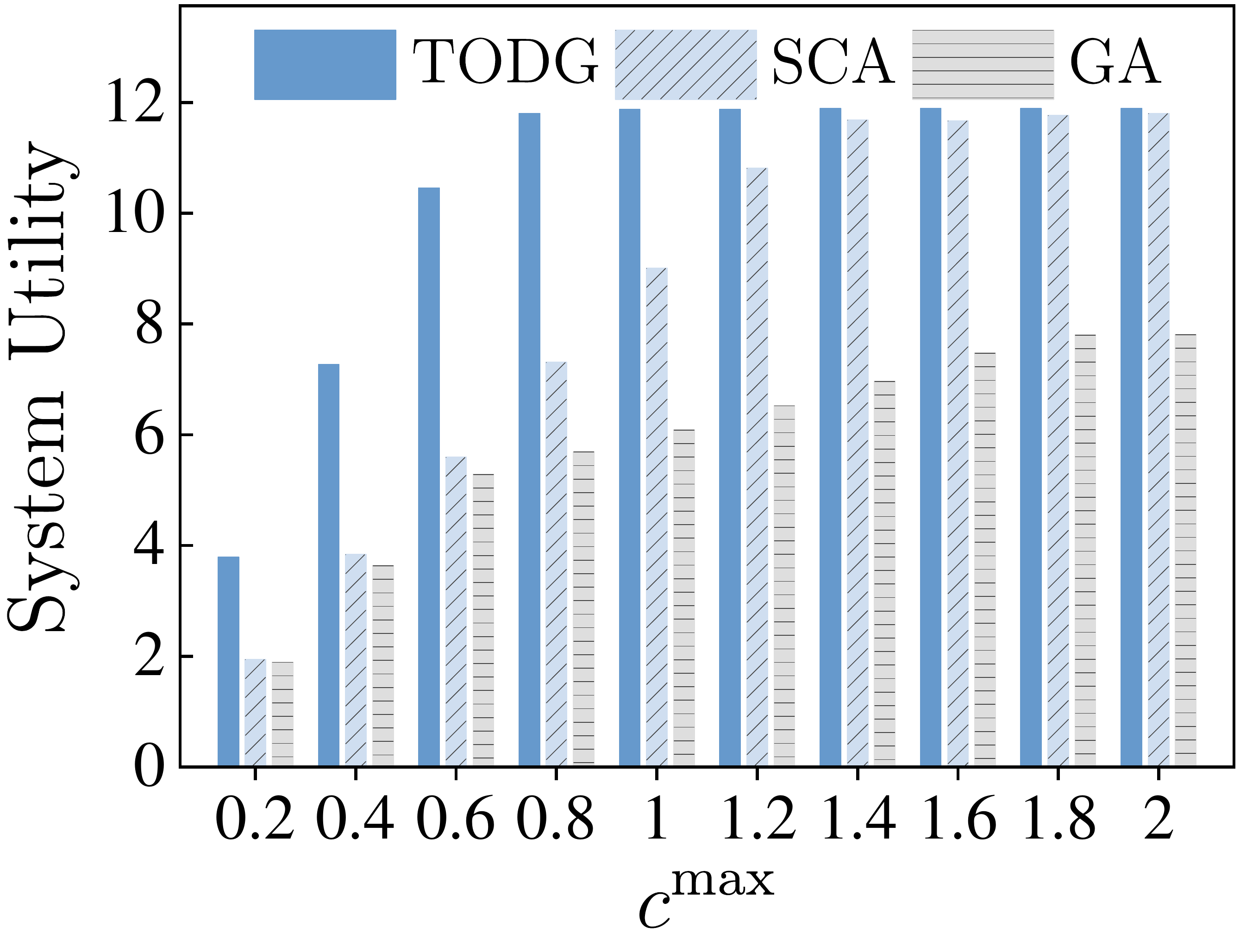}}
	\subfigure[Impact of service rates on edge servers.]{\label{subfig:range_U}\includegraphics[width=0.225\textwidth]{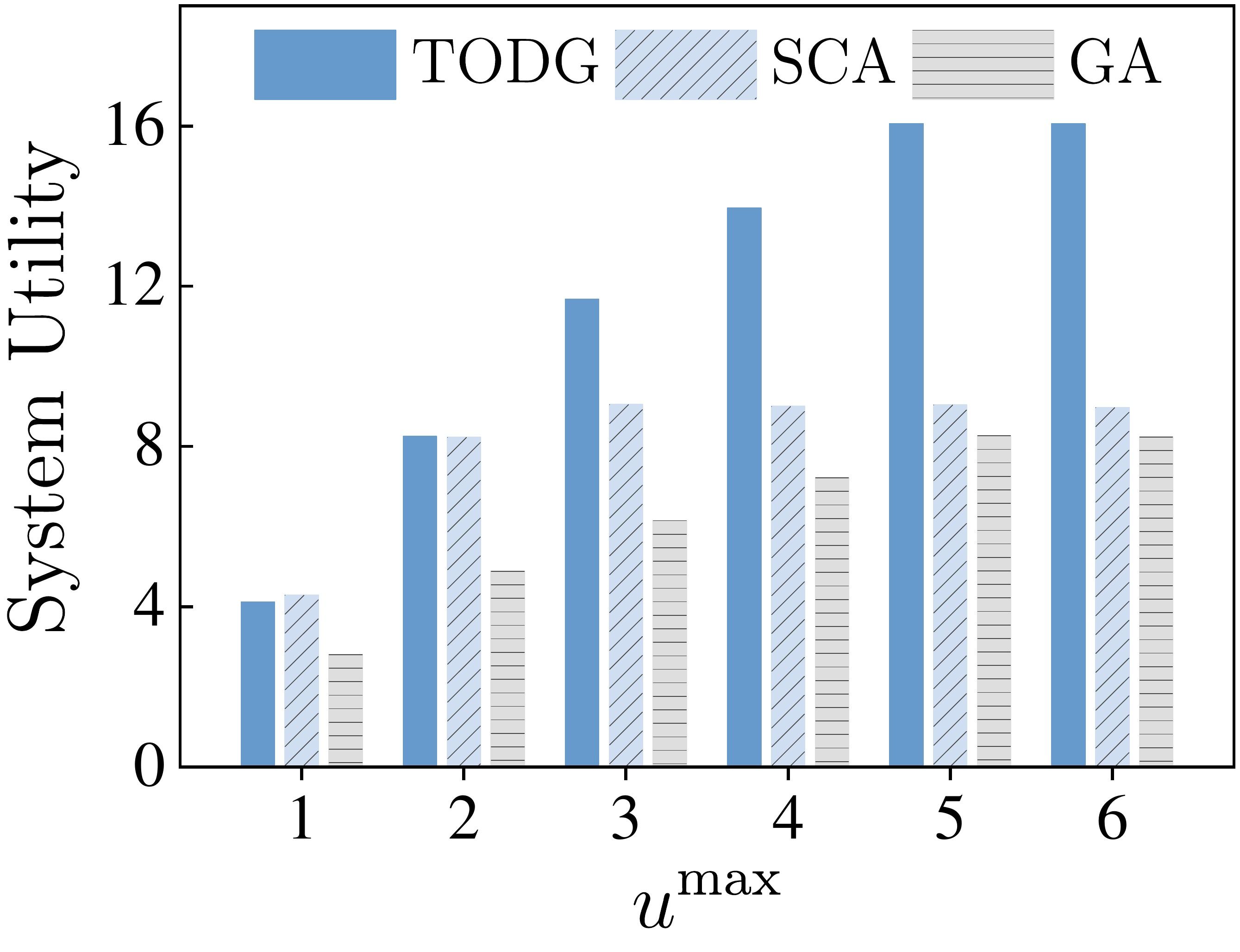}}
	\caption{Comparison of system utility among algorithms under different parameters.}
	\label{fig:comparason_paramter}
\end{figure}

\textbf{Impact of parameters.} We further study the impact of different parameters on the system utility. We first run the simulation experiments under varying numbers of user devices. As illustrated in Fig. \ref{subfig:range_N}, TODG achieves better performance with more user devices. On the contrary, since SCA and GA cannot exploit the limited and stochastic communication and computational resources, the contention among user devices would hinder the further improvement of system utility. Especially for GA, the contention even causes a performance decrease because it neglects the fairness among devices. Then, we vary the number of task types $K$ from 3 to 30, while fixing the number of each type's devices as $N_k=10$ and channels as $L=20$. For convenience, let the task arrival of each type and the processing rate of each VM follow $U(0,1)$ and $U(0,3)$, respectively. As Fig. \ref{subfig:comp_K_M3} shows, with the increase of $K$, TODG outperforms the baseline algorithms, and the gap becomes larger. After that, we evaluate the performance under different transmission rates and plot the results in Fig. \ref{subfig:range_C}. As shown in Fig. \ref{subfig:range_C}, TODG substantially outperforms the baseline algorithms in the cases of poor channel capacities, which indicates that TODG can fully utilize the limited communication resources. Besides, we vary the processing capabilities of edge servers to show the corresponding impact on the performance. It can be seen from Fig. \ref{subfig:range_U} that TODG can vastly improve the performance of SCA and GA, especially with powerful edge servers, which implies the importance of effective task scheduling and channel allocation. Due to the inefficient utilization of communication resources, a large number of tasks cannot be transmitted to edge servers timely. Thus, despite more powerful servers, the system utility of SCA and GA cannot be further improved. In contrast to SCA and GA, since TODG enables exploiting the communication resources, the bottleneck of TODG is the processing capabilities on edge servers. After eliminating this limitation, TODG shows its great advantages of effective task scheduling.  

\section{Conclusion}
\label{sec:conclusion}
In this paper, we have proposed a distributed online task offloading algorithm, called TODG, which jointly allocates resources and schedules the offloading tasks with delay guarantees while also achieving inexpensive computational cost. We further provide comprehensive theoretical insights into TODG and particularly show it can balance the near-optimal system utility and computational complexity. Extensive simulation results validate the effectiveness of TODG and demonstrate that TODG outperforms the baseline algorithms, especially in the cases with poor channel conditions. There are many interesting directions for future work. First, it is of interest to consider the task migration problem in high mobility scenarios into TODG. Secondly, our simulation results indicate that the optimality gap is much smaller than the theoretical bound. It is intriguing to get a more deep understanding of this phenomenon. Moreover, it remains largely open to incorporate the learning methods (\eg, online learning) into the task scheduling for edge computing.


%

%

%
%

\ifCLASSOPTIONcaptionsoff
\newpage
\fi



%
\bibliographystyle{IEEEtran}
\bibliography{reference}

\begin{thebibliography}{10}
\providecommand{\url}[1]{#1}
\csname url@samestyle\endcsname
\providecommand{\newblock}{\relax}
\providecommand{\bibinfo}[2]{#2}
\providecommand{\BIBentrySTDinterwordspacing}{\spaceskip=0pt\relax}
\providecommand{\BIBentryALTinterwordstretchfactor}{4}
\providecommand{\BIBentryALTinterwordspacing}{\spaceskip=\fontdimen2\font plus
\BIBentryALTinterwordstretchfactor\fontdimen3\font minus
  \fontdimen4\font\relax}
\providecommand{\BIBforeignlanguage}[2]{{%
\expandafter\ifx\csname l@#1\endcsname\relax
\typeout{** WARNING: IEEEtran.bst: No hyphenation pattern has been}%
\typeout{** loaded for the language `#1'. Using the pattern for}%
\typeout{** the default language instead.}%
\else
\language=\csname l@#1\endcsname
\fi
#2}}
\providecommand{\BIBdecl}{\relax}
\BIBdecl

\bibitem{ren2019survey}
J.~Ren, D.~Zhang, S.~He, Y.~Zhang, and T.~Li, ``A survey on end-edge-cloud
  orchestrated network computing paradigms: Transparent computing, mobile edge
  computing, fog computing, and cloudlet,'' \emph{ACM Comput. Surv.}, vol.~52,
  no.~6, pp. 1--36, 2019.

\bibitem{abbas2017mobile}
N.~Abbas, Y.~Zhang, A.~Taherkordi, and T.~Skeie, ``Mobile edge computing: A
  survey,'' \emph{IEEE Internet Things J.}, vol.~5, no.~1, pp. 450--465, 2017.

\bibitem{jovsilo2018decentralized}
S.~Jo{\v{s}}ilo and G.~D{\'a}n, ``Decentralized algorithm for randomized task
  allocation in fog computing systems,'' \emph{IEEE/ACM Trans. Networking},
  vol.~27, no.~1, pp. 85--97, 2018.

\bibitem{yi2019multi}
C.~Yi, J.~Cai, and Z.~Su, ``A multi-user mobile computation offloading and
  transmission scheduling mechanism for delay-sensitive applications,''
  \emph{IEEE Trans. Mob. Comput.}, vol.~19, no.~1, pp. 29--43, 2019.

\bibitem{hekmati2019optimal}
A.~Hekmati, P.~Teymoori, T.~D. Todd, D.~Zhao, and G.~Karakostas, ``Optimal
  mobile computation offloading with hard deadline constraints,'' \emph{IEEE
  Trans. Mob. Comput.}, vol.~19, no.~9, pp. 2160--2173, 2020.

\bibitem{chen2016joint}
M.-H. Chen, B.~Liang, and M.~Dong, ``Joint offloading decision and resource
  allocation for multi-user multi-task mobile cloud,'' in \emph{Proc. IEEE
  ICC}, 2016, pp. 1--6.

\bibitem{kamoun2015joint}
M.~Kamoun, W.~Labidi, and M.~Sarkiss, ``Joint resource allocation and
  offloading strategies in cloud enabled cellular networks,'' in \emph{Proc.
  IEEE ICC}, 2015, pp. 5529--5534.

\bibitem{labidi2015joint}
W.~Labidi, M.~Sarkiss, and M.~Kamoun, ``Joint multi-user resource scheduling
  and computation offloading in small cell networks,'' in \emph{Proc. IEEE
  WiMob}, 2015, pp. 794--801.

\bibitem{eshraghi2019joint}
N.~Eshraghi and B.~Liang, ``Joint offloading decision and resource allocation
  with uncertain task computing requirement,'' in \emph{Proc. IEEE INFOCOM},
  2019.

\bibitem{zhang2019near}
D.~Zhang, L.~Tan, J.~Ren, M.~K. Awad, S.~Zhang, Y.~Zhang, and P.-J. Wan,
  ``Near-optimal and truthful online auction for computation offloading in
  green edge-computing systems,'' \emph{IEEE Trans. Mob. Comput.}, vol.~19,
  no.~4, pp. 880--893, 2019.

\bibitem{xu2018offloading}
Z.~P. Xu~J, Chen~L, ``Joint service caching and task offloading for mobile edge
  computing in dense networks,'' in \emph{Proc. IEEE INFOCOM}, 2018.

\bibitem{zhu2019blot}
Z.~Zhu, T.~Liu, Y.~Yang, and X.~Luo, ``Blot: Bandit learning-based offloading
  of tasks in fog-enabled networks,'' \emph{IEEE Trans. Parallel Distrib.
  Syst.}, vol.~30, no.~12, pp. 2636--2649, 2019.

\bibitem{wang2020multi}
X.~Wang, Z.~Ning, and S.~Guo, ``Multi-agent imitation learning for pervasive
  edge computing: a decentralized computation offloading algorithm,''
  \emph{IEEE Trans. Parallel Distrib. Syst.}, vol.~32, no.~2, pp. 411--425,
  2020.

\bibitem{liu2020resource}
B.~Liu, C.~Liu, and M.~Peng, ``Resource allocation for energy-efficient mec in
  noma-enabled massive iot networks,'' \emph{IEEE J. Sel. Areas Commun.},
  vol.~39, no.~4, pp. 1015--1027, 2020.

\bibitem{mao2016dynamic}
Y.~Mao, J.~Zhang, and K.~B. Letaief, ``Dynamic computation offloading for
  mobile-edge computing with energy harvesting devices,'' \emph{IEEE J. Sel.
  Areas Commun.}, vol.~34, no.~12, pp. 3590--3605, 2016.

\bibitem{wang2020online}
K.~Wang, Y.~Zhou, Z.~Liu, Z.~Shao, X.~Luo, and Y.~Yang, ``Online task
  scheduling and resource allocation for intelligent noma-based industrial
  internet of things,'' \emph{IEEE J. Sel. Areas Commun.}, vol.~38, no.~5, pp.
  803--815, 2020.

\bibitem{chen2018computation}
L.~Chen, S.~Zhou, and J.~Xu, ``Computation peer offloading for
  energy-constrained mobile edge computing in small-cell networks,''
  \emph{IEEE/ACM Trans. Networking}, vol.~26, no.~4, pp. 1619--1632, 2018.

\bibitem{jovsilo2020computation}
S.~Jo{\v{s}}ilo and G.~D{\'a}n, ``Computation offloading scheduling for
  periodic tasks in mobile edge computing,'' \emph{IEEE/ACM Trans. Networking},
  vol.~28, no.~2, pp. 667--680, 2020.

\bibitem{tang2020deep}
M.~Tang and V.~W. Wong, ``Deep reinforcement learning for task offloading in
  mobile edge computing systems,'' \emph{arXiv preprint arXiv:2005.02459},
  2020.

\bibitem{chen2019collaborative}
L.~Chen, C.~Shen, P.~Zhou, and J.~Xu, ``Collaborative service placement for
  edge computing in dense small cell networks,'' \emph{IEEE Trans. Mob.
  Comput.}, vol.~20, no.~2, pp. 377--390, 2019.

\bibitem{liu2020latency}
T.~Liu, L.~Fang, Y.~Zhu, W.~Tong, and Y.~Yang, ``Latency-minimized and
  energy-efficient online task offloading for mobile edge computing with
  stochastic heterogeneous tasks,'' in \emph{Proc. IEEE ICPADS}, 2019, pp.
  376--383.

\bibitem{you2016energy}
C.~You, K.~Huang, H.~Chae, and B.-H. Kim, ``Energy-efficient resource
  allocation for mobile-edge computation offloading,'' \emph{IEEE Trans.
  Wireless Commun.}, vol.~16, no.~3, pp. 1397--1411, 2016.

\bibitem{kao2014optimizing}
Y.-H. Kao and B.~Krishnamachari, ``Optimizing mobile computational offloading
  with delay constraints,'' in \emph{Proc. IEEE GLOBECOM}, 2014.

\bibitem{liu2016delay}
J.~Liu, Y.~Mao, J.~Zhang, and K.~B. Letaief, ``Delay-optimal computation task
  scheduling for mobile-edge computing systems,'' in \emph{Proc. IEEE ISIT},
  2016, pp. 1451--1455.

\bibitem{lyu2017optimal}
X.~Lyu, W.~Ni, H.~Tian, R.~P. Liu, X.~Wang, G.~B. Giannakis, and A.~Paulraj,
  ``Optimal schedule of mobile edge computing for internet of things using
  partial information,'' \emph{IEEE J. Sel. Areas Commun.}, vol.~35, no.~11,
  pp. 2606--2615, 2017.

\bibitem{mao2017stochastic}
Y.~Mao, J.~Zhang, S.~Song, and K.~B. Letaief, ``Stochastic joint radio and
  computational resource management for multi-user mobile-edge computing
  systems,'' \emph{IEEE Trans. Wireless Commun.}, vol.~16, no.~9, pp.
  5994--6009, 2017.

\bibitem{mao2017joint}
Y.~Mao, J.~Zhang, and K.~B. Letaief, ``Joint task offloading scheduling and
  transmit power allocation for mobile-edge computing systems,'' in \emph{Proc.
  IEEE WCNC}, 2017.

\bibitem{zhang2017optimal}
K.~Zhang, Y.~Mao, S.~Leng, S.~Maharjan, and Y.~Zhang, ``Optimal delay
  constrained offloading for vehicular edge computing networks,'' in
  \emph{Proc. IEEE ICC}, 2017.

\bibitem{chen2017joint}
M.-H. Chen, B.~Liang, and M.~Dong, ``Joint offloading and resource allocation
  for computation and communication in mobile cloud with computing access
  point,'' in \emph{Proc. IEEE INFOCOM}, 2017.

\bibitem{ren2018latency}
J.~Ren, G.~Yu, Y.~Cai, and Y.~He, ``Latency optimization for resource
  allocation in mobile-edge computation offloading,'' \emph{IEEE Trans.
  Wireless Commun.}, vol.~17, no.~8, pp. 5506--5519, 2018.

\bibitem{zhang2017energy}
J.~Zhang, X.~Hu, Z.~Ning, E.~C.-H. Ngai, L.~Zhou, J.~Wei, J.~Cheng, and B.~Hu,
  ``Energy-latency tradeoff for energy-aware offloading in mobile edge
  computing networks,'' \emph{IEEE Internet Things J.}, vol.~5, no.~4, pp.
  2633--2645, 2017.

\bibitem{zhou2018computation}
F.~Zhou, Y.~Wu, R.~Q. Hu, and Y.~Qian, ``Computation rate maximization in
  uav-enabled wireless-powered mobile-edge computing systems,'' \emph{IEEE J.
  Sel. Areas Commun.}, vol.~36, no.~9, pp. 1927--1941, 2018.

\bibitem{alameddine2019dynamic}
H.~A. Alameddine, S.~Sharafeddine, S.~Sebbah, S.~Ayoubi, and C.~Assi, ``Dynamic
  task offloading and scheduling for low-latency iot services in multi-access
  edge computing,'' \emph{IEEE J. Sel. Areas Commun.}, vol.~37, no.~3, pp.
  668--682, 2019.

\bibitem{chen2019delay}
S.~Chen, Y.~Zheng, K.~Wang, and W.~Lu, ``Delay guaranteed energy-efficient
  computation offloading for industrial iot in fog computing,'' in \emph{Proc.
  IEEE ICC}, 2019.

\bibitem{maswood2020novel}
M.~M.~S. Maswood, M.~R. Rahman, A.~G. Alharbi, and D.~Medhi, ``A novel strategy
  to achieve bandwidth cost reduction and load balancing in a cooperative
  three-layer fog-cloud computing environment,'' \emph{IEEE Access}, vol.~8,
  pp. 113\,737--113\,750, 2020.

\bibitem{liang2019multiuser}
Z.~Liang, Y.~Liu, T.-M. Lok, and K.~Huang, ``Multiuser computation offloading
  and downloading for edge computing with virtualization,'' \emph{IEEE Trans.
  Wireless Commun.}, vol.~18, no.~9, pp. 4298--4311, 2019.

\bibitem{li2020qos}
Q.~Li, S.~Wang, A.~Zhou, X.~Ma, A.~X. Liu \emph{et~al.}, ``Qos driven task
  offloading with statistical guarantee in mobile edge computing,'' \emph{IEEE
  Trans. Mob. Comput.}, 2020.

\bibitem{nath2020multi}
S.~Nath, Y.~Li, J.~Wu, and P.~Fan, ``Multi-user multi-channel computation
  offloading and resource allocation for mobile edge computing,'' in
  \emph{Proc. IEEE ICC}, 2020.

\bibitem{li2021task}
S.~Li, C.~Li, Y.~Huang, B.~A. Jalaian, Y.~T. Hou, and W.~Lou, ``Task offloading
  with uncertain processing cycles,'' in \emph{Proc. ACM MobiHoc}, 2021, pp.
  51--60.

\bibitem{chen2014decentralized}
X.~Chen, ``Decentralized computation offloading game for mobile cloud
  computing,'' \emph{IEEE Trans. Parallel Distrib. Syst.}, vol.~26, no.~4, pp.
  974--983, 2014.

\bibitem{dinh2017offloading}
T.~Q. Dinh, J.~Tang, Q.~D. La, and T.~Q. Quek, ``Offloading in mobile edge
  computing: Task allocation and computational frequency scaling,'' \emph{IEEE
  Trans. Commun.}, vol.~65, no.~8, pp. 3571--3584, 2017.

\bibitem{ren2020joint}
J.~Ren, K.~M. Mahfujul, F.~Lyu, S.~Yue, and Y.~Zhang, ``Joint channel
  allocation and resource management for stochastic computation offloading in
  mec,'' \emph{IEEE Trans. Veh. Technol.}, vol.~69, no.~8, pp. 8900--8913,
  2020.

\bibitem{wang2017computation}
C.~Wang, C.~Liang, F.~R. Yu, Q.~Chen, and L.~Tang, ``Computation offloading and
  resource allocation in wireless cellular networks with mobile edge
  computing,'' \emph{IEEE Trans. Wireless Commun.}, vol.~16, no.~8, pp.
  4924--4938, 2017.

\bibitem{fang2016stochastic}
W.~Fang, X.~Yao, X.~Zhao, J.~Yin, and N.~Xiong, ``A stochastic control approach
  to maximize profit on service provisioning for mobile cloudlet platforms,''
  \emph{IEEE Trans. Syst. Man Cybern.: Syst.}, vol.~48, no.~4, pp. 522--534,
  2016.

\bibitem{neely2011opportunistic}
M.~J. Neely, ``Opportunistic scheduling with worst case delay guarantees in
  single and multi-hop networks,'' in \emph{Proc. IEEE INFOCOM}, 2011, pp.
  1728--1736.

\bibitem{neely2010stochastic}
------, ``Stochastic network optimization with application to communication and
  queueing systems,'' \emph{Synth. Lect. Commun. Networks}, vol.~3, no.~1, pp.
  1--211, 2010.

\bibitem{neely2006energy}
------, ``Energy optimal control for time-varying wireless networks,''
  \emph{IEEE Trans. Inf. Theory}, vol.~52, no.~7, pp. 2915--2934, 2006.

\bibitem{kann1991maximum}
V.~Kann, ``Maximum bounded 3-dimensional matching is max snp-complete,''
  \emph{Inf. Process. Lett.}, vol.~37, no.~1, pp. 27--35, 1991.

\bibitem{kushagra2020three}
S.~Kushagra, ``Three-dimensional matching is np-hard,'' \emph{arXiv preprint
  arXiv:2003.00336}, 2020.

\bibitem{bourgeois1971extension}
F.~Bourgeois and J.-C. Lassalle, ``An extension of the munkres algorithm for
  the assignment problem to rectangular matrices,'' \emph{Commun. ACM},
  vol.~14, no.~12, pp. 802--804, 1971.

\bibitem{chopra2017distributed}
S.~Chopra, G.~Notarstefano, M.~Rice, and M.~Egerstedt, ``A distributed version
  of the hungarian method for multirobot assignment,'' \emph{IEEE Trans. Rob.},
  vol.~33, no.~4, pp. 932--947, 2017.

\end{thebibliography}

\end{document}